\documentclass{sig-alternate-05-2015}

\usepackage{url}
\usepackage{amsmath}

\newtheorem{theorem}{Theorem}

\usepackage{changepage}
\usepackage{algorithmic}
\usepackage[boxed]{algorithm2e}
\usepackage{multirow}
\usepackage{amsmath}
\usepackage{xcolor}
\usepackage{color}
\usepackage{times}
\usepackage{mathptmx}
\usepackage{mathtools}
\usepackage{mathrsfs}
\usepackage{enumerate}
\usepackage{amssymb}
\usepackage{subfigure, float}

\DeclareMathAlphabet{\mathcal}{OMS}{cmsy}{m}{n}

\newcommand{\nop}[1]{}

\newtheorem{example}{Example}

\newtheorem{lemma}{Lemma}
\newtheorem{corollary}{Corollary}
\newtheorem{fact}{Fact}

\begin{document}
\begin{sloppy}
%
\setcopyright{acmcopyright}
\title{Tracking Influential Nodes in Dynamic Networks}

\numberofauthors{1} 
\author{
\alignauthor
Yu Yang$^{\dag}$, Zhefeng Wang$^{\ddag}$, Jian Pei$^{\dag}$ and Enhong Chen$^{\ddag}$\\
       \affaddr{$^\dag$Simon Fraser University, Burnaby, Canada}\\
       \affaddr{$^\ddag$University of Science and Technology of China,Hefei, China}\\
       \email{yya119@sfu.ca, zhefwang@mail.ustc.edu.cn, jpei@cs.sfu.ca, cheneh@ustc.edu.cn}       
}

\maketitle
\begin{abstract}
In this paper, we tackle a challenging problem inherent in a series of applications: tracking the influential nodes in dynamic networks.  Specifically, we model a dynamic network as a stream of edge weight updates. This general model embraces many practical scenarios as special cases, such as edge and node insertions, deletions as well as evolving weighted graphs.  Under the popularly adopted linear threshold model and independent cascade model, we consider two essential versions of the problem: finding the nodes whose influences passing a user specified threshold and finding the top-$k$ most influential nodes.  Our key idea is to use the polling-based methods and maintain a sample of random RR sets so that we can approximate the influence of nodes with provable quality guarantees.  We develop an efficient algorithm that incrementally updates the sample random RR sets against network changes. We also design methods to determine the proper sample sizes for the two versions of the problem so that we can provide strong quality guarantees and, at the same time, be efficient in both space and time.  In addition to the thorough theoretical results, our experimental results on $5$ real network data sets clearly demonstrate the effectiveness and efficiency of our algorithms.
\end{abstract}

\terms{Algorithms}

\section{Introduction}\label{sec:intro}

More and more applications are built on dynamic networks and need to track influential nodes.  For example, consider cold-start recommendation in a dynamic social network -- we want to recommend to a new comer some existing users in a social network.  A new user may want to subscribe to the posts from some users in order to obtain hot posts (posts that are widely spread in the social network) at the earliest time. Clearly for such a new user we should recommend her some influential users in the current network. Traditional Influence Maximization cannot find those influential users we want here because it is for marketing in which all seed users have to be synchronized to spread the same content, while in reality online influential individuals often produce and spread their own contents in an asynchronized manner. The influential users we want are those who have high individual influence. 

More often than not, the underlying network is highly dynamic, where each node is a user and an edge captures the interaction from a user to another.  User interactions evolve continuously over time.  In an active social network, such as Twitter, Facebook, LinkedIn, Tencent WeChat, and Sina Weibo, the evolving dynamics, such as rich user interactions over time, is the most important value. It is critical to capture the most influential users in an online manner.  To address the needs, we have to tackle two challenges at the same time, influence computation and dynamics in networks.

Influence computation is very costly, technically \#P-hard under most influence models. Most existing studies have to compromise and consider the influence maximization problem only on a static network. Here, influence maximization in a network is to find a set of vertices $S$ such that the combined influence of the nodes in the set is maximized and $S$ satisfies some constraints such as the size of $S$ is within a budget.  The incapability of handling dynamics in large evolving networks seriously deprives many opportunities and potentials in applications.  Also note that influence maximization is very different from finding influential individuals, for the reason that the best $k$-vertices set $S$ does not consist of the $k$ most influential individual nodes because influence spreads of different individuals may overlap.

Although influence maximization and finding influential nodes are highly related since they both need to compute influence in one way or another, these two problems serve very different application scenarios and face different technical challenges.  For example, influence maximization is a core technique in viral marketing~\cite{domingos2001mining}. At the same time, influence maximization is not useful in the cold-start recommendation scenario discussed above, since a user is interested in being connected with individual users of great potential influence and may follow them in interaction.

To the best of our knowledge, our study is the first to tackle the problem of tracking influential nodes in dynamic networks. Please note that finding influential nodes is different from influence maximization.  Specifically, we model a dynamic network as a stream of edge weight updates. Our model is general and embraces many practical scenarios as special cases.  Under the popularly adopted linear threshold model and independent cascade model, we consider two essential versions of the problem: (1) finding the nodes whose influences passing a user specified threshold; and (2) finding the top-$k$ most influential nodes.  Our key idea is to use the polling-based methods and maintain a sample of random RR sets so that we can approximate the influence of nodes with provable quality guarantees.  

Recently, there is encouraging progress in influence maximization on dynamic networks~\cite{chen2015influential,aggarwal2012influential,ohsaka2016dynamic}.  Due to the difference between influence maximization and finding influential nodes, the methods in those studies~\cite{chen2015influential,aggarwal2012influential,ohsaka2016dynamic} cannot be applied directly to find influential nodes.  Moreover, in terms of specific techniques, our study is also very different from~\cite{chen2015influential,aggarwal2012influential}.  Most importantly, the methods in~\cite{chen2015influential,aggarwal2012influential} are heuristic, and do not provide any provable quality guarantee. Although authors of~\cite{ohsaka2016dynamic} claim that the algorithm in~\cite{ohsaka2016dynamic} has theoretical guarantees, in experiments reported, a key parameter is empirically set and makes the error rate $\epsilon$ even greater than 1. The reason that the algorithm in~\cite{ohsaka2016dynamic} cannot be implemented with small error rate is that the constant factor in its complexity is too large to be practical in use.  In addition, the influence model considered in~\cite{chen2015influential,ohsaka2016dynamic} is the Independent Cascade model. The one in~\cite{aggarwal2012influential} is a non-linear system.  We address both the Linear Threshold model and the Independent Cascade model in this study. To the best of our knowledge, we are the first to tackle influence computation with provable quality guarantee and report experiment results where algorithms are implemented strictly to fulfill the theoretical guarantee under the two most widely adopted influence models on dynamic networks.

To tackle the novel and challenging problem of finding influential nodes in dynamic networks, we make several technical contributions.  We develop an efficient algorithm that incrementally updates the sample random RR sets against network changes.  We also design methods to determine the proper sample sizes for the two versions of the problem so that we can provide strong quality guarantees and at the same time be efficient in both space and time.  In addition to the thorough theoretical results, our experimental results on $5$ real data sets clearly demonstrate the effectiveness and efficiency of our algorithms.  The largest data set used contains over $41$ million nodes, $1.5$ billion edges and $0.3$ billion edge updates.

The rest of the paper is organized as follows. We review the related work in Section~\ref{sec:related}. In Section~\ref{sec:pre}, we recall the Linear Threshold model and the Independent cascade model, review the polling-based method for computing influence spread, and formulate influence in dynamic networks. In Section~\ref{sec:update}, we present methods updating random RR sets over a stream of edge weight updates. In Section~\ref{sec:threshold}, we tackle the problem of tracking nodes whose influence spreads pass a user-defined threshold. In Section~\ref{sec:topk}, the problem of finding the top-$k$ influential nodes is settled. We report the experimental results in Section~\ref{sec:exp}. We conclude the paper in Section~\ref{sec:con}.

\section{Related Work}\label{sec:related}
Domingos~\textit{et~al.}~\cite{domingos2001mining} proposed to take advantage of peer influence between users in social networks for marketing. Kempe~\textit{et~al.}~\cite{kempe2003maximizing} formulated the problem using two discrete influence models, namely Independent Cascade model and Linear Threshold model. Since then, influence computation, especially influence maximization, has drawn much attention from both academia and industry~\cite{chen2013information, du2013scalable, borgs2014maximizing, tang2015influence,  chen2010scalableLT, goyal2011simpath, lucier2015influence, rossi2015spread}. Some heuristic methods were designed for computing influence spread under the Linear Threshold model~\cite{goyal2011simpath, chen2010scalableLT}. For the Independent Cascade Model, ~\cite{cohen2014sketch, lucier2015influence} proposed approximations of influence spread estimations. Note that there are still gaps between estimations of influence spread and real influence spreads, which were not clearly quantified in~\cite{cohen2014sketch, lucier2015influence}. Consequently, both~\cite{cohen2014sketch} and~\cite{lucier2015influence} cannot compute influence spread with provable quality guarantees. Recently, a polling-based method~\cite{borgs2014maximizing, tang2015influence, tang2014influence} was proposed for influence maximization under general triggering models. The key idea is to use some ``Reversely Reachable'' (RR) sets~\cite{tang2015influence, tang2014influence} to approximate the real influence spread of nodes. The error of approximation can be bounded with a high probability if the number of RR sets is large enough.

Extracting influential nodes in social networks is also an important problem in social network analysis and has been extensively investigated~\cite{goyal2008discovering, agarwal2008identifying, weng2010twitterrank, cha2010measuring}. In addition to the marketing value, influential individuals are also useful in recommender systems in online web service~\cite{agarwal2008identifying, weng2010twitterrank}. Due to the computational hardness of influence spread~\cite{chen2010scalableLT, chen2010scalable}, most methods did not use influence models to measure a user's influence, but adopted measures like PageRank which can be efficiently computed.

In a few applications, the underlying networks are evolving all the time~\cite{leskovec2008microscopic, leskovec2005graphs}. Rather than re-computing from scratch, incremental algorithms are more desirable in graph analysis tasks on dynamic networks. Maintaining PageRank values of nodes on an evolving graph was studied in~\cite{bahmani2010fast, ohsaka2015efficient}. Hayashi~\textit{et~al.}~\cite{hayashi2015fully} proposed to utilize a sketch of all shortest paths to dynamically maintain the edge betweenness value. The dynamics considered by the above work is a stream of edge insertions/deletions, which is not suitable for influence computation. The dynamics of influence network is more complicated, because besides edge insertions/deletions, influence probabilities of edges may also evolve over time~\cite{lei2015online}.

Aggarwal~\textit{et~al.}~\cite{aggarwal2012influential} explored how to find a set of nodes that has the highest influence within a time window $[t_0,t_0+h]$. They modeled influence propagation as a non-linear system which is very different from triggering models like the Linear Threshold model or the Independent Cascade model. The algorithm in~\cite{aggarwal2012influential} is heuristic and the results produced do not come with any provable quality guarantee. 

Chen~\textit{et~al.}~\cite{chen2015influential} investigated incrementally updating the seed set for influence maximization under the Independent Cascade model. They proposed an algorithm which utilizes the seed set mined from the former network snapshot to efficiently find the seed set of the current snapshot. An Upper Bound Interchange heuristic is applied in the algorithm. However, the algorithm in~\cite{chen2015influential} is costly in processing updates, since updating the Upper Bound vector for filtering non-influential nodes takes $O(m)$ time where $m$ is the number of edges. Moreover, the SP1M heuristic~\cite{kimura2006tractable}, which does not have any approximation quality guarantee, was adopted in~\cite{chen2015influential} for estimating influence spread of nodes. Thus, the set of influential nodes, even when the size of the seed set is set to $1$, does not have any provable quality guarantee.  

Independently and simultaneously\footnote{Early versions of our paper can be found at \url{https://arxiv.org/abs/1602.04490}} Ohsaka~\textit{et~al.}~\cite{ohsaka2016dynamic} studied a related problem, maintaining a number of RR sets over a stream of network updates under the IC model such that $(1-1/e-\epsilon)$-approximation influence maximization queries can be achieved with probability at least $1-\frac{1}{n}$. Our work is different from~\cite{ohsaka2016dynamic} in the following aspects. First, the problems are different. The problem tackled in~\cite{ohsaka2016dynamic}  is influence maximization, while our problem is tracking influential individuals. Second, ~\cite{ohsaka2016dynamic} only studied the IC model while in our work we addressed both the IC and the LT models. Moreover, our algorithm is theoretically sound and was strictly implemented to fulfill the theoretical guarantee in experiments, while it is not the case in~\cite{ohsaka2016dynamic}. To enable theoretical guarantees for the algorithm in~\cite{ohsaka2016dynamic}, one has to collect enough RR sets until the cost of all RR sets (i.e., the number of edges traversed when generating those RR sets) is $\Theta(\frac{(m+n)\log{n}}{\epsilon^3})$, which is a very large number in practice. Thus, in the experiments reported in~\cite{ohsaka2016dynamic}, the demanded cost is empirically set to $32(m+n)\log{n}$, which means $\epsilon$ is even greater than 1, because the constant factor hidden in $\Theta(\frac{(m+n)\log{n}}{\epsilon^3})$ is greater than 32.

\section{Preliminaries}\label{sec:pre}

In this section, we recall the Linear Threshold influence model and the Independent Cascade Model~\cite{kempe2003maximizing}.  We also review the polling method for computing influence spread~\cite{borgs2014maximizing, tang2014influence, tang2015influence}.  We then formulate influence in dynamic networks. For readers' convenience, Table~\ref{tab:notation} lists the frequently used notations.

\begin{table}[t]
\centering\small
\begin{tabular}{|c|p{55mm}|}
\hline
Notation & Description \\ \hline
$G=\langle V,E,w\rangle$ & A social network, where each edge $(u,v) \in E$ is associated with an influence weight $w_{uv}$\\ \hline
$w_{uv}$ & weight of the edge $(u,v)$ (LT model); propagation probability of the edge $(u,v)$ (IC model) \\ \hline
$n=|V|$ & The number of nodes in $G$ \\ \hline
$m=|E|$ & The number of edges in $G$ \\ \hline
$N^{in}(u)$ & The set of in-neighbors of $u$ \\ \hline
$w_{u}$ & Self-weight of $u$ \\ \hline
$W_{u}$ & $W_{u}=w_{u}+\sum_{v \in N^{in}(u)}{w_{vu}}$, the total weight of $u$ \\ \hline
$p_{uv}$ & $p_{uv}=\frac{w_{uv}}{W_u}$, the probability that $v$ is influenced by its neighbor $u$ (LT Model)\\ \hline
$I_u$ & The influence spread of node $u$ \\ \hline
$\bar{I}$ & The average influence spread of individual nodes \\ \hline
$M$ & The number of random RR sets \\ \hline
$\mathcal{R}$ & A random hyper graph, which can also be regarded as a collection of $M$ random RR sets \\ \hline
$\mathcal{D}(u)$ & The degree of $u \in V$ in $\mathcal{R}$ \\ \hline
$\mathcal{F_R}(u)$ & $\mathcal{F_R}(u)=\frac{\mathcal{D}(u)}{M}$, the fraction of random RR sets containing $u$ \\ \hline
$T$ & Influence threshold set by users \\ \hline
$I_{max}$ & Influence spread of the most influential individual node \\ \hline
$I^k$ & Influence spread of the $k-$th most influential individual node \\ \hline
$\mathcal{F_R}^*$ & The highest $\mathcal{F_R}(u)$ value for $u \in V$ \\ \hline
$\mathcal{F_R}^k$ & The $k-$th highest $\mathcal{F_R}(u)$ value for $u \in V$ \\ \hline
\end{tabular}
\caption{Frequently used notations.}
\label{tab:notation}
\end{table}

\subsection{Linear Threshold Model}\label{subsec:LT}

Consider a directed social network $G = \langle V,E,w \rangle$ where $V$ is a set of vertices, $E \subseteq V \times V$ is a set of edges, and each edge $(u,v) \in E$ is associated with an influence weight $w_{uv} \in [0,+\infty)$. Each node $v \in V$ also carries a weight $w_v$, which is called the \emph{self-weight} of $v$. Denote by $W_v=w_v+\sum_{u \in N^{in}(v)}{w_{uv}}$ the total weight of $v$, where $N^{in}(v)$ is the set of $v$'s in-neighbors. 

We define the \emph{influence probability} $p_{uv}$ of an edge $(u,v)$ as $\frac{w_{uv}}{W_v}$. Clearly, for $v \in V$, $\sum_{u \in N^{in}(v)}{p_{uv}} \leq 1$. 

In the Linear Threshold (LT) model~\cite{kempe2003maximizing}, given a seed set $S \subseteq V$, the influence propagates in $G$ as follows. First, every node $u$ randomly selects a threshold $\lambda_u \in [0,1]$, which reflects our lack of knowledge about users' true thresholds. Then, influence propagates iteratively. Denote by $S_i$ the set of nodes that are active in step $i$ $(i=0, 1, \ldots)$ and $S_0=S$. In each step $i \geq 1$, an inactive node $v$ becomes active if
$$
    \sum_{u \in N^{in}(v) \cap S_{i-1}}{p_{uv}} \geq \lambda_v
$$
The propagation stops at step $t$ if $S_t=S_{t-1}$. Let $I(S)$ be the expected number of nodes that are finally active when the seed set is $S$. We call $I(S)$ the \emph{influence spread} of $S$. Let $I_u$ be the influence spread of a single node $u$. 

Kempe~\textit{et~al.}~\cite{kempe2003maximizing} proved that the LT model is equivalent to a ``live-edge'' process where each node $v$ picks at most one incoming edge $(u,v)$ with probability $p_{uv}$.  Consequently, $v$ does not pick any incoming edges with probability $1-\sum_{u \in N^{in}(v)}{p_{uv}}=\frac{w_v}{W_v}$. All edges picked are ``live'' and the others are ``dead''. Then, the expected number of nodes reachable from $S \subseteq V$ through live edges is $I(S)$, the influence spread of $S$.

It is worth noting that our description of the LT model here is slightly different from the original~\cite{kempe2003maximizing}: we use a function of edge weights and self-weight of nodes to represent influence probabilities. Representing influence probabilities in this way is widely adopted in the existing literature~\cite{chen2010scalableLT, goyal2011simpath, tang2014influence, tang2015influence, goyal2010learning}.

\subsection{Independent Cascade Model}\label{subsec:IC}

A social network in the Independent Cascade (IC) model is also a weighted graph $G = \langle V,E,w \rangle$. Let $w_{uv}$ represent the propagation probability of the edge $(u,v)$, which is the probability that $v$ is activated by $u$ through the edge in the next step after u is activated. Clearly for the IC model, all $w_{uv} \in [0,1]$. 

In the IC model~\cite{kempe2003maximizing}, given a seed set $S \subseteq V$, the influence propagates in $G$ iteratively as follows. Denote by $S_i$ the set of nodes that are active in step $i$ $(i=0, 1, \ldots)$ and $S_0=S$. At step $i+1$, each node $u$ in $S_i$ has a single chance to activate each inactive neighbor $v$ with an independent probability $w_{uv}$. The propagation stops at step $t$ if $S_t=\emptyset$. Similar to the LT model, the influence spread $I(S)$ denotes the expected number of nodes that are finally active when the seed set is $S$.

The ``live-edge'' process~\cite{kempe2003maximizing} of the IC model is to keep each edge $(u,v)$ with a probability $w_{uv}$ independently. All kept edges are ``live'' and the others are ``dead''. Then, the expected number of nodes reachable from $S$ via live edges is the influence spread $I(S)$.

\subsection{The Polling Method for Influence Computation}\label{subsec:poll}

Computing influence spread is \#P-hard under both the LT model and the IC model~\cite{chen2010scalableLT, chen2010scalable}. Recently, a polling-based method~\cite{borgs2014maximizing, tang2014influence, tang2015influence} was proposed for approximating influence spread of triggering models~\cite{kempe2003maximizing} like the LT model and the IC model. Here we briefly review the polling method for computing influence spread.

Given a social network $G=\langle V,E,w \rangle$, a poll is conducted as follows: we pick a node $v \in V$ in random and then try to find out which nodes are likely to influence $v$. We run a Monte Carlo simulation of the equivalent ``live-edge'' process.  The nodes that can reach $v$ via live edges are considered as the potential influencers of $v$. The set of influencers found by each poll is called a \emph{random RR (Reversely Reachable) set}. 

Let $R_1$, $R_2$, ..., $R_M$ be a sequence of random RR sets generated by $M$ polls, where $M$ can also be a random variable. The $M$ random RR sets form a random hyper-graph $\mathcal{R}$ where the set of nodes is still $V$ and each random RR set is a hyper edge. Denote by $\mathcal{D}(S)$ the degree of a set of nodes $S$ in the hyper-graph, which is the number of hyper-edges containing at least one node in $S$. Let $\mathcal{F_R}(S)=\frac{\mathcal{D}(S)}{M}$. By the linearity of expectation, it has been shown that $n\mathcal{F_R}(S)$ is an unbiased estimator of $I(S)$~\cite{borgs2014maximizing, tang2015influence}. Tang~\textit{et~al.}~\cite{tang2015influence} proved that the corresponding sequence $x_1$, $x_2$, ..., $x_M$ is a martingale~\cite{chung2006concentration}, where $x_i=1$ if $S \cap RR_i \neq \emptyset$ and $x_i=0$ otherwise. We have $E[\sum_{i=1}^M{x_i}]=E[\mathcal{D}(S)]=\frac{MI(S)}{n}$. The following results~\cite{tang2015influence} show how $E[\sum_{i=1}^M{x_i}]$ is concentrated around $\frac{MI(S)}{n}$.

\begin{corollary}[\cite{tang2015influence}]\label{cor:c1}\small
For any $\xi>0$,
$$
\begin{aligned}
    \textup{Pr}\Big[\sum_{i=1}^M{x_i}-Mp \geq \xi Mp \Big] & \leq & \textup{exp}\Big ( -\frac{\xi^2}{2+\frac{2}{3} \xi} Mp \Big)\\
    \textup{Pr}\Big[\sum_{i=1}^M{x_i}-Mp \leq -\xi Mp \Big] & \leq & \textup{exp}\Big ( -\frac{\xi^2}{2} Mp \Big)
\end{aligned}
$$
where $p=\frac{I(S)}{n}$.
\end{corollary}

Sections~\ref{sec:threshold} and~\ref{sec:topk} will use the above results to analyze how many random RR sets are needed for extracting influential nodes. Note that since the problem we study in this paper is different from influence maximization, the results (theorems and lemmas) in~\cite{tang2015influence} cannot be applied to our analysis.

\subsection{Influence in Dynamic Networks}\label{subsec:DN}

Real online social networks, such as the Facebook network and the Twitter network, change very fast and all the time. Relationships among users keep changing, and influence strength of relationships also varies over time. Lei~\textit{et~al.}~\cite{lei2015online} pointed out that influence probabilities may change due to former inaccurate estimation or evolution of users' relations over time. However, the traditional formulation of dynamic networks only considers the topological updates, that is, edge insertions and edge deletions~\cite{bahmani2010fast, ohsaka2015efficient, hayashi2015fully}. Such a formulation is not suitable for realtime accurate analysis of influence.

According to the LT model reviewed in Section~\ref{subsec:LT}, the change of influence probabilities along edges can be reflected by the change of edge weights. For the IC model, since the weight of an edge is the propagation probability, the updates on edge weights are updates on propagation probabilities. Therefore, we model a dynamic network as a stream of weight updates on edges. 

A \emph{weight update} on an edge is a 5-tuple $(u,v,+/-,\Delta,t)$, where $(u,v)$ is the edge updated, $+/-$ is a flag indicating whether the weight of $(u,v)$ is increased or decreased, $\Delta>0$ is the amount of change to the weight and $t$ is the time stamp. The update is applied to the self-weight $w_u$ if $u=v$. Clearly, edge insertions/deletions considered in the existing literature~\cite{bahmani2010fast, ohsaka2015efficient, hayashi2015fully, chen2015influential} can be easily written as weight increase/decrease updates. Moreover, node insertions/deletions can be written as edge insertions/deletions, too.


\begin{example}\label{example:DN}\em
A retweet network is a weighted graph $G = \langle V,E,w \rangle$, where $V$ is a set of users. An edge $(u,v) \in E$ captures that user $v$ retweeted from user $u$. We can set $w_{uv}$ according to the propagation model adopted as follows.

\textit{LT Model:} The edge weight $w_{uv}$ is the number of tweets that $v$ retweeted from $u$. The self-weight $w_v$ is the number of original tweets posted by $v$. The weights reflect the influence in the social network. By intuition, if $v$ retweeted many tweets from $u$, $v$ is likely to be influenced by $u$. In contrast, if most of $v$'s tweets are original, $v$ is not likely to be influenced by others.

\textit{IC Model:} The edge weight $w_{uv}$ is the probability that $v$ retweets from $u$, which can be calculated according to $v$'s retweeting record in the past~\cite{saito2008prediction, goyal2010learning}. 
	
An essential task in online social influence analysis is to capture how the influence changes over time.  For example, one may want to consider only the retweets within the past $\Delta t$ time. Clearly, the set of edges $E$ may change and the weights $w_{uv}$ and $w_v$ may increase or decrease over time. The dynamics of the retweet network can be depicted by a stream of edge weight updates $\{(u,v,+/-,\Delta,t)\}$. 
\end{example}

Given a dynamic network like the retweet network in Example~\ref{example:DN}, how can we keep track of influential users dynamically? In order to know the influential nodes, the critical point is to monitor influence of users. To solve this problem, we adopt the polling-based method for computing influence spread, and extend it to tackle dynamic networks. The major challenge is how to maintain a number of RR sets over a stream of weight updates, such that $n\mathcal{F_R}(S)$ is always an unbiased estimator of $I(S)$. We propose a framework for updating RR sets that addresses various tasks of tracking influential nodes. 

The framework is shown in Algorithm~\ref{alg:framework}. In Section~\ref{sec:update}, we discuss how to efficiently update the existing RR sets.  How to decide if our current RR sets are insufficient, redundant or in proper amount depends on the specific task of tracking influential nodes. In Sections~\ref{sec:threshold} and~\ref{sec:topk}, respectively, we discuss this issue for two common tasks of tracking influential nodes, namely tracking nodes with influence greater than a threshold and tracking top-k influential nodes. 

\begin{algorithm}[t]
\caption{Framework of Updating RR Sets}
\label{alg:framework}
\begin{algorithmic}[1]
\STATE retrieve RR Sets affected by the updates of the graph 
\STATE update retrieved RR sets
\IF {the current RR sets are insufficient}
	\STATE add new RR sets
\ELSE 
	\IF {the current RR sets are redundant}
		\STATE delete the redundant RR sets
	\ENDIF
\ENDIF 
\end{algorithmic}
\end{algorithm}

\section{Updating RR Sets}\label{sec:update}

In this section, we propose an incremental algorithm for updating existing RR sets over a stream of edge weight updates under both the LT model and the IC model. We prove that, by updating RR sets using our algorithm, $n\mathcal{F_R}(S)$ is always an unbiased estimation of $I(S)$. We also analyze the cost of an update based on the assumption that we are maintaining in total $M$ RR sets. Note that the value of $M$ should be decided for specific tasks. In Sections~\ref{sec:threshold} and~\ref{sec:topk} we discuss the value of $M$ for two common tasks of tracking influential nodes.

\subsection{Updating under the LT Model}\label{subsec:update_LT}

First, we have a key observation about random RR sets for the LT model.

\begin{figure}[h]
    \centering
    \includegraphics[width=.3\textwidth]{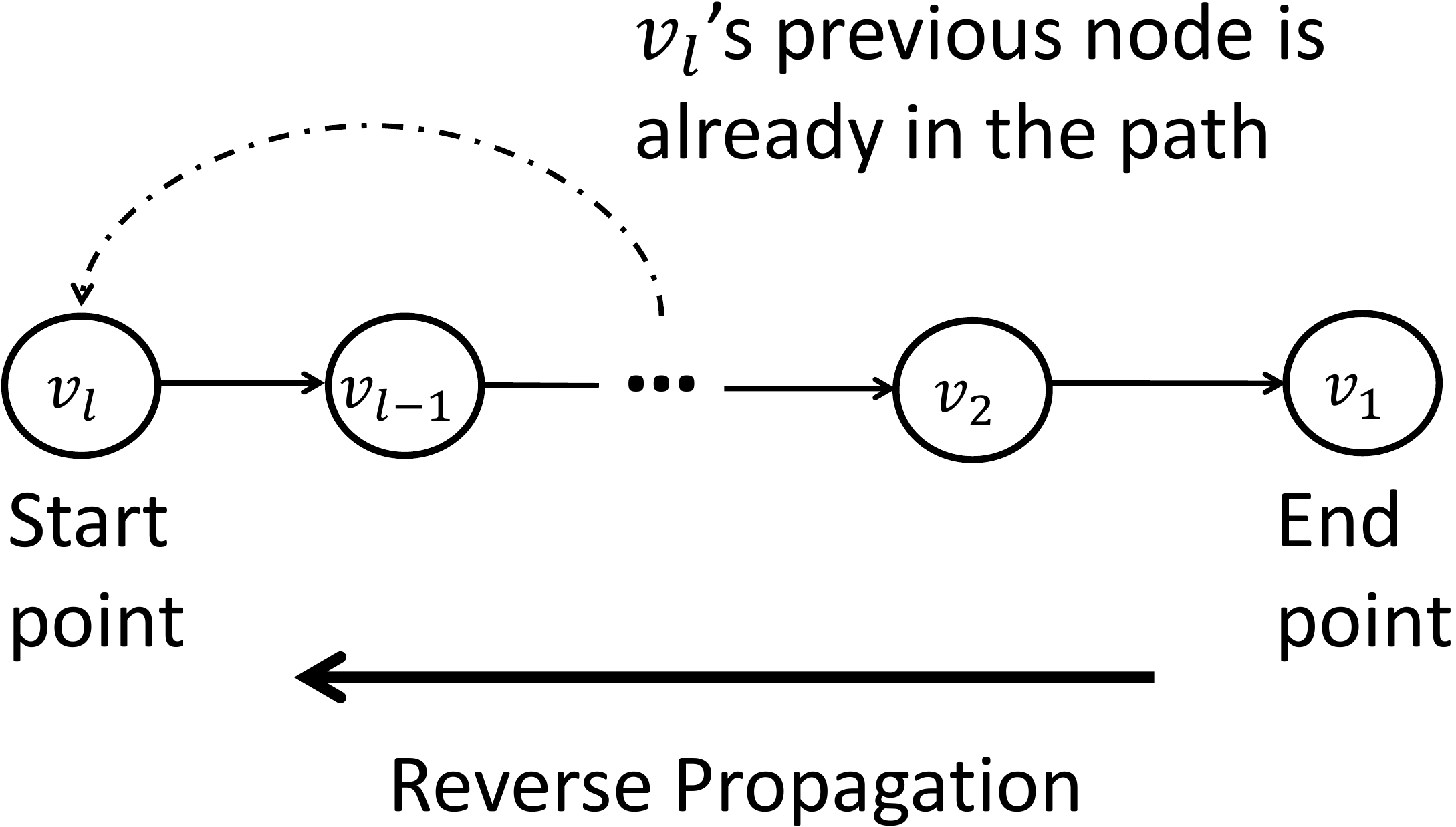}
    \caption{A random path. $v_{i}$ is the previous node of $v_{i-1}$.}
    \label{fig:path}
\end{figure}

\begin{fact}
    A random RR set of the LT model is a simple path.
\smallskip\em

{\sc Rationale.} In the equivalent ``live-edge'' selection process of the LT model, each node selects at most one incoming edge as a live edge. In the polling process, a random RR set is the set of nodes that can be reversely reachable from a randomly picked node $v$ via live edges. Thus, the nodes in a random RR set together form a simple path. 
\medskip
\end{fact}

Fig.~\ref{fig:path} illustrates a random RR set. The end point $v_1$ is picked in random at the beginning of the polling process. Then the path is generated by reversely propagating from $v_1$. The reverse propagation ends at $v_l$ because $v_l$ picks one of the nodes already in the path as its previous node. Note that the situation that $v_l$ does not pick any previous nodes can be regarded as $v_l$ picks itself as the previous node.

For a random RR set, suppose the starting node is $v_l$, we also store the previous node picked by $v_l$, which is useful in our algorithm for updating random RR sets maintained. 
Clearly the space complexity of a RR set is $O(L)$ where $L$ is the number of nodes in the RR set.
We maintain an inverted index on all random RR sets so that we can access all the random RR sets passing a node.  Moreover, we assume that the whole graph is stored and maintained in a way allowing random access to every node and its in-neighbors. It is not difficult to verify that the expected number of nodes of a RR set is $\bar{I}$, the average individual influence in the network. Thus, the expected space cost of $M$ RR sets and the inverted index is $O(M\bar{I}+n)$.

When there is an edge weight update $(u,v,+/-,\Delta,t)$ at time $t$, our incremental algorithm works as follows. Denote by $w_{uv}^t$ the edge weight of $(u,v)$ and $W_v^t$ the total weight of $v$ at time $t$. We first update the edge weight of $(u,v)$ and the total weight of $v$ in the graph. Then, we consider the following two cases.
\begin{enumerate}

\item If the update is a weight increase $(u,v,+,\Delta,t)$, we retrieve all RR sets passing $v$ using the inverted index. For each RR set retrieved, with probability $\frac{\Delta}{W_v^t}$ it is rerouted from $v$. If a RR set is rerouted, the previous node of $v$ is set to $u$ and we keep reversely propagating until no new nodes can be reversely reached.  

\item If the update is a weight decrease $(u,v,-,\Delta,t)$, we retrieve all RR sets passing $v$ where the previous node of $v$ is $u$. Each retrieved RR set is rerouted from $v$ with probability $\frac{\Delta}{w_{uv}^{t-1}}$. If a RR set is rerouted, we choose $u'$ among the in-neighbors of $v$ at time $t$ as the previous node of $v$ with probability $\frac{w_{u'v}^t}{W_v^t}$. We keep reversely propagating until no new nodes can be reversely reached.
\end{enumerate}
When rerouting random RR sets, we use random access to obtain the nodes and the in-neighbors of them in the graph. We also update the inverted index.


The update operations are similar to Reservoir Sampling~\cite{vitter1985random}. It is easy to prove that, at any time $t$, after the incremental maintenance, for any $(u,v)$ where $u$ is an in-neighbor of $v$, $u$ is picked as the previous node of $v$ with probability $\frac{w_{uv}^t}{W_v^t}$. Thus, $n\mathcal{F_R}(S)$ is always an unbiased estimator of $I(S)$ for any $S$.

\begin{theorem}
At any time $t$, after our incremental maintenance of the random RR sets under the LT model as described in this section, $n\mathcal{F_R}(S)$ is an unbiased estimator of $I(S)$ for any seed set $S$.
\end{theorem}
\begin{proof}
We only need to consider the basic case where, at time $t$, there is at most one edge weight update. A general case of multiple weight updates can be simply treated as a series of the basic case.

We prove by induction. Apparently, at time 0, when the network has no edges, the theorem holds.

Assume when $t=k-1$ $(k \geq 1)$, the probability that $v$ selects its in-neighbor $u$ is $pp_{uv}^{k-1}=\frac{w_{uv}^{k-1}}{W_v^{k-1}}$.

When $t=k$, three possible situations may arise.
\begin{description}
    \item [Case 1:] There is no update on any incoming edges of $v$. In such a case, for each $u$ that is an in-neighbor of $v$, $pp_{uv}^{k}=pp_{uv}^{k-1}=\frac{w_{uv}^{k-1}}{W_v^{k-1}}=\frac{w_{uv}^k}{W_v^k}$.
    \item [Case 2:] An edge weight increase $(u,v,+,\Delta,k)$ happens at time $t=k$. So, $w_{uv}^k=w_{uv}^{k-1}+\Delta$ and $W_v^k=W_v^{k-1}+\Delta$. For $u$, we have
    $$
        pp_{uv}^k = pp_{uv}^{k-1}(1-\frac{\Delta}{W_v^k})+\frac{\Delta}{W_v^k} = \frac{w_{uv}^{k-1}}{W_v^{k-1}}\frac{W_v^k-\Delta}{W_v^k}+\frac{\Delta}{W_v^k} = \frac{w_{uv}^k}{W_v^k}
    $$
    For any other in-neighbor $u'$ of $v$, at time $t=k$,
    $$
        pp_{u'v}^k = pp_{u'v}^{k-1}(1-\frac{\Delta}{W_v^k}) = \frac{w_{u'v}^{k-1}}{W_v^{k-1}}\frac{W_v^k-\Delta}{W_v^k} = \frac{w_{u'v}^k}{W_v^k}
    $$
    \item [Case 3:] An edge weight decrease $(u,v,-,\Delta,t)$ happens at time $t=k$. Note that $w_{uv}^k=w_{uv}^{k-1}-\Delta$ and $W_v^k=W_v^{k-1}-\Delta$. For $u$,
    \begin{align*}
        pp_{uv}^k &= pp_{uv}^{k-1}[(1-\frac{\Delta}{w_{uv}^{k-1}})+\frac{\Delta}{w_{uv}^{k-1}}\frac{w_{uv}^k}{W_v^k}] \\
                  &= \frac{w_{uv}^{k-1}}{W_v^{k-1}}[\frac{w_{uv}^k}{w_{uv}^{k-1}}+\frac{w_{uv}^k}{w_{uv}^{k-1}}\frac{\Delta}{W_v^k}]= \frac{w_{uv}^{k-1}}{W_v^{k-1}}\frac{w_{uv}^k}{w_{uv}^{k-1}}\frac{W_v^{k-1}}{W_v^k} = \frac{w_{uv}^k}{W_v^k}
    \end{align*}
    For any in-neighbor $u'$ of $v$ other than $u$,
    \begin{align*}
        pp_{u'v}^k &= pp_{u'v}^{k-1} +  pp_{uv}^{k-1}\frac{\Delta}{w_{uv}^{k-1}}\frac{w_{u'v}^k}{W_v^k} \\
          &=\frac{w_{u'v}^{k-1}}{W_v^{k-1}}+\frac{w_{uv}^{k-1}}{W_v^{k-1}}\frac{\Delta}{w_{uv}^{k-1}}\frac{w_{u'v}^k}{W_v^k}=\frac{w_{u'v}^k}{W_v^{k-1}}\frac{W_v^k+\Delta}{W_v^k} = \frac{w_{u'v}^k}{W_v^k}
    \end{align*}
\end{description}
By treating $v$ as also an in-neighbor of $v$ itself and thus $w_v$ is $w_{vv}$, we can prove the case when the weight update is on $w_v$.
\end{proof}

The expected number of RR sets needed to be retrieved is $\frac{MI_v^{t-1}}{n} \ll M$ for an update $(u,v,+/-,\Delta,t)$. Only a small fraction of the retrieved RR sets need to be updated. Specifically, the expected number of RR sets updated is $\frac{MI_v^{t-1}\Delta}{nW_v^t} \ll M$ for a weight increase update $(u,v,+,\Delta,t)$, and $\frac{MI_v^{t-1}\Delta}{nW_v^{t-1}} \ll M$ for a weight decrease update $(u,v,-,\Delta,t)$. Clearly the cost of incremental maintenance is much less than re-generating $M$ RR sets from scratch.

\subsection{Updating under the IC Model}\label{subsec:update_IC}

The idea of updating RR sets under the IC model is similar to~\cite{ohsaka2016dynamic}. We briefly introduce the idea in this section.

Rather than a simple path, a random RR set in the IC model is a random connected component. Fig.~\ref{fig:cc} illustrates an example. Suppose the start point (the randomly picked node at the beginning of a poll) of a RR set is $v_1$, then each node in this RR set can be reversely reachable from $v_1$ via live edges.  

\begin{figure}[h]
    \centering
    \includegraphics[width=.3\textwidth]{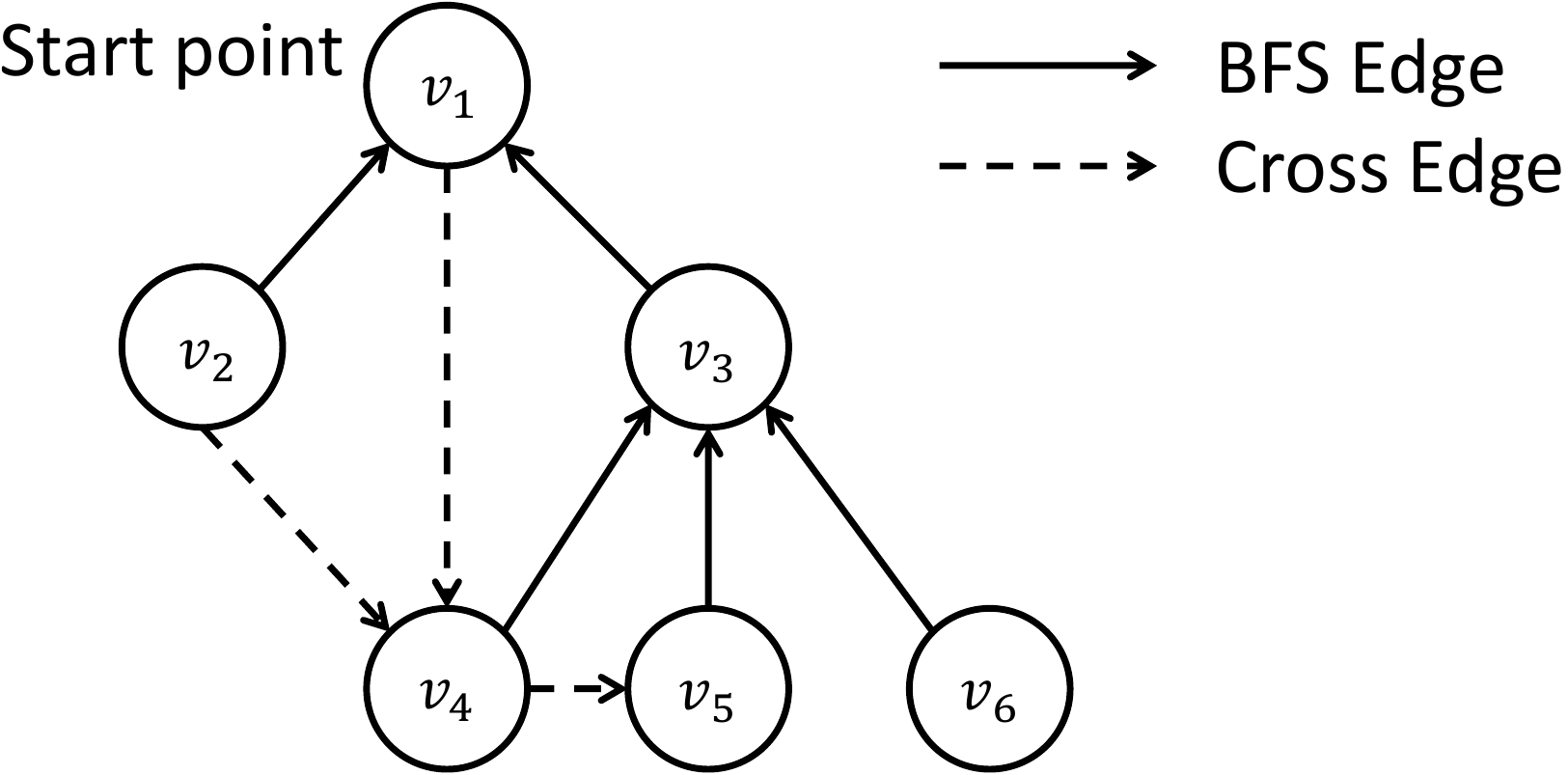}
    \caption{A random RR set of the IC model is a random connected component.}
    \label{fig:cc}
\end{figure}

For a random RR set, we not only record the nodes in it but also all live edges among those nodes. We categorize live edges into two classes, namely BFS edges and cross edges. When a RR set is being generated by reversely propagating from the start point in a breadth-first search manner, if a live edge $(v_i, v_j)$ makes $v_i$ propagated for the first time, $(v_i,v_j)$ is labeled as a BFS edge; otherwise it is labeled as a cross edge. For each node in a RR set, we use an adjacent list to store all live edges pointing to it. We also treat every node as a string and keep all nodes in a RR set in a prefix tree for fast retrieving a node and the address of its adjacent list of live edges. The major difference of our data structure for storing a RR set to the one in~\cite{ohsaka2016dynamic} is we do not store the propagation probabilities on live edges in a RR set, while~\cite{ohsaka2016dynamic} does. We only store propagation probabilities in the graph data structure. This is obviously an improvement in space because the propagation probability of an edge is only stored once in our method.

Like the LT model, for the IC model, we also maintain an inverted index on all random RR sets so that we can access all RR sets containing a node. Since in the ``live-edge'' process of the IC model, every edge is picked independently, when there is an update $(u,v,+/-,\Delta,t)$ at time $t$, status (``live'' or ``dead'') of edges other than $(u,v)$ in RR sets stay the same. Thus, we have the following incremental maintenance,

\begin{enumerate}

\item If the update is a weight increase $(u,v,+,\Delta,t)$, we retrieve all RR sets passing $v$ using the inverted index. For each RR set retrieved, if $(u,v)$ is not a live edge of it, we add $(u,v)$ as a live edge to it with probability $\frac{\Delta}{1-w_{uv}^{t-1}}$. After adding $(u,v)$, if $u$ does not belong to this RR set at time $t-1$, we further extend this RR set by reversely propagating from $u$ in a breadth-first search manner.  

\item If the update is a weight decrease $(u,v,-,\Delta,t)$, we retrieve all RR sets passing $v$. If a retrieved RR set contains a live edge $(u,v)$, with probability $\frac{\Delta}{w_{uv}^{t-1}}$ we remove $(u,v)$. If $(u,v)$ is removed, we traverse from the start point $v_1$ via live edges other than $(u,v)$ of this RR set to find all nodes reversely reachable from $v_1$ and all live edges among them. Then, this RR set is updated to one containing only those nodes and live edges we find.  
\end{enumerate}

Similar to the LT model, after updating the RR sets, we also update the inverted index.

Clearly, our incremental maintenance ensures that, for each edge $(u,v)$ at time $t$, if $v$ is a node of a RR set, the probability that $(u,v)$ is a live edge of this RR set is $w_{uv}^t$. So the same as the LT model, our incremental maintenance ensures that from the RR sets we can always have unbiased estimations of influence spreads. 

\begin{theorem}
At any time $t$, after our incremental maintenance of the random RR sets under the IC model as described in this section, $n\mathcal{F_R}(S)$ is an unbiased estimator of $I(S)$ for any seed set $S$.
\end{theorem}

In our incremental maintenance, we need to find out if an edge $(u,v)$ is a live edge in a RR set. Suppose the number of nodes in a RR set is $L$. Because normally the length of a node id is a constant, given an edge $(u,v)$, using the prefix tree we can find the address of $v$'s adjacent list in $O(1)$ time. Then a linear search is performed to find out if $(u,v)$ is a live edge. In practice propagation probabilities are often small and $\sum_{u \in N^{in}(u)}{w_{uv}}$ is often a small constant. Therefore, in practice the average complexity of the linear search is $O(1)$ and in total we only need $O(1)$ time to decide if $(u,v)$ is a live edge in a RR set. Moreover, the space complexity of the RR set is $O(L)$ in practice since every node only has a constant number of live edges pointing to it. Similar to the LT model, maintaining $M$ RR sets and the inverted index under the IC model takes $O(M\bar{I}+n)$ space in expectation, where $\bar{I}$ is the average individual influence.

For the second situation when a live edge $(u,v)$ is deleted, it is not always necessary to traverse from the start point, which takes $O(L)$ time if there are $L$ nodes in the RR set. It is easy to see that removing cross edges does not change the connectivity of nodes in a RR set. Thus, if the removed live edge is labeled as a cross edge, we do not need to further update the RR set. 

Similar to LT model, under IC model, the expected number of RR sets needed to be retrieved is $\frac{MI_v^{t-1}}{n} \ll M$ for an update $(u,v,+/-,\Delta,t)$ and only a small fraction of the retrieved RR sets need to be updated. The expected number of RR sets containing a live edge $(u,v)$ is $\frac{MI_{v}^{t-1}w_{uv}^{t-1}}{n}$ and the expected number of RR sets that do not contain $(u,v)$ as a live edge is $\frac{MI_{v}^{t-1}(1-w_{uv}^{t-1})}{n}$. Therefore, when there is an update on the edge $(u,v)$, no matter it is weight increase or weight decrease, the expected number of RR sets needed to be updated is $\frac{MI_v^{t-1}\Delta}{n} \ll M$. Clearly the cost of incremental maintenance is much less than re-generating $M$ RR sets from scratch.

\section{Tracking Threshold-based Influential Nodes}\label{sec:threshold}

A natural problem setting of finding influential nodes is to find all nodes whose influence spread is at least $T$, where $T$ is a user-specified threshold. In this section, we discuss how to use random RR sets to approximate the desired result.

Before our discussion, we clarify that our problem is not Heavy Hitters~\cite{cormode2008finding} even when we treat the influence spread of a node as the ``frequency/popularity'' of an element. First, the definitions of ``frequency'' are different and have dramatically different properties. In Heavy Hitters, a stream of items is a multiset of elements and the frequency of an element is its multiplicity over the total number of items. Thus, the sum of frequencies of all elements is $1$, which means there are at most $1/\phi$ elements with frequency passing a threshold $\phi$. In our problem, if we define the ``frequency'' of a node $v$ as $I_v/n$, the value of $\sum_{v \in V}{\frac{I_v}{n}}$ is not necessarily 1. Actually one can easily prove that computing $\sum_{v \in V}{I_v}$ is \#P-hard because computing $I_v$ is \#P-hard. As a result, normalizing $I_v$ is difficult. Thus, given any influence threshold $T < n$, we cannot have an upper bound on the number of nodes that have influence greater than $T$. Also, the input of our problem is a stream of edge updates but not a stream of insertion/deletion of nodes (elements). Moreover, the influence of a node is not a simple aggregation of weights on the associated edges.  In terms of technical solutions, it is hard to use a sublinear space to convert an update of edge weight to a list of  insertions/deletions of nodes. As illustrated in Section~\ref{sec:update}, we need both the graph and RR sets to decide which nodes should be increased/decreased in frequency by an edge update. This is very different from the settings of Heavy Hitters where only a sublinear space is allowed, while the graph itself already takes space $\Omega(n)$. We also need to access a number of RR sets, while in Heavy Hitters only counters of elements are allowed to be kept in memory.

Due to the \#P-hardness of computing influence spread under the LT model~\cite{chen2010scalableLT}, it is not likely that we can find in polynomial time the exact set of nodes whose influence spread is at least $T$. Thus, we turn to algorithms that allow controllable small errors. Specifically, we ensure that the recall of the set of nodes found by our algorithm is $100\%$ and we tolerate some false positive nodes. Moreover, the influence spread of those false positive nodes should take a high probability to have a lower bound that is not much smaller than $T$.  We set the lower bound to $T-\epsilon n$, where $\epsilon$ controls the error. 

According to Corollary~\ref{cor:c1}, the larger $M$, the more accurate the unbiased estimator $n\mathcal{F_R}(u)$. Thus, the intuition of deciding $M$ is to make sure that, for each $u$, $n\mathcal{F_R}(u)$ is large enough when $I_u \geq T$, and small enough when $I_u \leq T-\epsilon n$.

We first show that $n\mathcal{F_R}(u)$ is not likely to be too much smaller than $T$ if $I_u \geq T$ and $M$ is large enough.

\begin{lemma}\label{lemma:threshold1}
With $M$ random RR sets, if $I_u \geq T$, with probability at least $1-\textup{exp}(-\frac{M\epsilon^2 n}{8T})$, $n\mathcal{F_R}(u) \geq T-\frac{\epsilon n}{2}$.
\end{lemma}

\begin{proof}
If $I_u \geq T$, we have
\begin{align*}
    \textup{Pr}\{n\mathcal{F_R}(u) \leq T-\frac{\epsilon n}{2}\} &= \textup{Pr}\{n\mathcal{F_R}(u) \leq (1-\frac{I_u-T+\frac{\epsilon n}{2}}{I_u})I_u\}\\
    &\leq \textup{exp} \Big\{ -\frac{M(I_u-T+\frac{\epsilon n}{2})^2}{2nI_u} \Big\}
\end{align*}
$\frac{(I_u-T+\frac{\epsilon n}{2})^2}{I_u}$ is non-decreasing with respect to $I_u$ when $I_u \geq T$. Thus,
$$
    \textup{Pr}\{n\mathcal{F_R}(u) \leq T-\frac{\epsilon n}{2}\}  = \textup{exp}(-\frac{M\epsilon^2n}{8T})
$$
\end{proof}

Similarly, if $I_u \leq T-\epsilon n$, the probability that $n\mathcal{F_R}(u)$ is abnormally large is pretty small when $M$ is large.

\begin{lemma}\label{lemma:threshold2}
With $M$ random RR sets, if $I_u \leq T-\epsilon$, with probability at least $1-2\textup{exp}(-\frac{M\epsilon^2 n}{12T})$, $n\mathcal{F_R}(u) \leq T-\frac{\epsilon n}{2}$.
\end{lemma}
\begin{proof}
We prove that if $I_u \leq T-\epsilon n$, $\textup{Pr}\{n\mathcal{F_R}(u)-I_u \geq \frac{\epsilon n}{2}\} \leq 2\textup{exp}(-\frac{M\epsilon^2 n}{12T})$. Note that $n\mathcal{F_R}(u)-I_u \leq \frac{\epsilon n}{2}$ is a sufficient condition for $n\mathcal{F_R}(u) \leq T-\frac{\epsilon n}{2}$ when $I_u \leq T-\epsilon n$.

First, suppose $T \geq \frac{3\epsilon n}{2}$, which means $\frac{\epsilon n}{2} \leq T-\epsilon n$. There are two possible cases.
\begin{description}
    \item [Case 1:] $\frac{\epsilon n}{2} \leq I_u \leq T-\epsilon n$.  Then,
    \begin{align*}
        &~~~\textup{Pr}\{|n\mathcal{F_R}(u)-I_u| \geq \frac{\epsilon n}{2}\} \\
        &= \textup{Pr}\{|M\mathcal{F_R}(u)-\frac{MI_u}{n}| \geq \frac{\epsilon M}{2}\}\\
        &\leq 2\textup{exp}\{-\frac{1}{3}\frac{MI_u}{n}\frac{\epsilon^2 n^2}{4I_u^2}\} \leq 2\textup{exp}(-\frac{M\epsilon^2 n}{12T})
    \end{align*}
    \item [Case 2:] $I_u \leq \frac{\epsilon n}{2}$. Then,
    \begin{align*}
        &~~~~~\textup{Pr}\{n\mathcal{F_R}(u)-I_u \geq \frac{\epsilon n}{2}\} \\
        &= \textup{Pr}\{M\mathcal{F_R}(u)-\frac{MI_u}{n} \geq \frac{\epsilon M}{2}\}\\
        &\leq \textup{exp}\{-\frac{1}{(2+\frac{2}{3})\frac{\epsilon n}{2I_u}}\frac{MI_u}{n}\frac{\epsilon^2 n^2}{4I_u^2}\} \\
        &\leq \textup{exp}\{-\frac{3M\epsilon}{16}\} \leq 2\textup{exp}(-\frac{M\epsilon^2 n}{12T})
    \end{align*}
\end{description}

Second, if $T \leq \frac{3\epsilon n}{2}$, for all $I_u \leq T-\epsilon n$, $I_u \leq \frac{\epsilon n}{2}$. Then, all $I_u \leq T-\epsilon n$ fall into Case 2 above and the lemma still holds. 
\end{proof}

Because $\textup{exp}(-\frac{M\epsilon^2n}{8T}) \leq 2\textup{exp}(-\frac{M\epsilon^2 n}{12T})$, by applying Boole's inequality (that is, the Union Bound), with probability at least $1-2n\cdot \textup{exp}(-\frac{M\epsilon^2 n}{12T})$, every $n\mathcal{F_R}$ satisfies the conditions in Lemmas~\ref{lemma:threshold1} and~\ref{lemma:threshold2}.
Therefore, we have the following theorem on the sample size $M$ for finding nodes whose influence spread is at least $T$.

\begin{theorem}\label{th:threshold}
By setting the number of random RR sets $M=\frac{12T}{n\epsilon^2}\ln{\frac{2n}{\delta}}$, with probability at least $1-\delta$ the following conditions hold for every node $u$.
\begin{enumerate}
    \item If $I_u \geq T$, then $n\mathcal{F_R}(u) \geq T-\frac{\epsilon n}{2}$
    \item If $I_u < T-\epsilon n$, then $n\mathcal{F_R}(u)<T-\frac{\epsilon n}{2}$ 
\end{enumerate}
\end{theorem}

One nice property of $M$ in Theorem~\ref{th:threshold} is that, given $n$, $T$, $\epsilon$ and $\delta$, $M$ is a constant. Therefore, when we track nodes of influence spread at least $T$ in a dynamic network, no matter how the network changes, the sample size $M$ remains the same.

\section{Tracking Top-K Most Influential Nodes}\label{sec:topk}

Another useful problem setting is to find the top-k influential nodes, where $k$ is a user-specified parameter.

Denote by $I^k$ the influence spread of the $k$-th most influential node. Extracting top-k influential individual nodes equals extracting all nodes whose influence spread is at least $I^k$. Again, due to the \#P-hardness of influence computation, we probably have to tolerate errors in the result when designing algorithms. Similar to the task in Section~\ref{sec:threshold}, we hope the result returned by our algorithm contains all real top-k nodes, and for each false-positive node returned, its influence spread is no smaller than $I^k-\epsilon n$ with a high probability.

In this section, we first analyze the number of random RR sets $M$ we need to achieve the above goal with a high probability. We show that $M$ is proportional to the maximum individual influence spread $I_{max}$ and devise an algorithm that can give a really good estimation of $I_{max}$ with a high probability. Then, combining the theoretical results in Section~\ref{sec:threshold}, we propose a method that improves the precision of the result set of nodes, that is, reducing the number of false-positive nodes.

\subsection{Sample Size}\label{subsec:topk-size}

Unlike the task in section~\ref{sec:threshold}, we do not know the threshold $I^k$ in advance. Thus when selecting nodes according to values of $n\mathcal{F_R}(u)$, we do not have a threshold value. This is similar to mining top-k itemsets using sampled transactions~\cite{riondato2014efficient, riondato2015mining}. The intuition of our idea to solve the problem is that, if we have enough samples, we can bound the threshold value within a small range. 

To collect all real top-k influential nodes and filter out all nodes whose influence spreads are smaller than $I^k-\epsilon n$, we sample enough random RR sets such that for every $u \in V$, $|n\mathcal{F_R}(u)-I_u| \leq \frac{\epsilon n}{4}$ with a high probability. Denote by $\mathcal{F_R}^k$ the $k$-th highest $\mathcal{F_R}$ value. We have the following result.

\begin{lemma}\label{lemma:topk-size1}
    If for all $u \in V$, $|n\mathcal{F_R}(u)-I_u| \leq \frac{\epsilon n}{4}$, then the following conditions hold. (1) if $I_u \geq I^k$, then $n\mathcal{F_R}(u) \geq n\mathcal{F_R}^k-\frac{\epsilon n}{2}$; and (2) if $I_u \leq I^k-\epsilon n$, then $n\mathcal{F_R}(u) \leq n\mathcal{F_R}^k-\frac{\epsilon n}{2}$.
\end{lemma}
\begin{proof}
First, $n\mathcal{F_R}^k \geq I^k-\frac{\epsilon n}{4}$ because there are at least $k$ nodes having the value of $n\mathcal{F_R}$ at least $I^k-\frac{\epsilon n}{4}$. Second, $n\mathcal{F_R}^k \leq I^k+\frac{\epsilon n}{4}$ because there are at most $k$ nodes having the value of $n\mathcal{F_R}$ at least $I^k+\frac{\epsilon n}{4}$. Thus, we have $n\mathcal{F_R}^k-\frac{\epsilon n}{4} \leq I^k \leq n\mathcal{F_R}^k+\frac{\epsilon n}{4}$.

If $I_u \geq I^k$, we have $n\mathcal{F_R}(u) \geq I_u-\frac{\epsilon n}{4} \geq I^k-\frac{\epsilon n}{4} \geq n\mathcal{F_R}^k-\frac{\epsilon n}{2}$.

If $I_u \leq I^k-\epsilon n$, we have $n\mathcal{F_R}(u) \leq I_u+\frac{\epsilon n}{4} \leq I^k-\frac{3\epsilon n}{4} \leq n\mathcal{F_R}^k-\frac{\epsilon n}{2}$.
\end{proof}

So we need to derive a lower bound of $M$ to make sure $|n\mathcal{F_R}(u)-I_u| \leq \frac{\epsilon n}{4}$ for every $u \in V$ with a high probability. Denote by $I_{max}$ the maximum individual influence spread.

\begin{lemma}\label{lemma:topk-size2}
When the number of random RR sets is $M$, with probability at least $1-2\textup{exp}(-\frac{Mn\epsilon^2}{48I_{max}})$, $|n\mathcal{F_R}(u)-I_u| \leq \frac{\epsilon n}{4}$.   
\end{lemma}
\begin{proof}
We need to consider two possible cases.
\begin{description}
    \item [Case 1:] If $I_u \geq \frac{\epsilon n}{4}$
    \begin{align*}
        \textup{Pr}\{|n\mathcal{F_R}(u)-I_u| \geq \frac{\epsilon n}{4}\} &\leq 2\textup{exp}\{-\frac{1}{3}\frac{\epsilon^2n^2}{16I_u^2}\frac{MI_u}{n}\}\\ 
        &\leq 2\textup{exp}(-\frac{Mn\epsilon^2}{48I_{max}})
    \end{align*}
    \item [Case 2:] If $I_u \leq \frac{\epsilon n}{4}$
    \begin{align*}
        \textup{Pr}\{|n\mathcal{F_R}(u)-I_u| \geq \frac{\epsilon n}{4}\} &= \textup{Pr}\{n\mathcal{F_R}(u)-I_u \geq \frac{\epsilon n}{4}\}\\
        &\leq \textup{exp}(-\frac{3}{8}\frac{\epsilon n}{4I_u}\frac{MI_u}{n}) \\
        &= \textup{exp}(-\frac{3M\epsilon}{32})  \\
        &\leq 2\textup{exp}(-\frac{Mn\epsilon^2}{48I_{max}})
    \end{align*}
\end{description}
\end{proof}

By applying the Union Bound, with probability at least $1-2n\cdot\textup{exp}(-\frac{Mn\epsilon^2}{48I_{max}})$, we have $|n\mathcal{F_R}(u)-I_u| \leq \frac{\epsilon n}{4}$ for $u \in V$. In sequel, we have the following theorem settling the value of $M$.

\begin{theorem}\label{th:topk}
By setting the number of random RR sets $M \geq \frac{48I_{max}}{n\epsilon^2}\ln{\frac{2n}{\delta}}$, with probability at least $1-\delta$ the following conditions hold. (1) If $I_u \geq I^k$, then $n\mathcal{F_R}(u) \geq n\mathcal{F_R}^k-\frac{\epsilon n}{2}$; and (2) If $I_u < I^k-\epsilon n$, then $n\mathcal{F_R}(u)<n\mathcal{F_R}^k-\frac{\epsilon n}{2}$.
\end{theorem}

Unlike~\cite{riondato2014efficient, riondato2015mining, pietracaprina2010mining}, the sample size in Theorem~\ref{th:topk} not only depends on the confidence level $1-\delta$ and the error $\epsilon$, but also is proportional to $I_{max}$, which varies over different datasets. This is meaningful in practice, because for a social network, $I_{max}$ is normally very small comparing to $n$~\cite{chen2009efficient, chen2010scalableLT, goyal2011simpath, tang2015influence, du2013scalable, goyal2010learning}. One may link finding influential nodes with finding frequent itemsets remotely due to the intuition that a node frequent in many RR sets is likely influential.  In sampling based frequent itemsets mining~\cite{riondato2014efficient, riondato2015mining, pietracaprina2010mining}, the sample size is decided by $\epsilon$ and $\delta$ only, and thus is in general larger than ours here. 

\subsection{Estimating $I_{max}$ by Sampling}\label{subsec:topk-est}

The sample size $M$ in Theorem~\ref{th:topk} depends on $I_{max}$, which is unknown and hard to compute in exact. In this subsection, we devise a sampling algorithm that gives a tight upper bound of $I_{max}$ with a high probability.

Our algorithm sets $M=\frac{48x}{\epsilon^2}\ln{\frac{2n}{\delta}}$ and progressively increases $xn$ until it is enough larger than $I_{max}$. The intuition is that, if $n\mathcal{F_R}^*$ ($\mathcal{F_R}^*$ is the highest $\mathcal{F_R}(u)$ value) is sufficiently smaller than $xn$, probably the current $M$ is large enough.

\begin{algorithm}[t]
\caption{Sampling Sufficient Random RR sets for Top-K Influential Individuals}
\label{alg:topk}
\KwIn{$G=\langle V,E,w \rangle$, $\epsilon$, $\delta$ and $\mathcal{R}$ which is a set of random RR sets}
\KwOut{$\mathcal{R}$}
\begin{algorithmic}[1]
\WHILE {$|\mathcal{R}| < \frac{48\times 4\epsilon}{\epsilon^2}\ln{\frac{2n}{\delta}}$}
	\STATE Sample a random RR set and add to $\mathcal{R}$
\ENDWHILE
\STATE $x \leftarrow \frac{|\mathcal{R}|\epsilon^2}{48\ln{\frac{2n}{\delta}}}$
\WHILE {$\mathcal{F_R}^* \geq x-\epsilon$}
\STATE Sample a random RR sets and add to $\mathcal{R}$
\STATE $x \leftarrow \frac{|\mathcal{R}|\epsilon^2}{48\ln{\frac{2n}{\delta}}}$
\ENDWHILE
\RETURN $\mathcal{R}$
\end{algorithmic}
\end{algorithm}

Algorithm~\ref{alg:topk} shows our sampling method. We prove that the final random RR sets are enough and $xn$, the upper bound of $I_{max}$, is tight. 

\begin{lemma}\label{lemma:max1}
When $M=\frac{48x}{\epsilon^2}\ln{\frac{2n}{\delta}}$, if $I_{max} \geq xn$, with probability at least $1-\delta_1$, $\mathcal{F_R}^* \geq x-\epsilon$, where $\delta_1=(\frac{\delta}{2n})^{24}$ and $\mathcal{F_R}^*$ is the maximum $\mathcal{F_R}(u)$ for all $u \in V$.
\end{lemma}
\begin{proof}
Suppose $u$ is a node with the maximum influence. Since $xn \leq I_{max}$, we have\small
\begin{align*}
\textup{Pr}\{\mathcal{F_R}^* \leq x-\epsilon\} &\leq \textup{Pr}\{\mathcal{F_R}(u)\leq x-\epsilon\} \\
&= \textup{Pr}\{\mathcal{F_R}(u)\leq (1-\frac{\epsilon}{x})x\} \\
&\leq \textup{Pr}\{n\mathcal{F_R}(u) \leq (1-\frac{\epsilon}{x})I_{max}\} \\
&\leq \textup{exp}(-\frac{(\frac{\epsilon}{x})^2MI_{max}}{2n}) \\
&\leq \textup{exp}(-\frac{(\frac{\epsilon}{x})^2Mx}{2})=(\frac{\delta}{2n})^{24}
\end{align*}
\end{proof}

Lemma~\ref{lemma:max1} shows that with a high probability, if $\mathcal{F_R}^* < x-\epsilon$, then the current random RR sets are enough.

\begin{lemma}\label{lemma:max2}
When $xn \geq I_{max}$, $M=\frac{48x}{\epsilon^2}\ln{\frac{2n}{\delta}}$, with probability at least $1-\delta_2$, $\forall u$, if $I_u \leq (x-2\epsilon)n$, then $\mathcal{F_R}(u) \leq x-\epsilon$, where $\delta_2=n(\frac{\delta}{2n})^{16}$.
\end{lemma}
\begin{proof}
Let $\epsilon'=\frac{\epsilon}{x}$. Note that the first 3 lines of Algorithm~\ref{alg:topk} ensure that $x \geq 4\epsilon$ and $\epsilon' \leq \frac{1}{4}$. Thus, $\frac{1-\epsilon'}{2} \leq 1-2\epsilon'$. We have two possible cases.
\begin{description}
    \item [Case 1:] If $\frac{(1-\epsilon')xn}{2} \leq I_u \leq (1-2\epsilon')xn$
    \begin{align*}
        &~~~~\textup{Pr}\{\mathcal{F_R}(u) \geq x-\epsilon\} \\
        &= \textup{Pr}\{n\mathcal{F_R}(u) \leq [1+\frac{(1-\epsilon')xn-I_u}{I_u}]I_u\} \\
        &\leq \textup{exp}(-\frac{1}{3}\frac{[(1-\epsilon')xn-I_u]^2}{I_u^2}\frac{MI_u}{n}) \\
        &\leq \textup{exp}(-\frac{\epsilon'^2Mx}{3(1-2\epsilon')}) \leq \textup{exp}(-16\ln{\frac{2n}{\delta}})=(\frac{\delta}{2n})^{16}
    \end{align*}
    \item [Case 2:] If $I_u \leq \frac{(1-\epsilon')xn}{2}$
    \begin{align*}
        &~~~~\textup{Pr}\{n\mathcal{F_R}(u) \geq (1-\epsilon')xn\} \\
        &= \textup{Pr}\{n\mathcal{F_R}(u) \leq [1+\frac{(1-\epsilon')xn-I_u}{I_u}]I_u\} \\
        &\leq \textup{exp}(-\frac{3}{8}\frac{[(1-\epsilon')xn-I_u]}{I_u}\frac{MI_u}{n}) \\
        &\leq \textup{exp}(-\frac{3(1-\epsilon')Mx}{16}) \\
        &\leq \textup{exp}(\frac{9\ln{\frac{2n}{\delta}}}{2\epsilon'^2}) \leq \textup{exp}(-18\ln{\frac{2n}{\delta}})=(\frac{\delta}{2n})^{18}
    \end{align*}
\end{description}
Applying the Union Bound, we have that, with probability at least $1-n(\frac{\delta}{2n})^{16}$, $\mathcal{F_R}(u) \leq x-\epsilon$ for any $u$ such that $I_u \leq (x-2\epsilon)n$. 
\end{proof}

Lemma~\ref{lemma:max2} implies that when $(x-2\epsilon)n \geq I_{max}$, $\mathcal{F_R}^* \leq x-\epsilon$ with a high probability. The first time in Algorithm~\ref{alg:topk} when $(x-2\epsilon)n \geq I_{max}$ we have $xn \leq \max{(4\epsilon n, I_{max}+2\epsilon n)}$. If we set $\epsilon$ smaller than $\frac{I_{max}}{2n}$), the upper bound $xn$ is at most $2I_{max}$. This is achievable in practice since $I_{max}$ has some trivial lower bounds, such as $\max_{u \in V}{\sum_{v \in N^{out}(u)}{p_{uv}}}$.

\begin{theorem}\label{th:max}
Given $\epsilon$ and $\delta$, with probability $1-o(\frac{1}{n^{14}})$, Algorithm~\ref{alg:topk} returns $M=\frac{48x}{\epsilon^2}\ln{\frac{2n}{\delta}} \geq \frac{48I_{max}}{n\epsilon^2}\ln{\frac{\delta}{2n}}$ random RR sets, and $xn \leq \max{(4\epsilon n, I_{max}+2\epsilon n)}$.
\end{theorem}
\begin{proof}
There are only two possible reasons that Algorithm~\ref{alg:topk} may fail to achieve the above goals: (1) it stops sampling when $xn$ is still smaller than $I_{max}$; or (2) it does not stop sampling when $xn$ reaches $I_{max}+2\epsilon n$. Lemma~\ref{lemma:max2} indicates that the probability that (2) happens is at most $n(\frac{\delta}{2n})^{16}$. We bound the probability that (1) occurs.

Algorithm~\ref{alg:topk} stops when $\mathcal{F_R}^* < x-\epsilon$. According to Lemma~\ref{lemma:max1}, if $xn \leq I_{max}$, $\mathcal{F_R}^* < x-\epsilon$ happens with probability at most $(\frac{\delta}{2n})^{24}$. Before $xn$ is increased to $I_{max}$, the test whether $\mathcal{F_R}^* \geq x-\epsilon$ is called at most $\frac{48I_{max}\ln{\frac{2n}{\delta}}}{n\epsilon^2}=O(\log{n})$ times when $\epsilon$ and $\delta$ are fixed. Thus, the probability that Algorithm~\ref{alg:topk} stops before $xn$ reaches $I_{max}$ due to $\mathcal{F_R}^* < x-\epsilon$ is at most $O(\log{n})*(\frac{\delta}{2n})^{24}$.

Putting all things together and applying the Union bound, the failure probability is at most $O(\log{n})*(\frac{\delta}{2n})^{24}+n(\frac{\delta}{2n})^{16}=o(\frac{1}{n^{14}})$.
\end{proof}

When the network is updated, the value of $I_{max}$ may change. Thus, after we update the random RR sets, we call Algorithm~\ref{alg:topk} to ensure that we have enough but not too many random RR sets. In addition, if $I_{max}$ decreases dramatically, which means the current sample size is too large, we abandon some RR sets. Specifically, if $\mathcal{F_R}^* < x-\epsilon$, which means with very high probability that $I_{max} < xn$, we keep deleting the last RR set from $\mathcal{R}$ until if deleting the current last RR set leads to $\mathcal{F_R}^* > x-\epsilon$ (see Algorithm~\ref{alg:topk2}).

\begin{algorithm}[t]
\caption{Deleting Redundant Random RR Sets for Top-K Influential Individuals}
\label{alg:topk2}
\KwIn{$G=\langle V,E,w \rangle$, $\epsilon$, $\delta$ and $\mathcal{R}$ which is a set of random RR sets}
\KwOut{$\mathcal{R}$}
\begin{algorithmic}[1]
\WHILE {$\mathcal{F_R}^* < x-\epsilon \wedge |\mathcal{R}| > \frac{48\times 4\epsilon}{\epsilon^2}\ln{\frac{2n}{\delta}}$}
	\STATE $h \leftarrow \text{the last RR set of }\mathcal{R}$
	\STATE Delte $h$ from $\mathcal{R}$
	\IF {$\mathcal{F_R}^* \geq x-\epsilon \vee |\mathcal{R}| < \frac{48\times 4\epsilon}{\epsilon^2}\ln{\frac{2n}{\delta}}$}
		\STATE Add $h$ back to $\mathcal{R}$
		\STATE \textbf{break}
	\ENDIF
\ENDWHILE
\RETURN $\mathcal{R}$
\end{algorithmic}
\end{algorithm}

The reason we set the failure probability of estimating $I_{max}$ so small a value is because the number of updates in dynamic networks can be huge. When the failure probability is $o(\frac{1}{n^{14}})$, even there are $O(n^{14})$ updates, by applying the union bound, we can always get an tight upper bound of $I_{max}$ at any time with high probability.

\nop{
Combining Theorem~\ref{th:topk} and applying the Union bound, we have that, with probability at least $1-\delta-o(\frac{1}{n^{14}})$, our maintenance of RR sets ensures that by setting the filtering threshold $\mathcal{F_R}^k-\frac{\epsilon}{2}$, the set of nodes found includes all real top-K influential nodes and does not have any nodes such that $I_u < I^k-\epsilon n$.
}

\subsection{Collecting Influential Nodes and Improving Precision}\label{subsec:topk-precision}
Although Theorem~\ref{th:max} tells us that $M=|\mathcal{R}| \geq \frac{48I_{max}}{n\epsilon^2}\ln{\frac{\delta}{2n}}$ with high probability, we cannot directly apply Theorem~\ref{th:topk} on $\mathcal{R}$ to collect influential nodes by setting the filtering threshold $n\mathcal{F_R}^k-\frac{\epsilon n}{2}$. This is because probabilistic supports are not transitive\footnote{Suppose A implies B with probability $1-\delta_1$, and B implies C with probability $1-\delta_2$. A flawed argument that uses the transitivity is that by applying the union bound, A implies C with probability $1-\delta_1-\delta_2$.}~\cite{shogenji2003condition}. To better understand this issue, note that the probability that both (1) and (2) hold in Theorem~\ref{th:topk} is actually a conditional probability that $\textup{Pr}\{(1)\&(2) \mid \mathcal{F_R}^* \geq x-\epsilon~\text{for the first time}\}$ rather than $\textup{Pr}\{(1)\&(2) \mid M \geq \frac{48I_{max}}{n\epsilon^2}\ln{\frac{2n}{\delta}}\}$. Although ``$\mathcal{F_R}^* \geq x-\epsilon~\text{for the first time}$'' implies ``$M \geq \frac{48I_{max}}{n\epsilon^2}\ln{\frac{2n}{\delta}}$'' with a high probability, these two conditions are not exactly the same. To fix this issue, we need another collection $\mathcal{R}_1$ that consists of $M$ independently generated RR sets. In such a case, we do not have any prior knowledge about the $M$ RR sets in $\mathcal{R}_1$, which is different from that we know that for $\mathcal{R}$, $\mathcal{F}_{\mathcal{R}}^* \leq x-\epsilon$. The update of $\mathcal{R}_1$ is almost the same as the update of $\mathcal{R}$. After updating $\mathcal{R}$, we first update all RR sets in $\mathcal{R}_1$ using methods in Section~\ref{sec:update}, and then adjust the size of $\mathcal{R}_1$ to make $|\mathcal{R}_1|=M=\mathcal{R}$.

When $\mathcal{R}_1$ has enough RR sets, according to Theorem~\ref{th:topk}, we filter out nodes such that $\mathcal{F}_{\mathcal{R}_1}(u)<\mathcal{F}_{\mathcal{R}_1}^k-\frac{\epsilon}{2}$. Using Theorem~\ref{th:threshold}, we can further improve the filtering threshold to make the precision higher.

\begin{theorem}\label{th:topk-precision}
By keeping $|\mathcal{R}_1|=M=|\mathcal{R}|$, with probability at least $1-2\delta-o(\frac{1}{n^{14}})$, the following conditions hold.
\begin{enumerate}
    	\item If $I_u \geq I^k$, then $\mathcal{F}_{\mathcal{R}_1}(u) \geq \mathcal{F}_{\mathcal{R}_1}^k-\frac{\epsilon}{4}-\frac{\epsilon_1}{2}$
    	\item If $I_u < I^k-\epsilon n$, then $\mathcal{F}_{\mathcal{R}_1}(u)<\mathcal{F}_{\mathcal{R}_1}^k-\frac{\epsilon}{4}-\frac{\epsilon_1}{2}$ 
\end{enumerate}
where $\epsilon_1=\sqrt{\frac{\mathcal{F}_{\mathcal{R}_1}^k-\frac{\epsilon}{4}}{4x}}\epsilon \leq \frac{\epsilon}{2}$.
\end{theorem}
\begin{proof}
	According to Theorem~\ref{th:max}, after adjusting the number of RR sets by Algorithm~\ref{alg:topk} and Algorithm~\ref{alg:topk2}, with probability $1-o(\frac{1}{n^{14}})$, we have $M=\frac{48x}{\epsilon^2}\ln{\frac{2n}{\delta}} \geq \frac{48I_{max}}{n\epsilon^2}\ln{\frac{\delta}{2n}}$. When $|{\mathcal{R}_1}|=M \geq \frac{48I_{max}}{n\epsilon^2}\ln{\frac{\delta}{2n}}$, with probability at least $1-\delta$, $n\mathcal{F}_{\mathcal{R}_1}^k-\frac{\epsilon n}{4} \leq I^k \leq n\mathcal{F}_{\mathcal{R}_1}^k+\frac{\epsilon n}{4}$ (Lemma~\ref{lemma:topk-size1} and Lemma~\ref{lemma:topk-size2}). 
	
	Suppose $M=\frac{48x}{\epsilon^2}\ln{\frac{2n}{\delta}} \geq \frac{48I_{max}}{n\epsilon^2}\ln{\frac{\delta}{2n}}$ and $n\mathcal{F}_{\mathcal{R}_1}^k-\frac{\epsilon n}{4} \leq I^k \leq n\mathcal{F}_{\mathcal{R}_1}^k+\frac{\epsilon n}{4}$. Let $\epsilon_1=\sqrt{\frac{\mathcal{F}_{\mathcal{R}_1}^k-\frac{\epsilon}{4}}{4x}}\epsilon \leq \frac{\epsilon}{2}$. Applying Theorem~\ref{th:threshold} and setting the threshold $T=(\mathcal{F}_{\mathcal{R}_1}^k-\frac{\epsilon}{4})n$, with probability at least $1-\delta$, we have the follows. (1) if $I_u \geq (\mathcal{F}_{\mathcal{R}_1}^k-\frac{\epsilon}{4})n$, then $\mathcal{F}_{\mathcal{R}_1}(u) \geq \mathcal{F}_{\mathcal{R}_1}^k-\frac{\epsilon}{4}-\frac{\epsilon_1}{2}$; and (2) if $I_u < (\mathcal{F}_{\mathcal{R}_1}^k-\frac{\epsilon}{4}-\epsilon_1)n$, then $\mathcal{F}_{\mathcal{R}_1}(u)<\mathcal{F}_{\mathcal{R}_1}^k-\frac{\epsilon}{4}-\frac{\epsilon_1}{2}$. Clearly $I^k \geq (\mathcal{F}_{\mathcal{R}_1}^k-\frac{\epsilon}{4})n$ and $I^k-\epsilon n \leq n\mathcal{F}_{\mathcal{R}_1}^k-\frac{3\epsilon n}{4} \leq (\mathcal{F}_{\mathcal{R}_1}^k-\frac{\epsilon}{4}-\epsilon_1)n$.
	
	Applying the Union bound, we have that with probability at least $1-2\delta-o(\frac{1}{n^{14}})$, the above conditions hold after executions of Algorithm~\ref{alg:topk} and Algorithm~\ref{alg:topk2}.
\end{proof}

Theorem~\ref{th:topk-precision} shows that we can use a tighter filtering threshold $\mathcal{F}_{\mathcal{R}_1}^k-\frac{\epsilon}{4}-\frac{\epsilon_1}{2}$ which is no greater than the original one $\mathcal{F}_{\mathcal{R}_1}^k-\frac{\epsilon}{2}$. Meanwhile, the failure probability is only increased by $\delta$ at most. 

\subsection{Maintaining Nodes Ranking Dynamically}\label{subsec:node-rank}

Besides efficiently updating RR sets from where accurate estimations of influence spreads of influential nodes can be obtained, how to maintain the set of influential nodes is also an essential building block of influential nodes mining on dynamic networks. A brute force solution is to perform an $O(n\log{n})$ sorting every time after an update but the cost may be  unacceptably high in practice. To solve this problem, we adopt the data structure for maximum vertex cover in a hyper graph~\cite{borgs2014maximizing}. This data structure can help us maintain all nodes sorted by their estimated influence spreads, which are proportional to their degrees in a collection of RR sets $\mathcal{R}$. Clearly, if all nodes are sorted, the set of influential nodes are those ones in the top. Fig.~\ref{fig:linkedlist} shows the data structure. 

\begin{figure}[h]
    \centering
    \includegraphics[width=.3\textwidth]{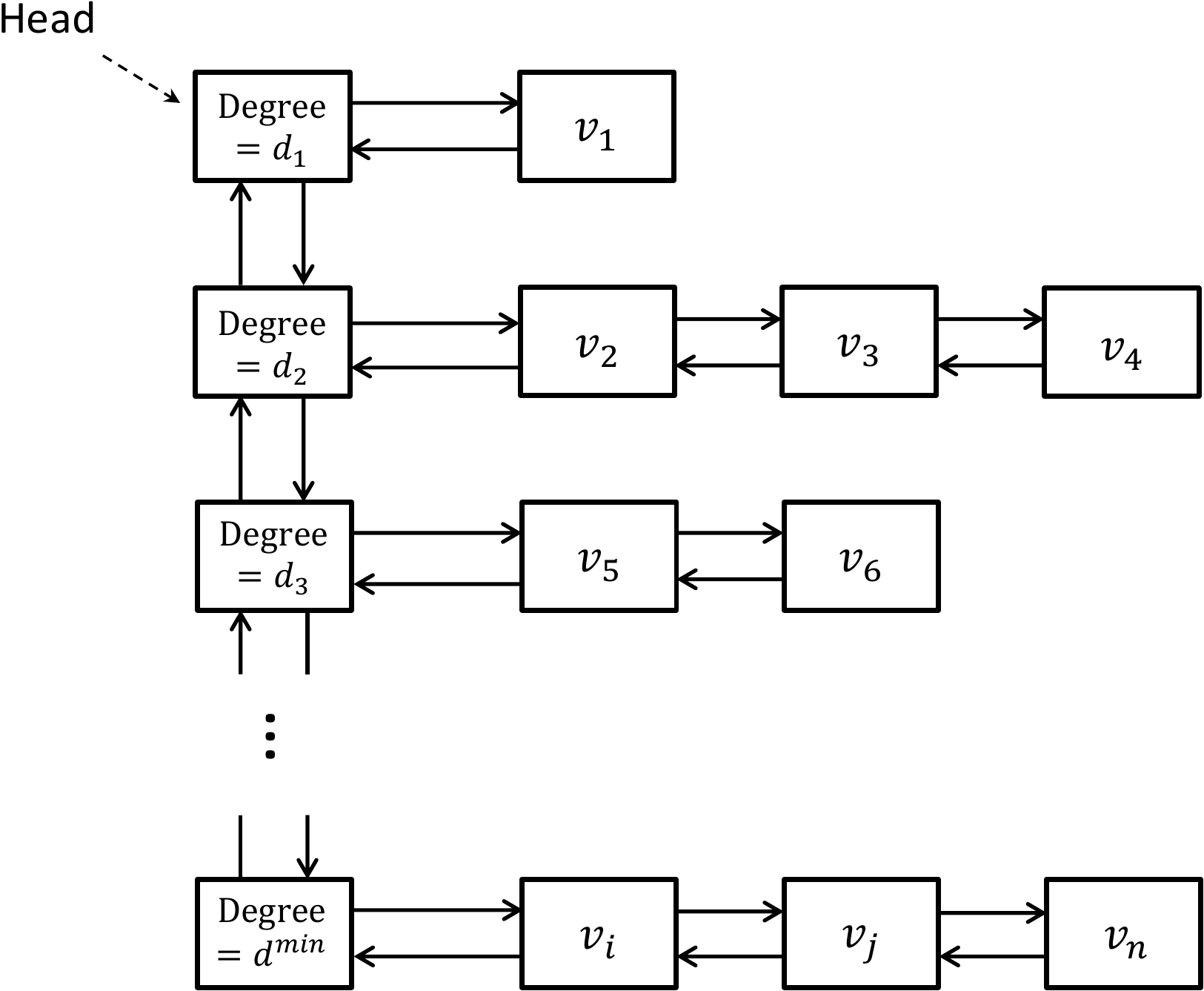}
    \caption{Linked List Structure, where $d_1 > d_2 > d_3 >...>d^{min}$}
    \label{fig:linkedlist}
\end{figure}

We maintain all nodes sorted by their degrees in $\mathcal{R}$ (recall that $\mathcal{D}(u)$, the degree of $u$ in $\mathcal{R}$, means how many RR sets contain $u$). Nodes with the same degree in $\mathcal{R}$ are grouped together and stored in a doubly linked list like in Fig.~\ref{fig:linkedlist}. Moreover, for those nodes, we create a head node which is the start of the linked list containing all nodes with the same given degree. Apparently, the number of head nodes is the number of distinctive values of $\mathcal{D}(u)$ in $\mathcal{R}$. We also maintain all head nodes sorted in a doubly linked list. For each $u \in V$, we maintain its address in the doubly linked lists structure and the corresponding head node. Note that when a RR set is updated, a new RR set is generated or an existing RR set is deleted, $\mathcal{D}(u)$ changes at most by 1 for each $u$. Thus, every time when $\mathcal{D}(u)$ is updated (increased or decreased by 1), we only need $O(1)$ time to find the head node of the linked list $u$ should be in (if such a head node does not exist now, we can create it and insert it into the doubly linked list of head nodes in $O(1)$ time) and insert it to the next of the head in $O(1)$ time. If after an update, a head node has no nodes after it, we delete it from the doubly linked list of head nodes in $O(1)$ time. Therefore, in total maintaining the linked list data structure only costs $O(1)$ time when the degree of a node changes due to the update of an RR set. With this data structure, we can always maintain all nodes sorted by their degrees in $\mathcal{R}$. Also, retrieving $\mathcal{F_R}^*$, which is needed in the frequently called test whether $\mathcal{F_R}^* \leq x-\epsilon$ , can be done in $O(1)$ time.

\section{Experiments}\label{sec:exp}
In this section, we report a series of experiments on $5$ real networks to verify our algorithms and our theoretical analysis. The experimental results demonstrate that our algorithms are both effective and efficient.

\subsection{Experimental Settings}\label{sec:exp_setting}
We ran our experiments on $5$ real network data sets that are publicly available online (\url{http://konect.uni-koblenz.de/networks/}, \url{http://www.cs.ubc.ca/~welu/} and \url{http://konect.uni-koblenz.de}). Table~\ref{tab:Datasets} shows the statistics of the four data sets.

\begin{table}
\centering
\begin{tabular}{|c|c|c|c|c|}
\hline
Network & \#Nodes & \#Edges & Average degree \\ \hline
wiki-Vote & 7,115 & 103,689 & 14.6 \\ \hline
Flixster & 99,053 & 977,738 & 9.9 \\ \hline
soc-Pokec & 1,632,803 & 30,622,564 & 18.8\\ \hline
flickr-growth & 2,302,925 & 33,140,018 & 14.4 \\ \hline
Twitter & 41,652,230 & 1,468,365,182 & 35.3 \\ \hline
\end{tabular}
\caption{The statistics of the data sets.}
\label{tab:Datasets}
\end{table}

To simulate dynamic networks, for each data set, we randomly partitioned all edges exclusively into 3 groups: $E_1$ (85\% of the edges), $E_2$ (5\% of the edges) and $E_3$ (10\% of the edges). We used $B=\langle V,E_1 \cup E_2\rangle$ as the base network. $E_2$ and $E_3$ were used to simulate a stream of updates. 

For the LT model, for each edge $(u,v)$ in the base network, we set the weight to be 1. For each edge $(u,v) \in E_3$, we generated a weight increase update $(u,v,+,1)$ (timestamps ignored at this time). For each edge $(u, v) \in E_2$, we generated one weight decrease update $(u,v,-,\Delta)$ and one weight increase update $(u,v,+,\Delta)$ where $\Delta$ was picked uniformly at random in $[0,1]$. We randomly shuffled those updates to form an update stream by adding random time stamps. For each data set, we generated 10 different instances of the base network and update stream, and thus ran the experiments 10 times. Note that for the 10 instances, although the base networks and update streams are different, the final snapshots of them are identical to the data set itself.

For the IC model, we first assigned propagation probabilities of edges in the final snapshot, i.e. the whole graph. We set $w_{uv}=\frac{1}{\text{in-degree}(v)}$, where $\text{in-degree}(v)$ is the number of in-neighbors of $v$ in the whole graph. Then, for each edge $(u,v)$ in the base network, we set $w_{uv}$ to $\frac{1}{\text{in-degree}(v)}$. For each edge $(u,v) \in E_3$, we generated a weight increase update $(u,v,+,\frac{1}{\text{in-degree}(v)})$ (timestamps ignored at this time). For each edge $(u, v) \in E_2$, we generated one weight decrease update $(u,v,-,\Delta\frac{1}{\text{in-degree}(v)})$ and one weight increase update $(u,v,+,\Delta\frac{1}{\text{in-degree}(v)})$ where $\Delta$ was picked uniformly at random in $[0,1]$. We randomly shuffled those updates to form an update stream by adding random time stamps. For each dataset we also generated 10 instances.

For the parameters of tracking nodes of influence at least $T$, we set $\epsilon=0.0002$, $\delta=0.001$, and $T=0.001 \times n$ for the first four data sets. We set $\epsilon=0.001$, $\delta=0.001$, and $T=0.005 \times n$ for the twitter data set. For the top-K influential individuals tracking task, we set $K=50$, $\delta=0.001$, and $\epsilon=0.0005$ for first four data sets. We set $K=100$, $\delta=0.001$, and $\epsilon=0.0025$ for the twitter data set. The reason we have different parameter settings for the twitter data is that it has more influential nodes than other networks.

All algorithms were implemented in Java and ran on a Linux machine of an Intel Xeon 2.00GHz CPU and 1TB main memory.

\subsection{Effectiveness}

We first assess the effectiveness of our techniques.

\subsubsection{Verifying Provable Quality Guarantees}
A challenge in evaluating the effectiveness of our algorithms is that the ground truth is hard to obtain. The existing literature of influence maximization~\cite{kempe2003maximizing, chen2010scalable, goyal2011simpath, tang2014influence, tang2015influence, cohen2014sketch} always use the influence spread estimated by 20,000 times Monte Carlo (MC) simulations as the ground truth. However, such a method is not suitable for our tasks, because the ranking of nodes really matters here. Even 20,000 times MC simulations may not be able to distinguish nodes with close influence spread. As a result, the ranking of nodes may differ much from the real ranking. Moreover, the effectiveness of our algorithms has theoretical guarantees while 20,000 times MC simulations is essentially a heuristic. It is not reasonable to verify an algorithm with a theoretical guarantee using the results obtained by a heuristic method without any quality guarantees.

In our experiments, we only used wiki-Vote and Flixster to run MC simulations and compare the results to those produced by our algorithms. We used 2,000,000 times MC simulations as the (pseudo) ground truth in the hope we can get more accurate results. According to our experiments, even so many MC simulations may generate slightly different rankings of nodes in two different runs but the difference is acceptably small. 
We only compare results on the identical final snapshot shared by all instances because running MC simulations on multiple snapshots is unaffordable (10 days on the final snapshots of Flixster).

\begin{table*}
\centering \caption{Recall and Maximum Error. The errors are measured in absolute influence value. ``w.h.p.'' is short for ``with high probability''.} \label{tab:recall_error}
\small
\begin{tabular}{|c|c|c|c|c|c|c|}
	\hline
		\multirow{2}{*}{} & \multicolumn{3}{c|}{wiki-Vote} & \multicolumn{3}{c|}{Flixster} \\ \cline{2-7}
			& Theoretical Value (w.h.p.) & Ave.$\pm$SD (LT) & Ave.$\pm$SD (IC) & Theoretical Value (w.h.p.) & Ave.$\pm$SD (LT) & Ave.$\pm$SD (IC) \\ \hline 
		Recall (Threshold) & 100\%   &  100\%   &  100\%   &  100\%  &  100\%  &  100\% \\ \hline
		Max.\ Error (Threshold) & $0.0002*7115=1.423$ & 0.758$\pm$0.033 & 0.814$\pm$0.013 & $0.0002*99053=19.81$ & 10.81$\pm$0.46 & 11.79$\pm$0.85 \\ \hline
		Recall (Topk-K) & 100\%   &  100\%   &  100\%   &  100\%  &  100\%  &  100\% \\ \hline
		Max.\ Error (Top-K) & $0.0005*7115=3.558$ & 1.254$\pm$0.080 & 1.272$\pm$0.090 & $0.0005*99053=49.53$ & 21.77$\pm$0.87 & 21.17$\pm$0.56 \\ \hline
\end{tabular}
\end{table*}

Table~\ref{tab:recall_error} reports the recall of the sets of influential nodes returned by our algorithms and the maximum errors of the false positive nodes in absolute influence value. Ave.$\pm$SD represents the average value and the standard deviation of a measurement on 10 instances. Our methods achieved 100\% recall every time as guaranteed theoretically.  Moreover, the real errors in influence were substantially smaller than the maximum error bound provided by our theoretical analysis. One may ask why we do not report the precision here. We argue that precision is indeed not a proper measure for our tasks when $100\%$ recall is required. Since we can only estimate influence spreads of nodes via a sampling method due to the exact computation being \#P-hard, if two nodes have close influence spreads, say $I_u=100$ and $I_v=99$, it is hard  for a sampling method to tell the difference between $I_u$ and $I_v$. Thus, if there are many nodes whose influence spreads are just slightly smaller than the threshold, it is hard to achieve a high precision when ensuring 100\% recall. Moreover, with a high probability, our method guarantees that influence spreads of false positive nodes are not far away from the real threshold. Such small errors are completely acceptable in many real applications.

\begin{table*}
\centering \caption{Estimated upper bounds of $I_{max}$ and the experimental values.}
\label{tab:est}\small
\begin{tabular}{*{7}{|c}|}
    \hline
    	\multirow{2}{*}{Dataset} & \multicolumn{3}{c|}{LT Model} & \multicolumn{3}{c|}{IC Model} \\ \cline{2-7}
			&	$xn$(Ave.$\pm$SD)	&	$I_{max}$	&	Approx. Ratio(Ave.$\pm$SD)	&	$xn$(Ave.$\pm$SD)	&	$I_{max}$	&	Approx. Ratio(Ave.$\pm$SD) \\ \hline
	wiki-Vote	& 57.7297$\pm$0.3372	&	51.770	&	1.1151$\pm$0.0065	&	55.2079$\pm$0.0941	&	51.7697	&	1.0664$\pm$0.0018 \\ \hline
	Flixster	& 456.3168$\pm$2.3474	&	404.2442	&	1.1288$\pm$0.0058	&	419.9483$\pm$1.6652	&	372.2075	&	1.1283$\pm$0.0045 \\ \hline
\end{tabular}
\end{table*}

Table~\ref{tab:est} reports our estimation of the upper bound of $I_{max}$ on the final snapshot of each network when tracking top-k influential nodes. The results indicate that the upper bound estimated by our algorithm is only a little greater than the real $I_{max}$.

For the $3$ large data sets, we did not run 2,000,000 times MC simulations to obtain the pseudo ground truth since the MC simulations are too costly. Instead, we compare the similarity between the results generated by different instances. Recall that the final snapshots of the 10 instances are the same. If the sets of influential nodes at the final snapshots of the 10 instances are similar, at least our algorithms are stable, that is, insensitive to the order of updates. To measure the similarity between two sets of influential nodes, we adopted the Jaccard similarity.

Fig.~\ref{fig:effectiveness} shows the results where I1, \ldots, I10 represent the results of the first, \ldots, tenth instances, respectively. We also ran the sampling algorithm directly on the final snapshot, that is, we computed the influential nodes directly from the final snapshot using sampling without any updates. The result is denoted by ST. 
The results show that the outcomes from different instances are very similar, and they are similar to the outcome from ST, too. The minimum similarity in all cases is 87\%. 

\begin{figure*}[h]
  \centering
    \subfigure[\small{soc-Pokec (Thres LT)}]{
      \includegraphics[width=.22\textwidth]{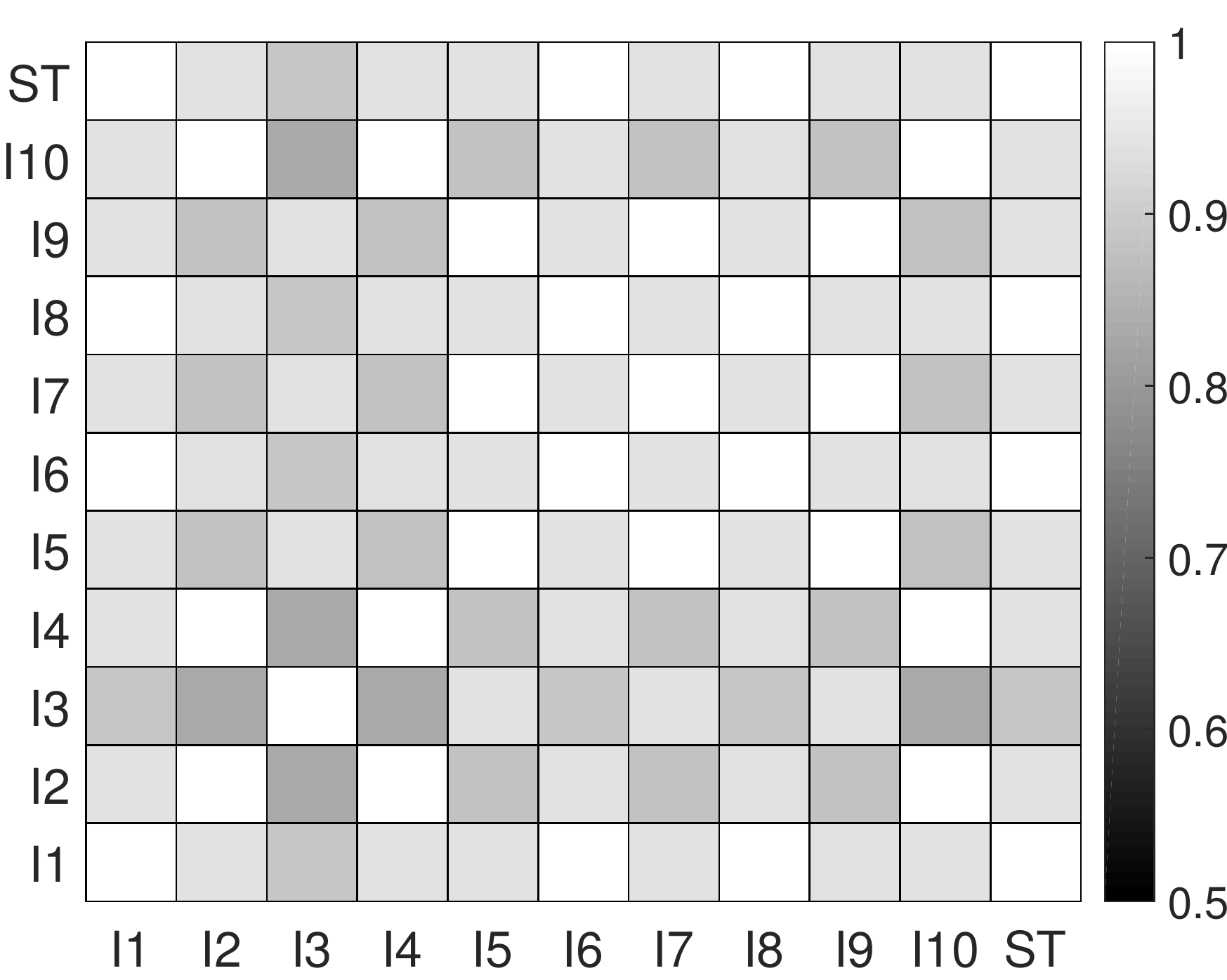}
    }
    \subfigure[\small{soc-Pokec (Top-K LT)}]{
      \includegraphics[width=.22\textwidth]{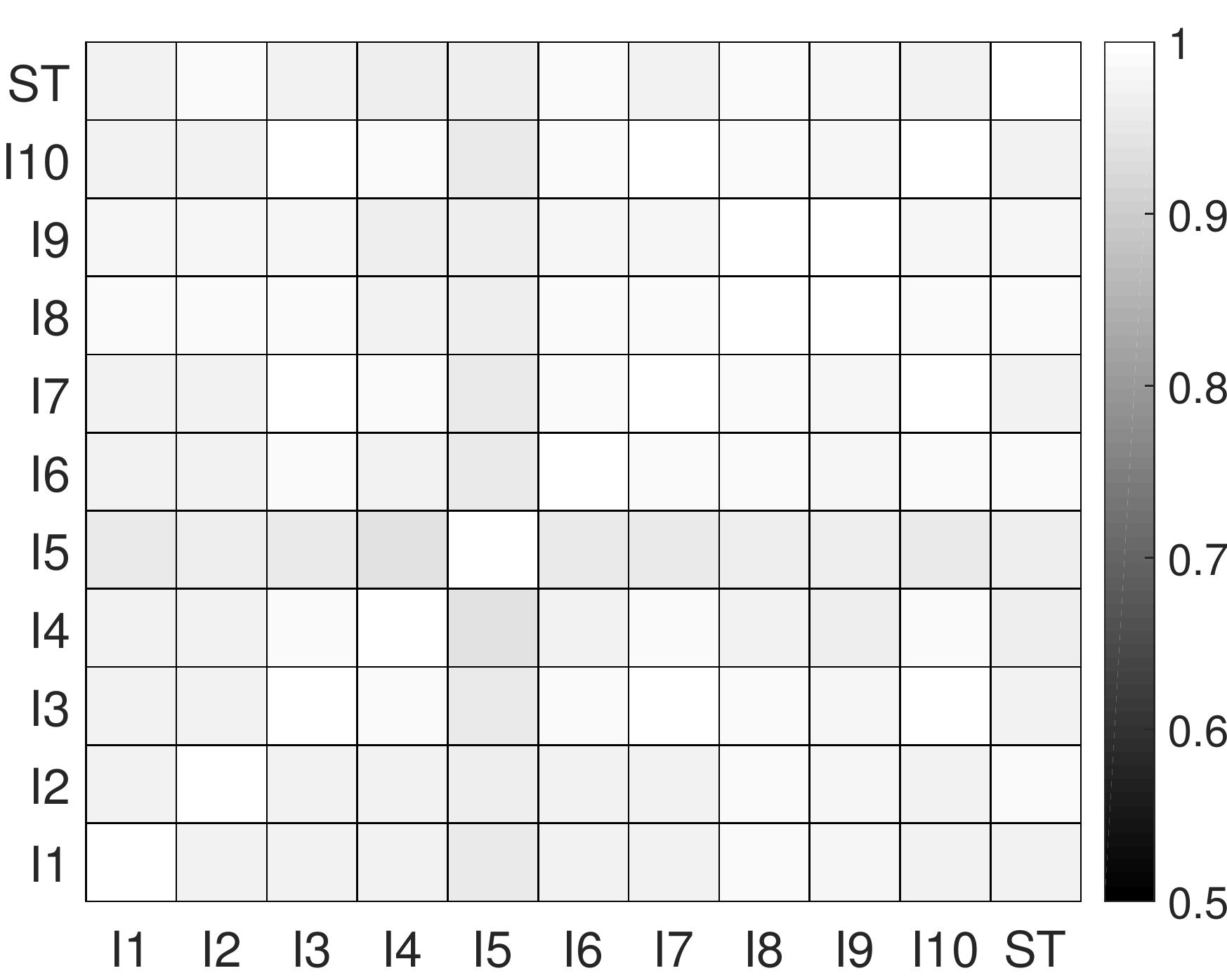}
    }
     \subfigure[\small{soc-Pokec (Thres IC)}]{
      \includegraphics[width=.22\textwidth]{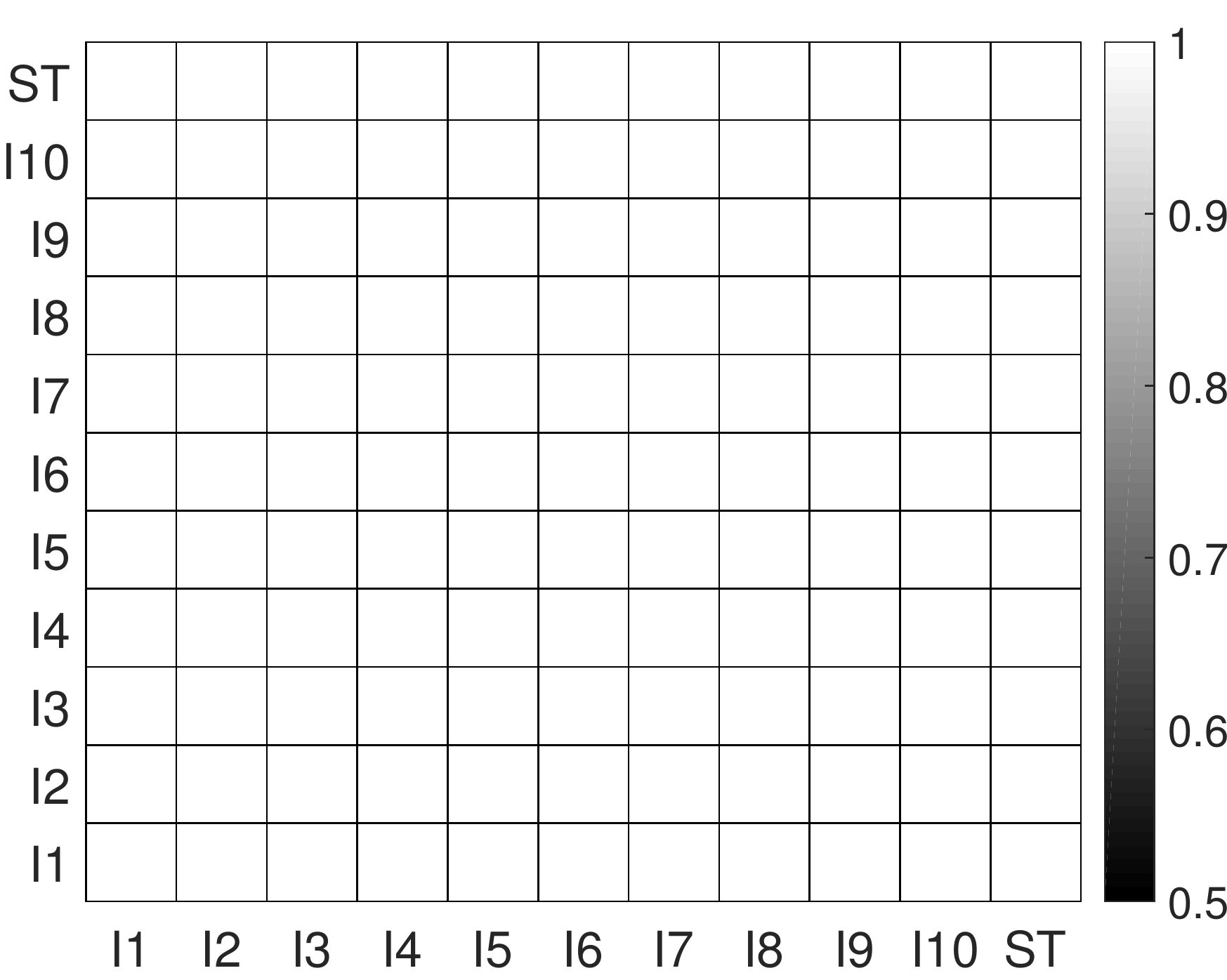}
    }
    \subfigure[\small{soc-Pokec (Top-K IC)}]{
      \includegraphics[width=.22\textwidth]{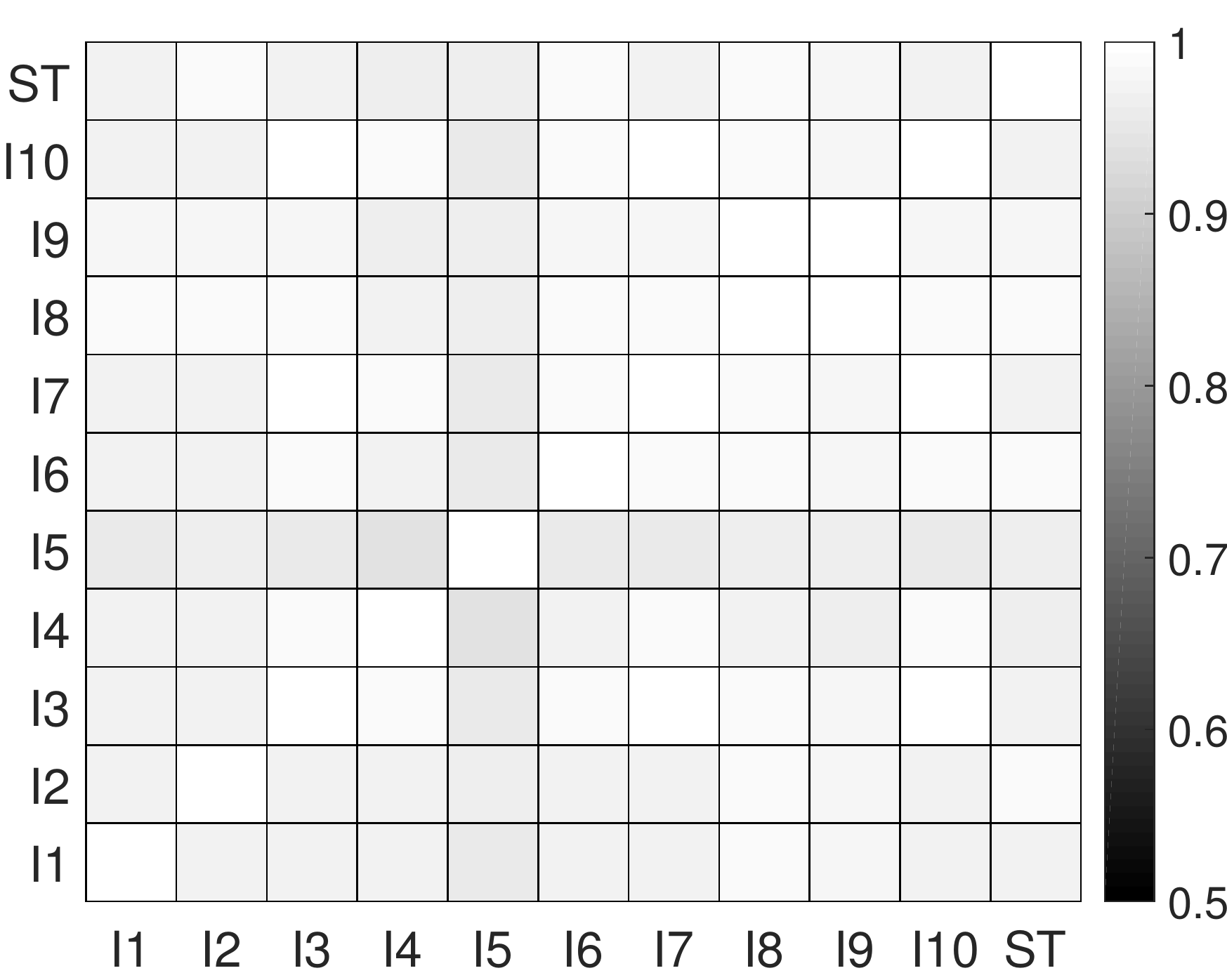}
    }
    \subfigure[\small{flickr-growth (Thres LT)}]{
      \includegraphics[width=.22\textwidth]{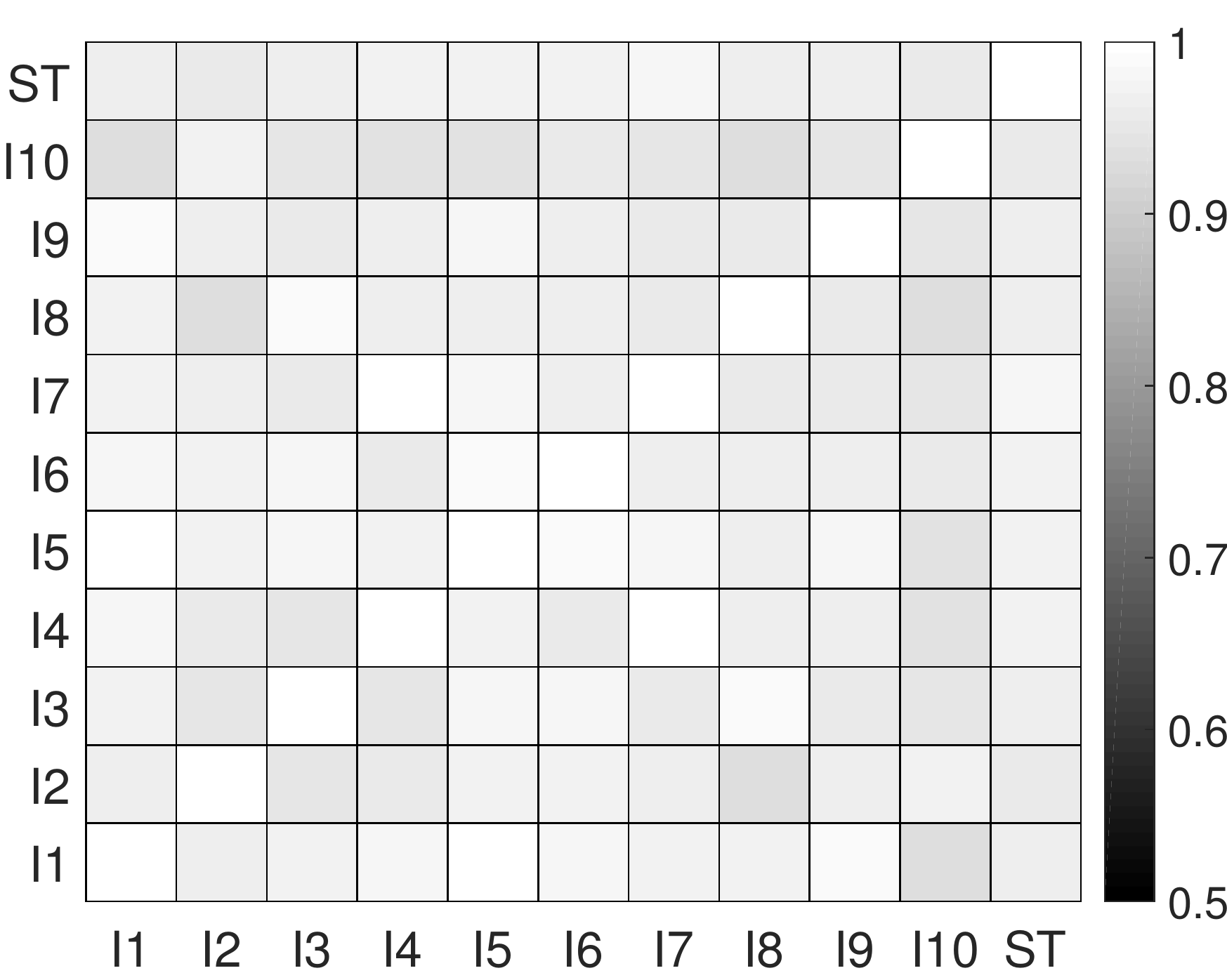}
    }
    \subfigure[\small{flickr-growth (Top-K LT)}]{
      \includegraphics[width=.22\textwidth]{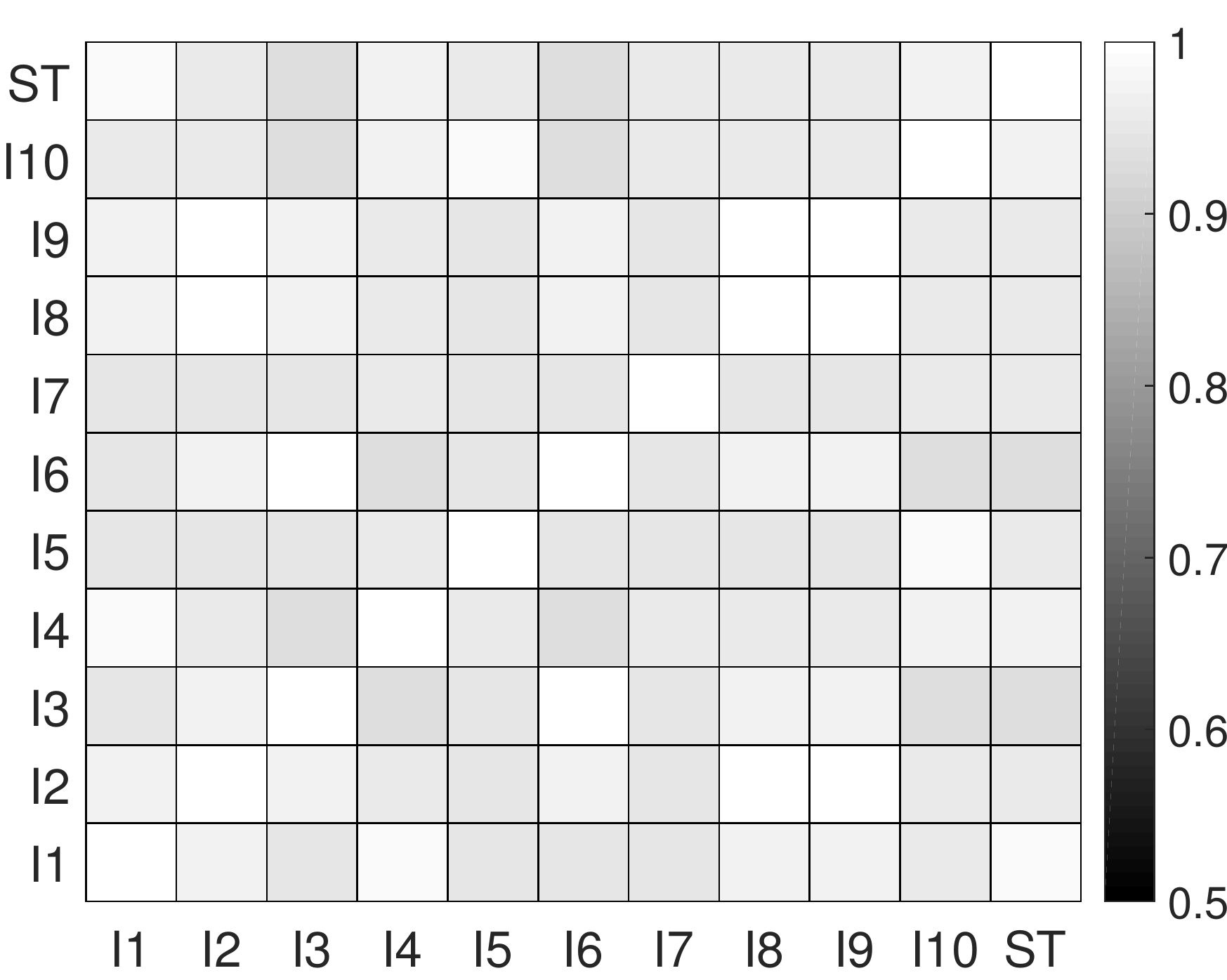}
    }
     \subfigure[\small{flickr-growth (Thres IC)}]{
      \includegraphics[width=.22\textwidth]{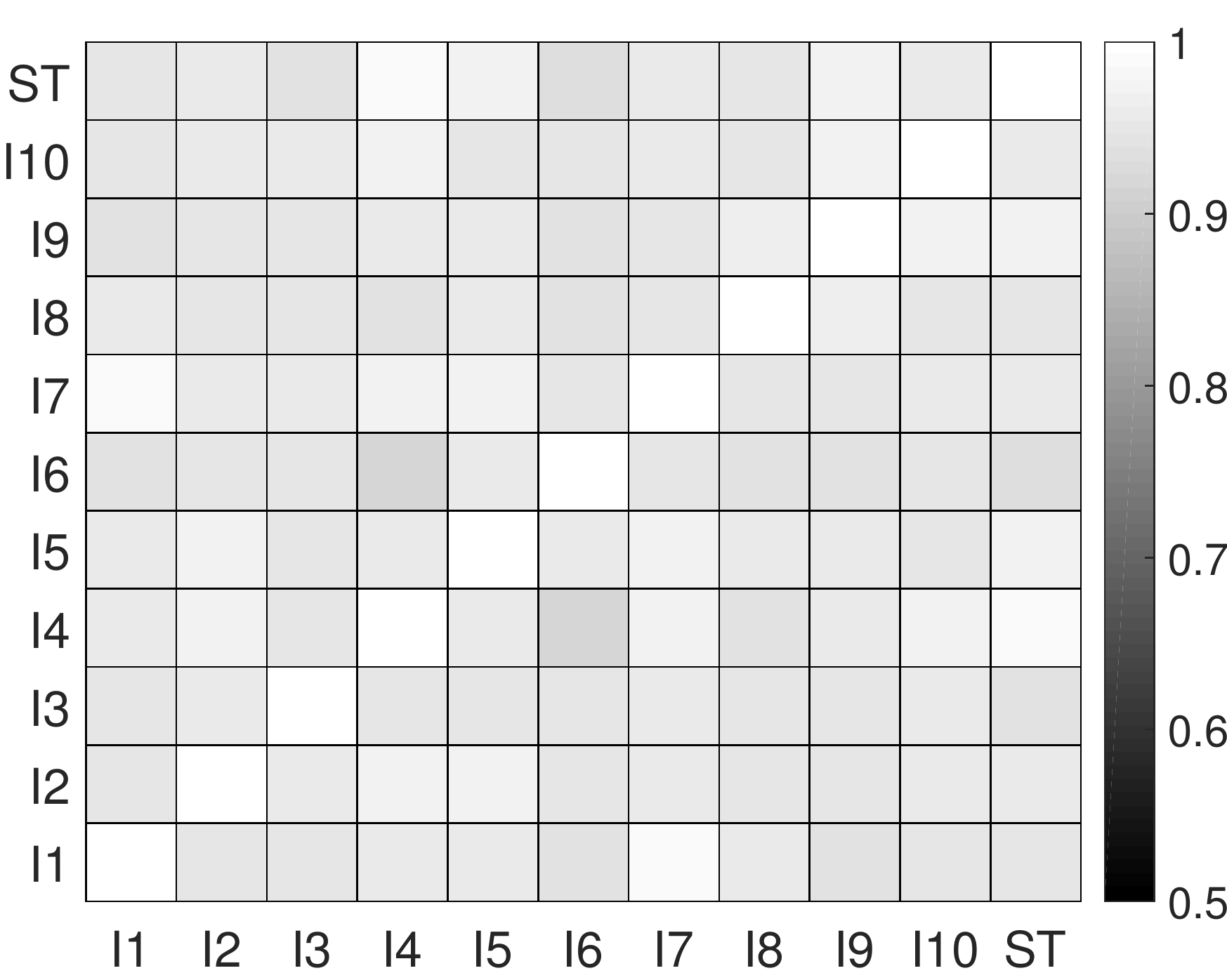}
    }
    \subfigure[\small{flickr-growth (Top-K IC)}]{
      \includegraphics[width=.22\textwidth]{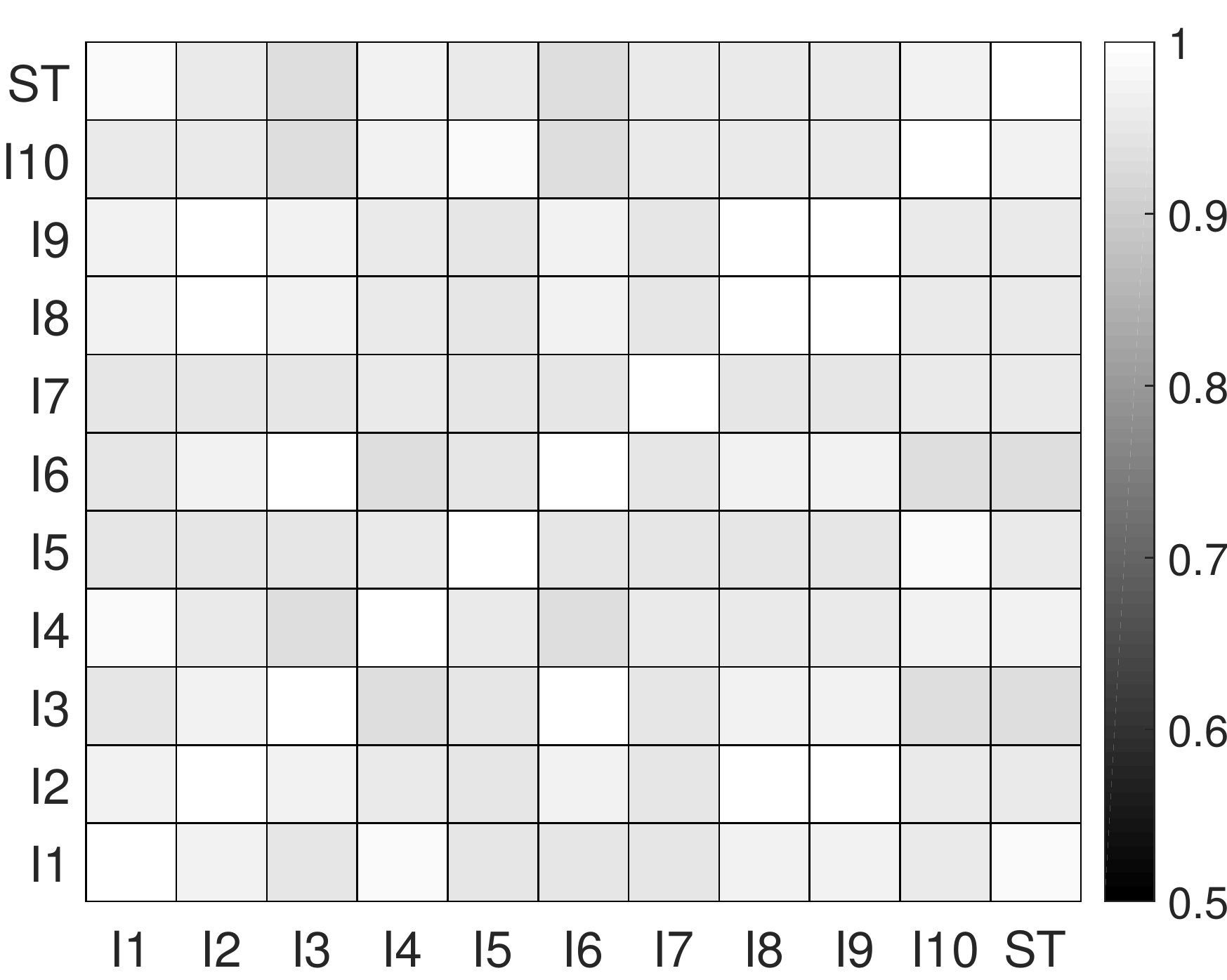}
    }
     \subfigure[\small{Twitter (Thres LT)}]{
      \includegraphics[width=.22\textwidth]{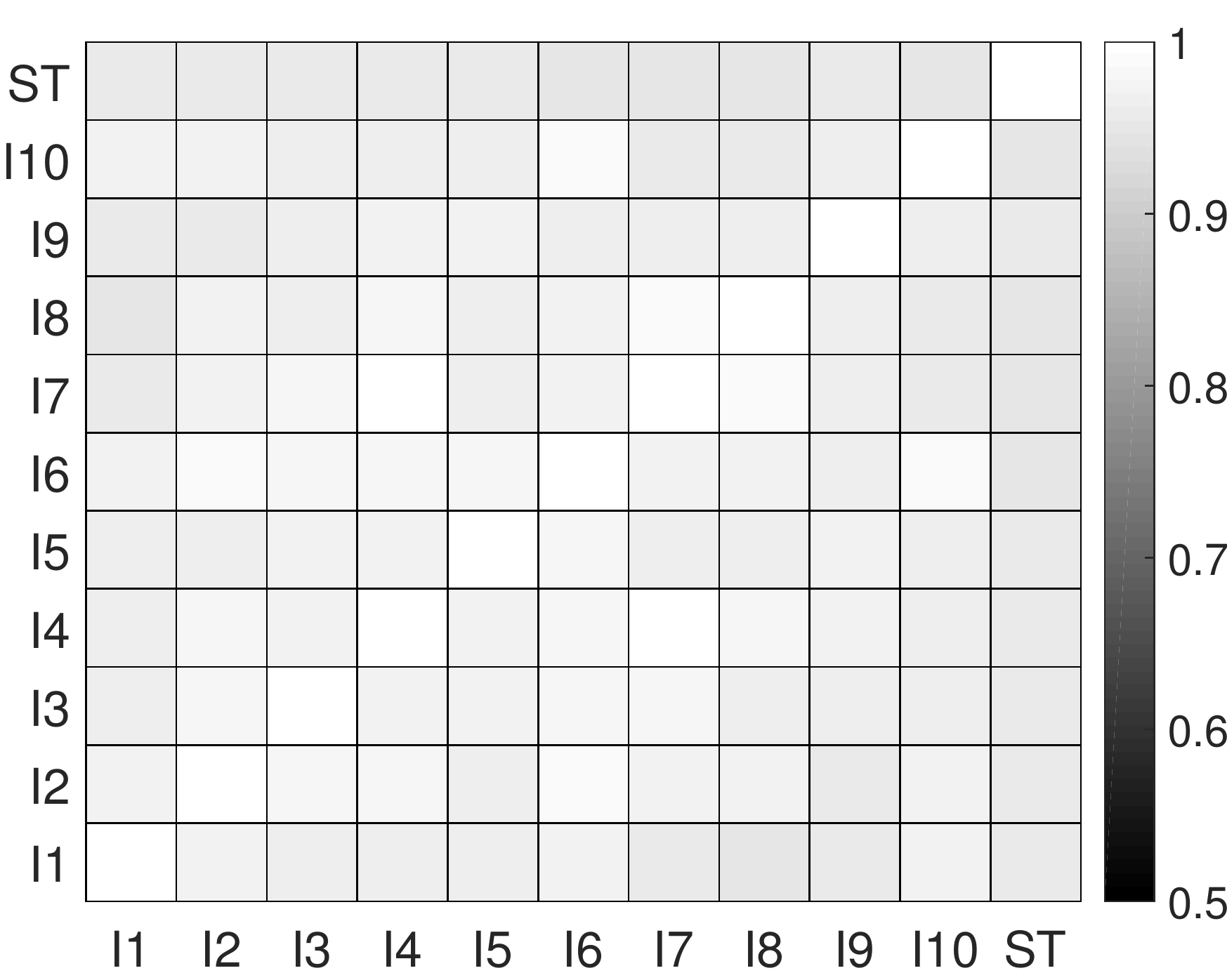}
    }
    \subfigure[\small{Twitter (Top-K LT)}]{
      \includegraphics[width=.22\textwidth]{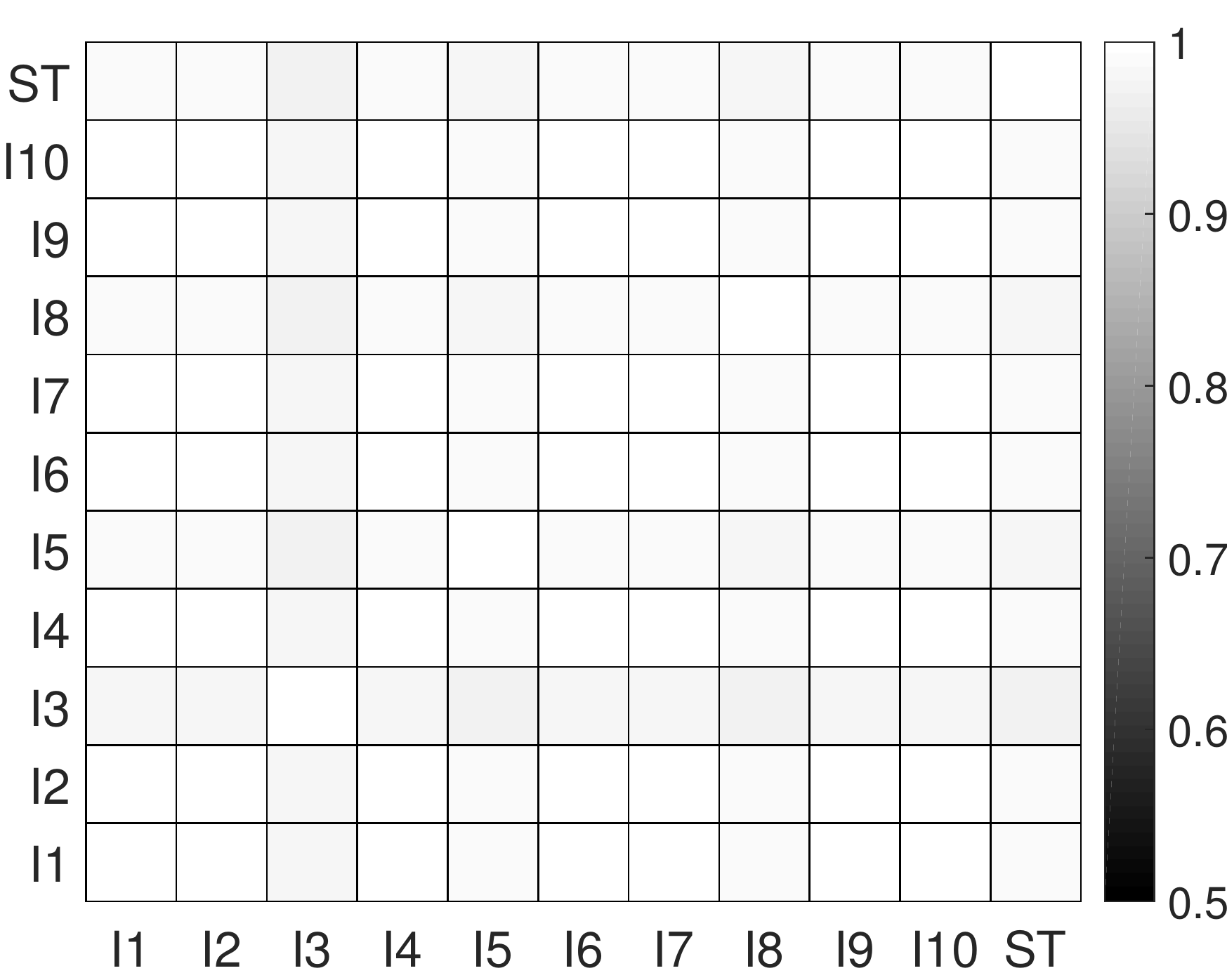}
    }
     \subfigure[\small{Twitter (Thres IC)}]{
      \includegraphics[width=.22\textwidth]{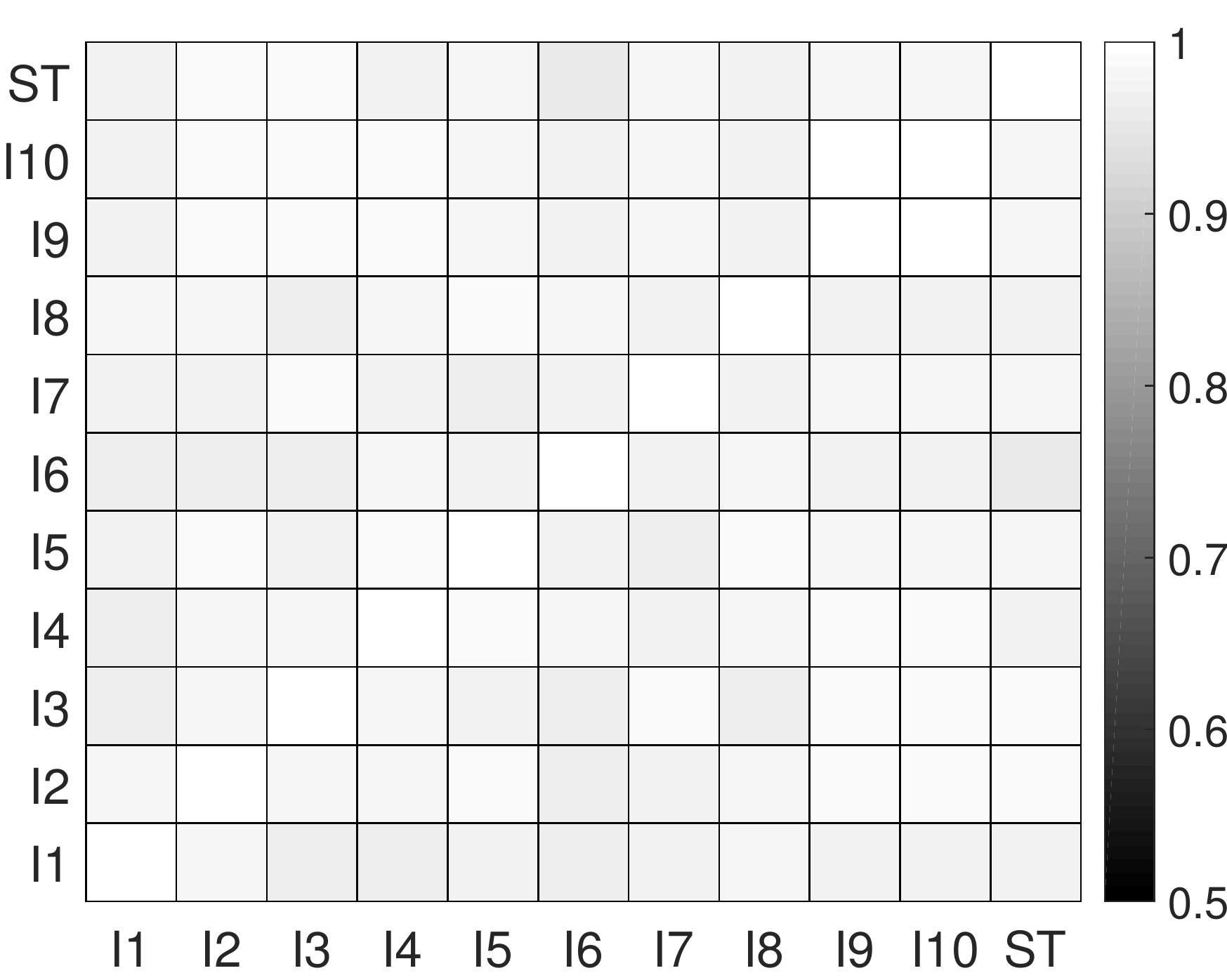}
    }
    \subfigure[\small{Twitter (Top-K IC)}]{
      \includegraphics[width=.22\textwidth]{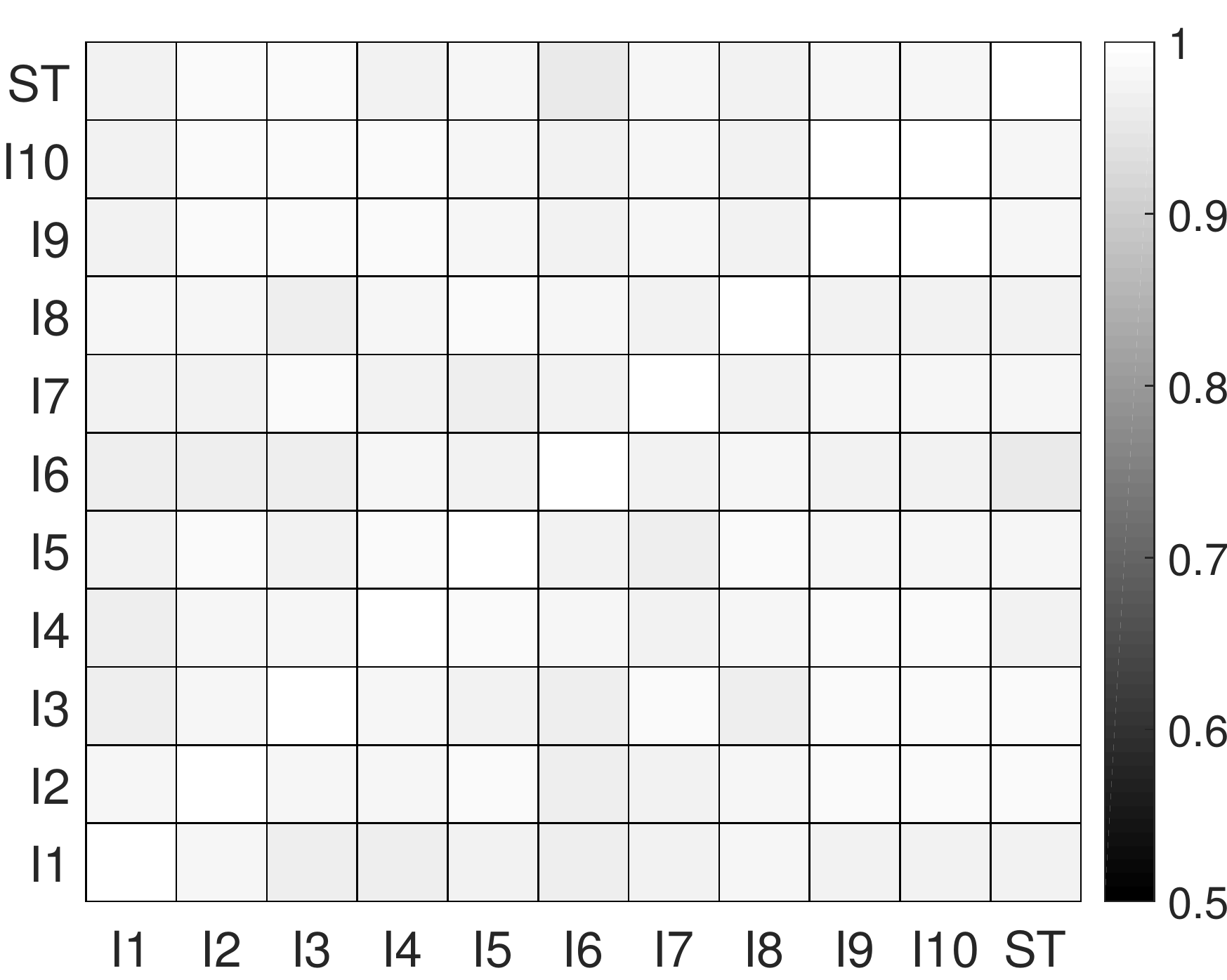}
    }
    \caption{Similarity among results in different instances.}
    \label{fig:effectiveness}
\end{figure*}

\subsubsection{Varying $\epsilon$}
For the $2$ datasets with (pseudo) ground truth, we also set the error parameter $\epsilon$ different values and report results of the Top-K tracking task. Due to limit of space, we omit results of the threshold-based tracking task and results of varying $k$ because they are all similar. In all cases the recall is always 100\%, we report maximum errors of nodes returned by our algorithm in different settings of $\epsilon$ in Fig.~\ref{fig:eps}. The maximum error is constantly smaller than the theoretical value, and it increases roughly linearly as $\epsilon$ increases. Moreover, the theoretical value increases faster than the maximum error.

\begin{figure}[h]
  \centering
  \caption{$\epsilon$ v.s. Maximum Error.}\label{fig:eps}
    \subfigure[wiki-Vote]{
      \includegraphics[width=.22\textwidth]{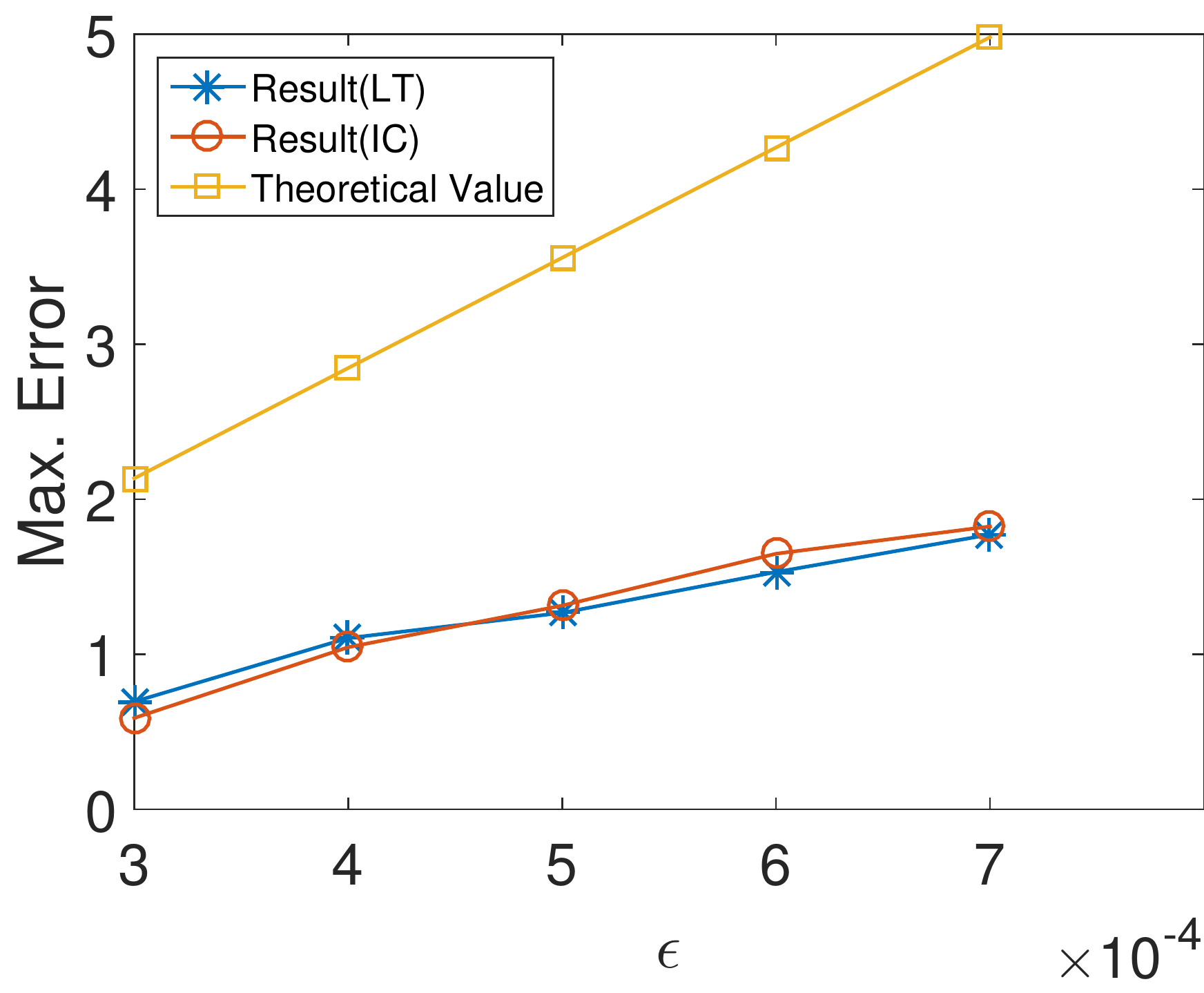}
    }
    \subfigure[Flixster]{
      \includegraphics[width=.22\textwidth]{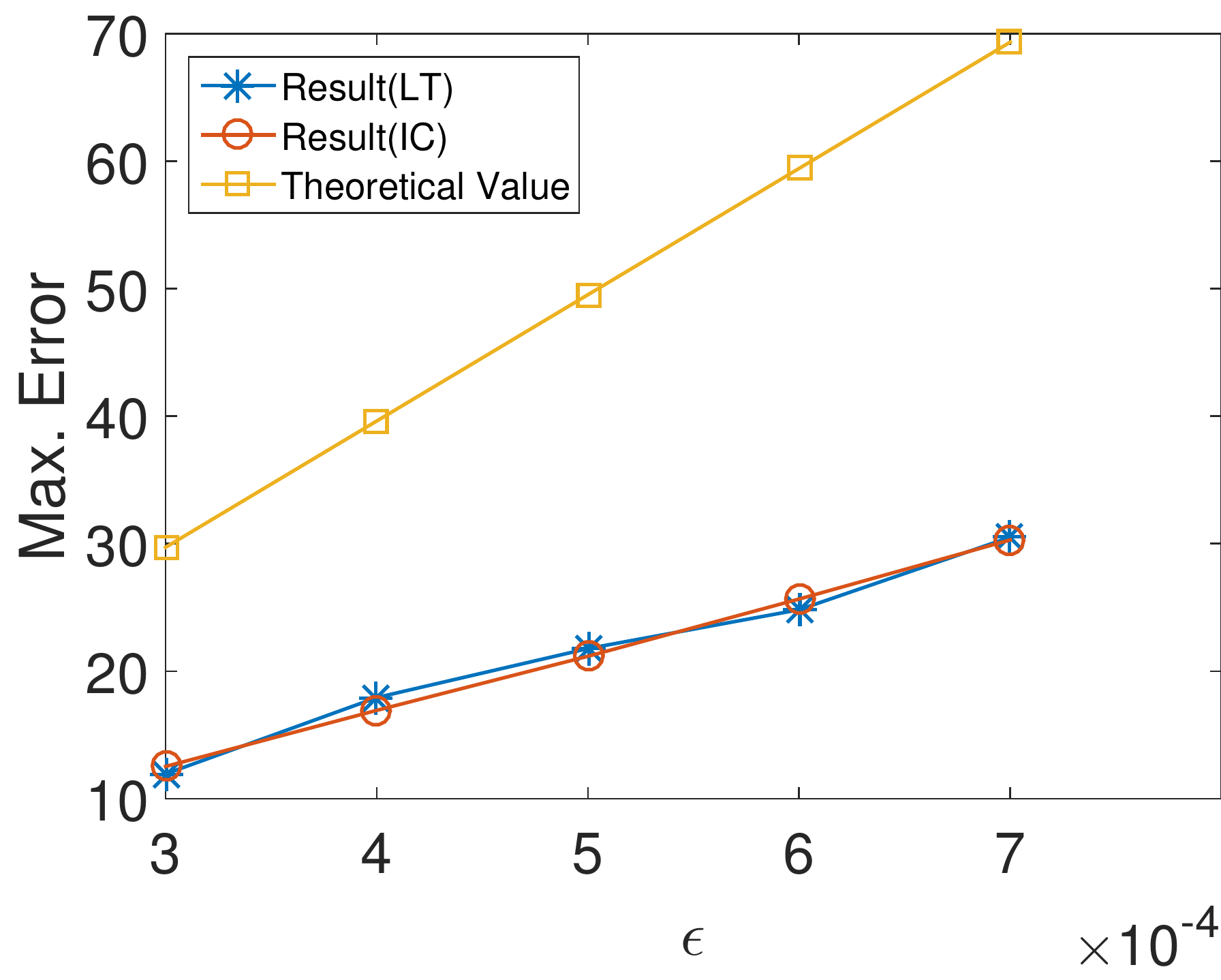}
    }
\end{figure} 

\subsubsection{Comparing with Simple Heuristics}

We also compare our algorithms with two simple heuristics, degree and PageRank, which simply return top ranked nodes by degree or PageRank values as influential nodes. The reason we choose these two heuristics is that they both can be efficiently implemented in the setting of dynamic networks. Note that these two heuristics cannot solve the threshold based influential nodes mining problem because they do not know the influence spread of each node.

\begin{figure}[h]
  \centering
    \subfigure[\small{wiki-Vote (LT)}]{
      \includegraphics[width=.22\textwidth]{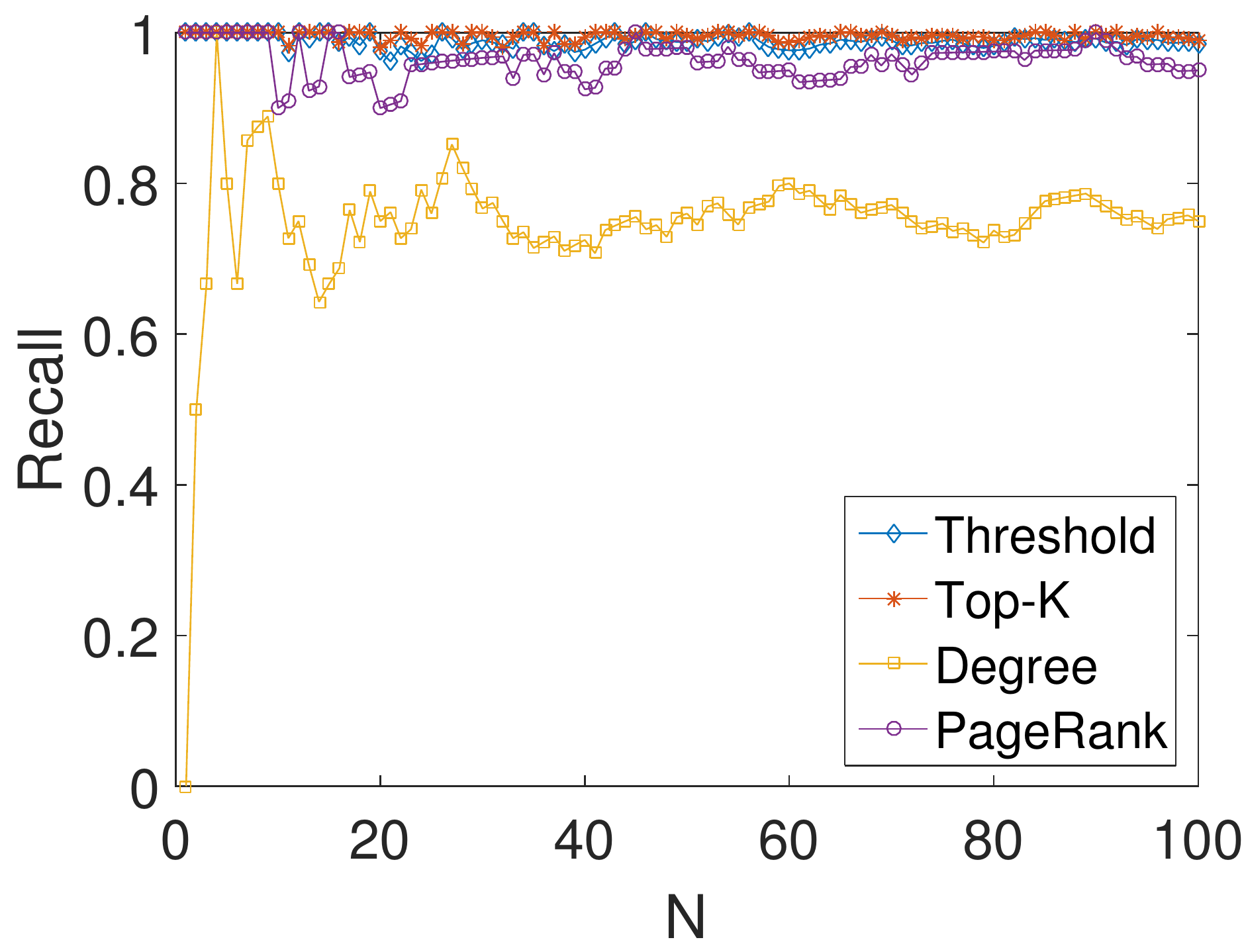}
    }
    \subfigure[\small{wiki-Vote (IC)}]{
      \includegraphics[width=.22\textwidth]{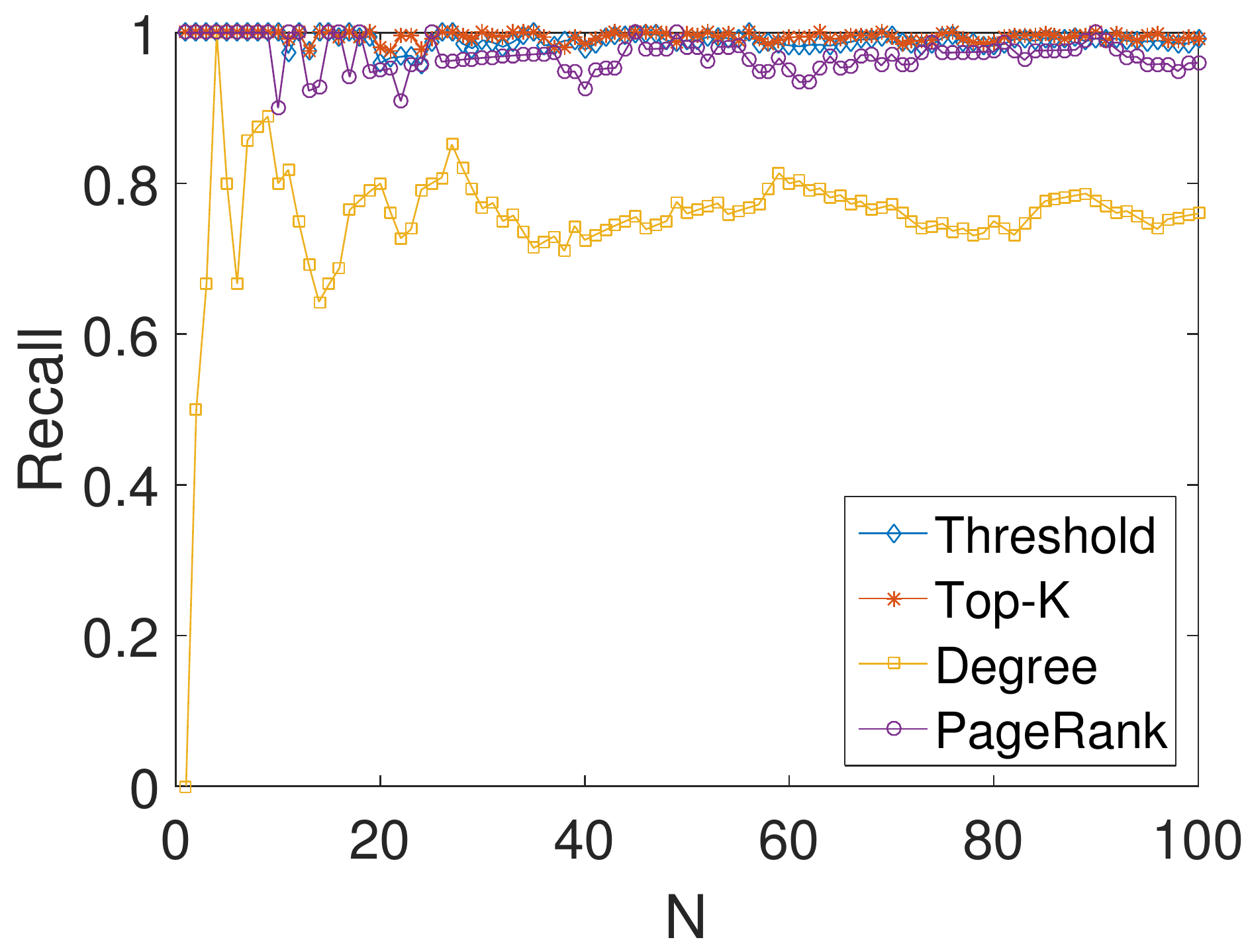}
    }
    
    \subfigure[\small{Flixster (LT)}]{
      \includegraphics[width=.22\textwidth]{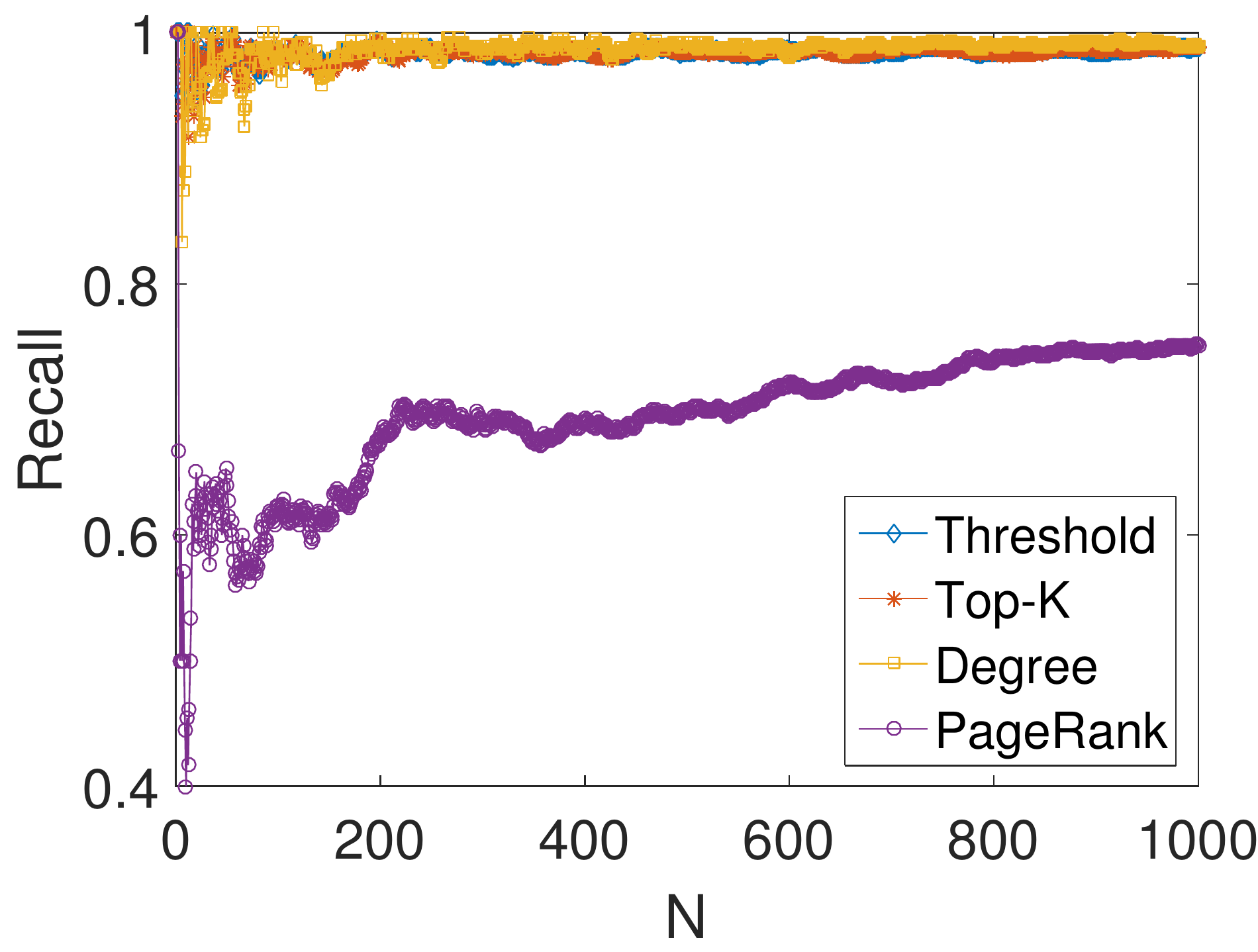}
    }
    \subfigure[\small{Flixster (IC)}]{
      \includegraphics[width=.22\textwidth]{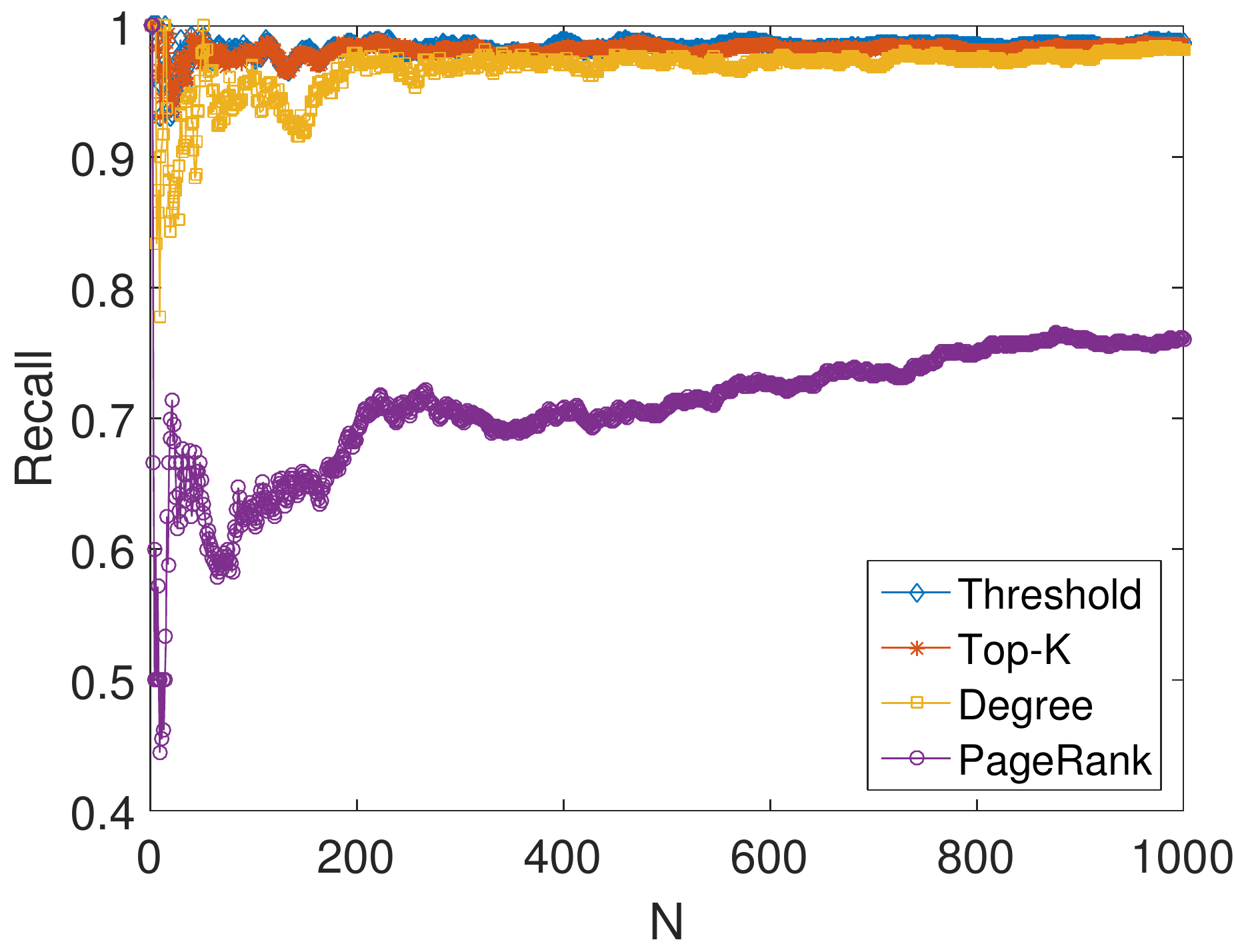}
    }
    \caption{$Recall@N$}
    \label{fig:recall}
\end{figure} 

To compare our algorithms with degree and PageRank heuristics, we report the recall of the top ranked nodes obtained by each method on wiki-Vote and Flixster data sets in Fig~\ref{fig:recall}. Nodes ranking by 2,000,000 times Monte Carlo simulations is regarded as the (pseudo) ground truth. The measure $Recall@N$ is calculated by $\frac{TP_N}{N}$, where $TP_N$ is the number of nodes ranked top-$N$ by both our algorithms and the ground truth. The results show that the rankings of the top nodes generated by our algorithms constantly have very good quality, while the two heuristics sometimes perform well but sometimes return really poor rankings. Moreover, performance of a heuristic algorithm is not predictable. 

\subsection{Scalability}

\subsubsection{Running Time with respect to Number of Updates}
We also tested the scalability of our algorithms. Fig.~\ref{fig:efficiency} shows the average running time with respect to the number of updates processed. The average is taken on the running times of the 10 instances. The time spent when the number of updates is 0 reflects the computational cost of running the sampling algorithm on the base network. In table~\ref{tab:time_static}, we also report running time of algorithms on the static final snapshot of each dataset. Clearly, the non-incremental algorithm (rerunning the sampling algorithm from scratch when the network changes) is not competent at all because the running time of processing the base network and the update stream is only several times larger than the running time of processing the whole network, and the number of updates is huge, tens of thousands or even hundreds of millions. This result shows that our incremental algorithm outperforms rerunning the sampling algorithm from scratch by several orders of magnitude. 

\begin{figure}[h]
  \centering
    \subfigure[$I_{max}$]{
      \includegraphics[width=.22\textwidth]{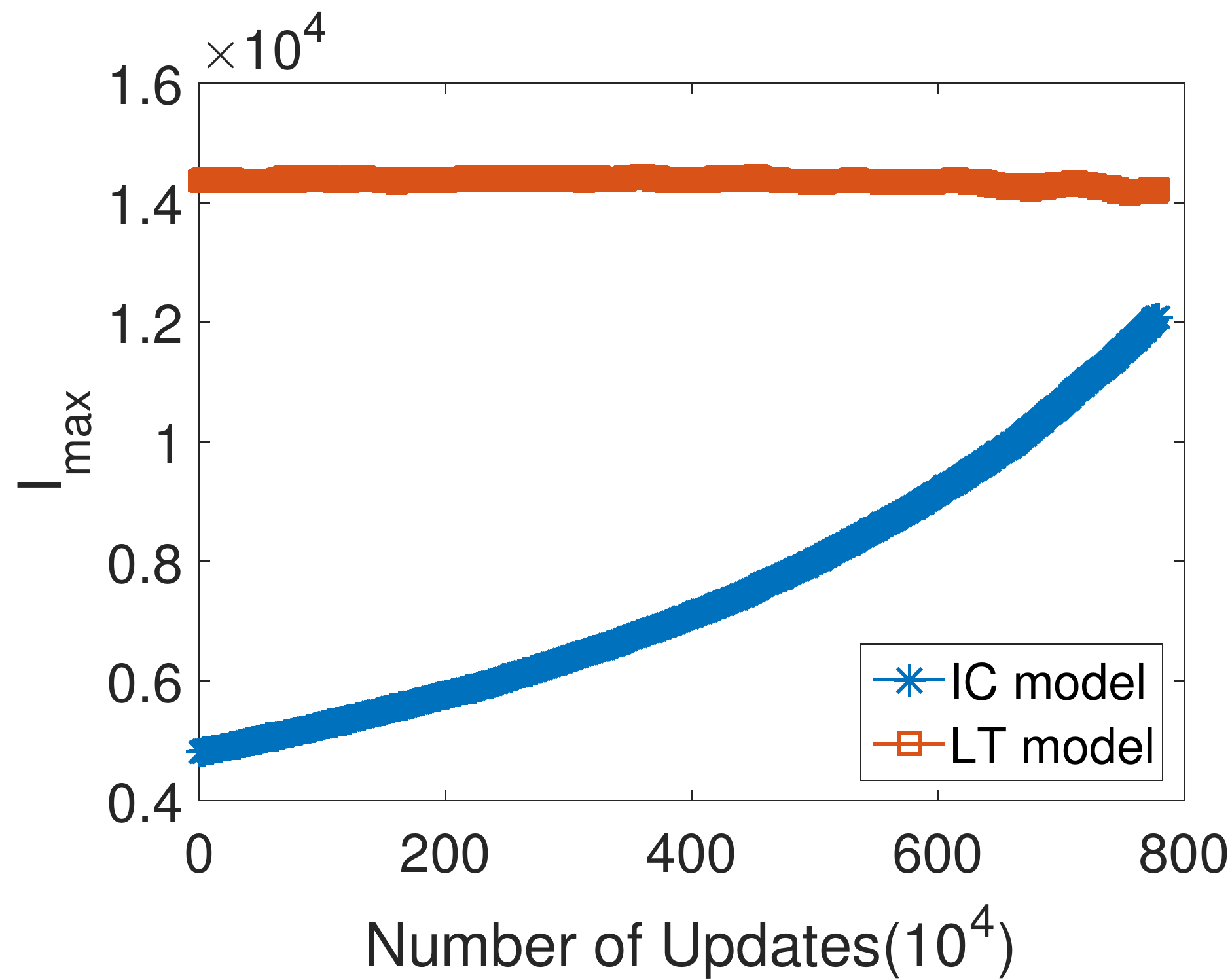}
    }
    \subfigure[Sample size $M$]{
      \includegraphics[width=.22\textwidth]{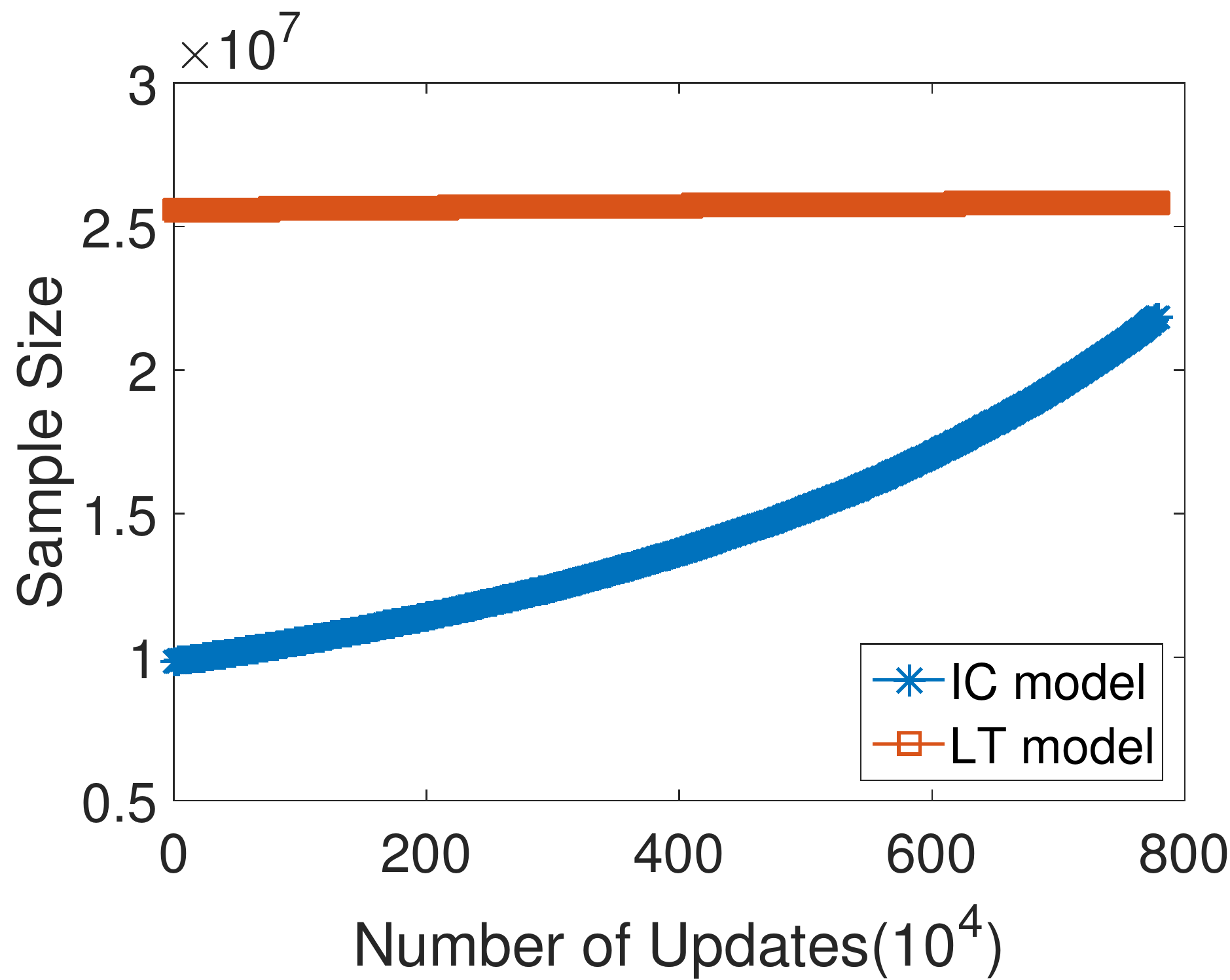}
    }
    \caption{$I_{max}$ and $M$ change over time of flickr-growth data.}
    \label{fig:change}
\end{figure} 

\begin{figure}[h]
  \centering
    \subfigure[\small{Graph Size}]{
      \includegraphics[width=.2\textwidth]{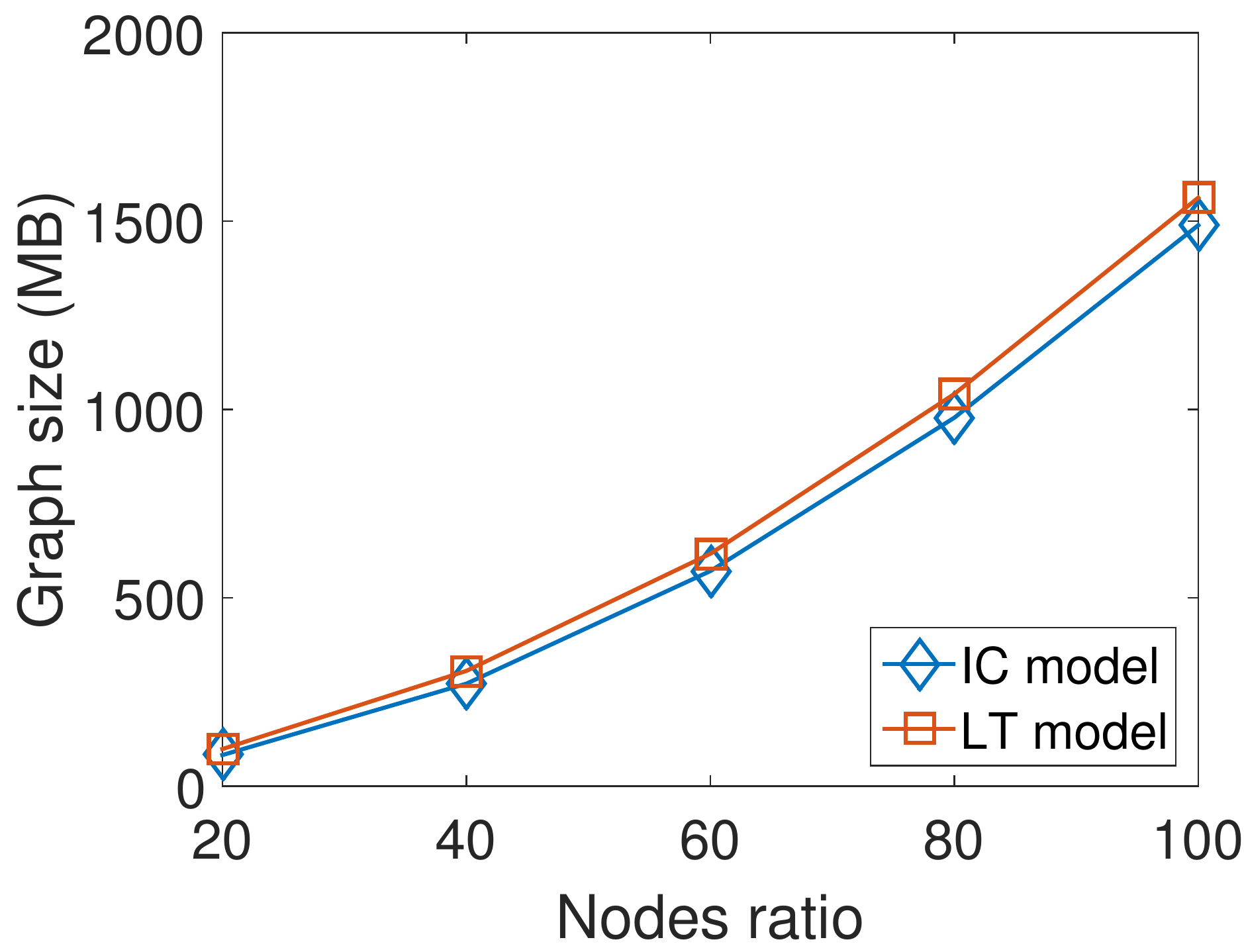}
    }
    \subfigure[\small{RR sets Size}]{
      \includegraphics[width=.2\textwidth]{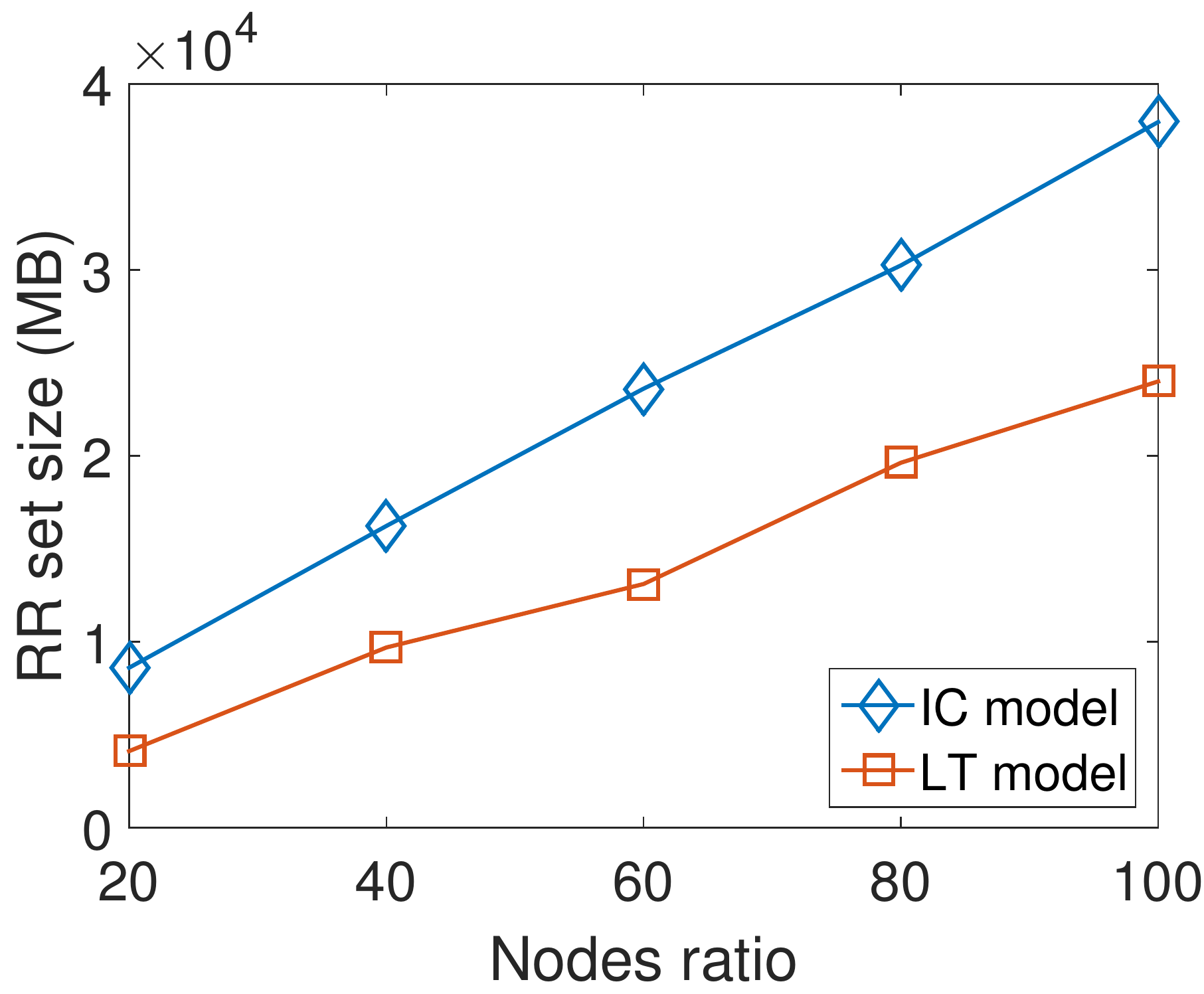}
    }  
    \subfigure[\small{RR w.r.t Graph}]{
      \includegraphics[width=.2\textwidth]{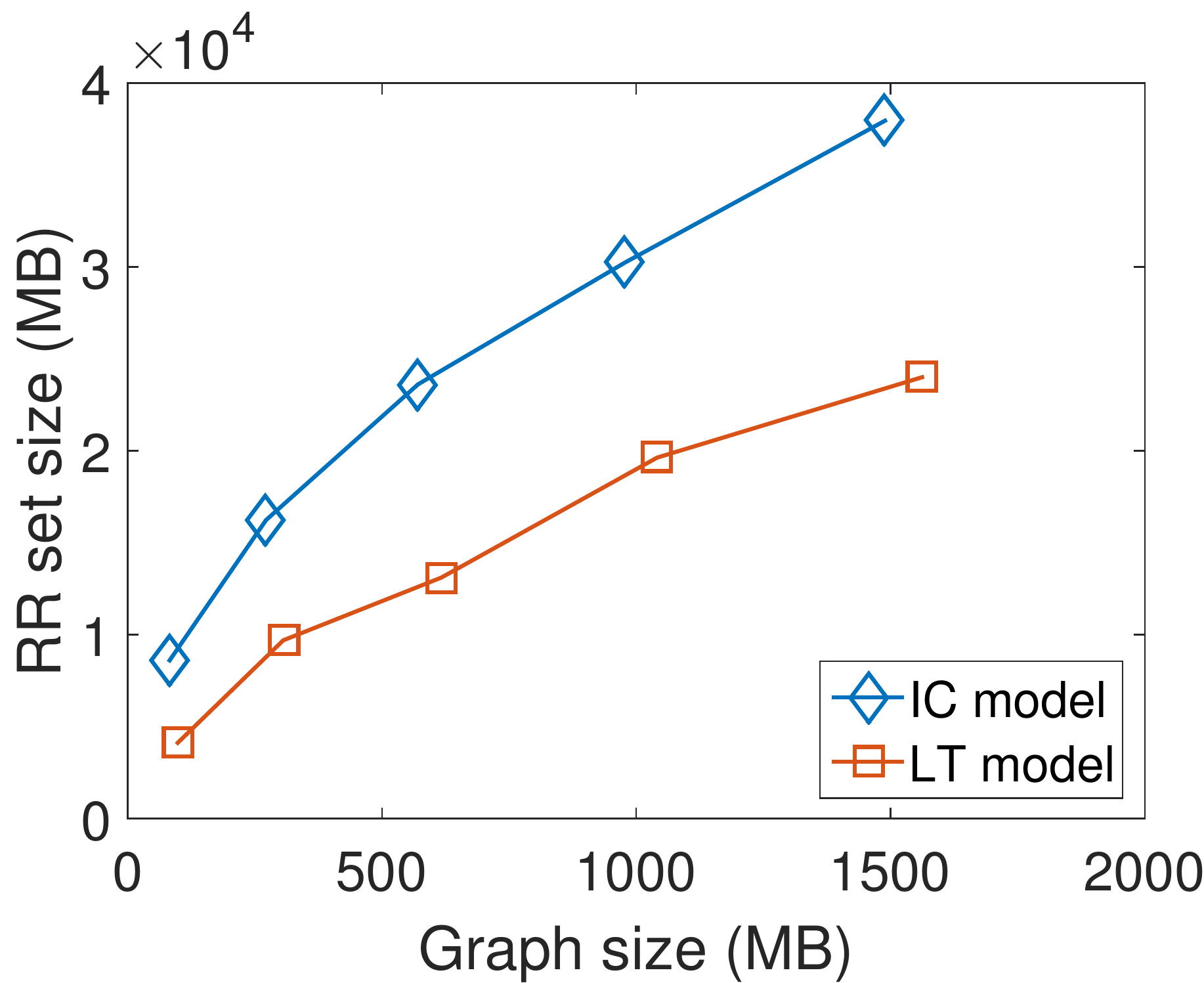}
    }
    \caption{Memory Usage}
    \label{fig:memory}
\end{figure}

\begin{figure}[h]
  \centering
    \subfigure[wiki-Vote (Threshold)]{
      \includegraphics[width=.22\textwidth]{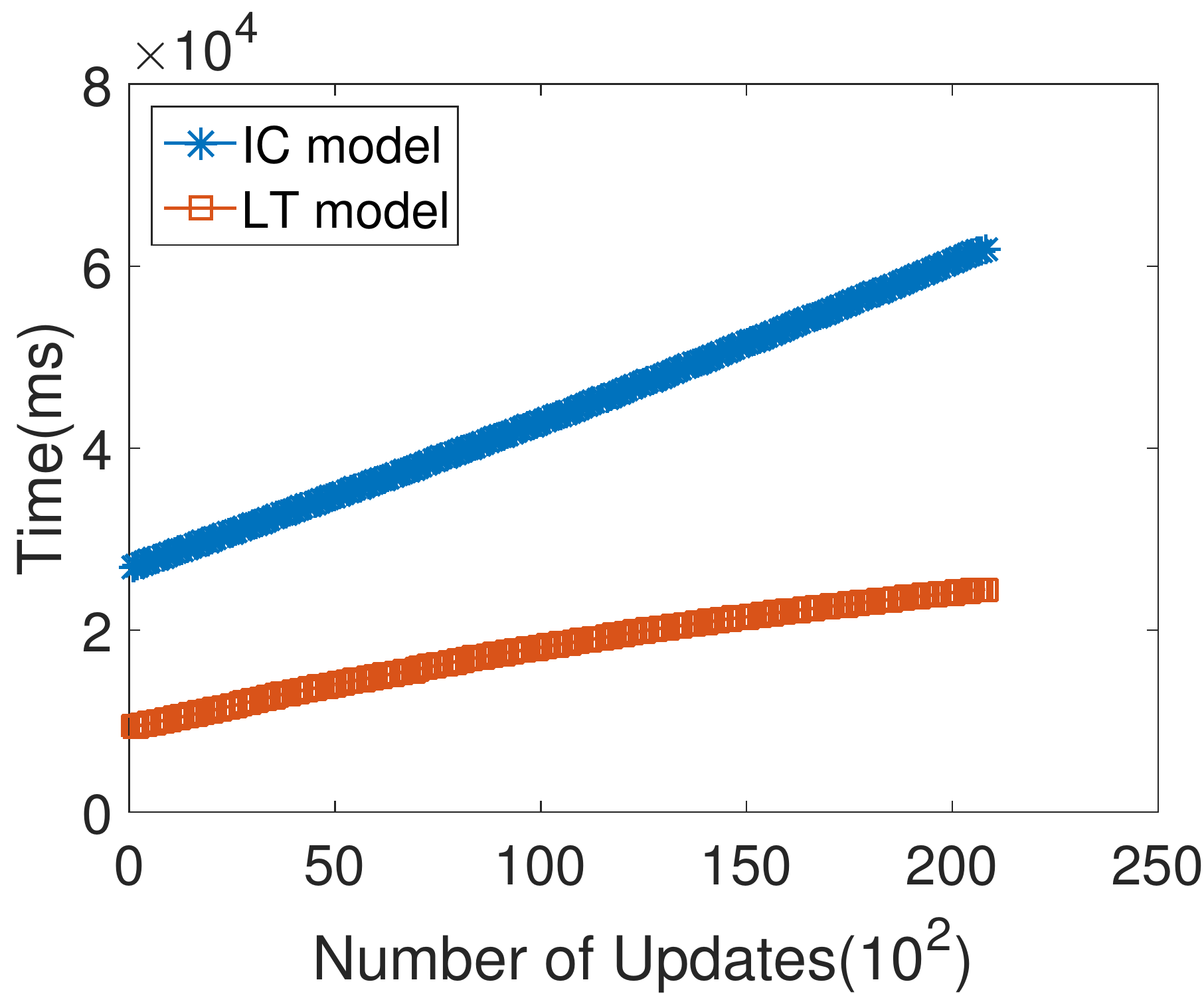}
    }
    \subfigure[wiki-Vote (Top-K)]{
      \includegraphics[width=.22\textwidth]{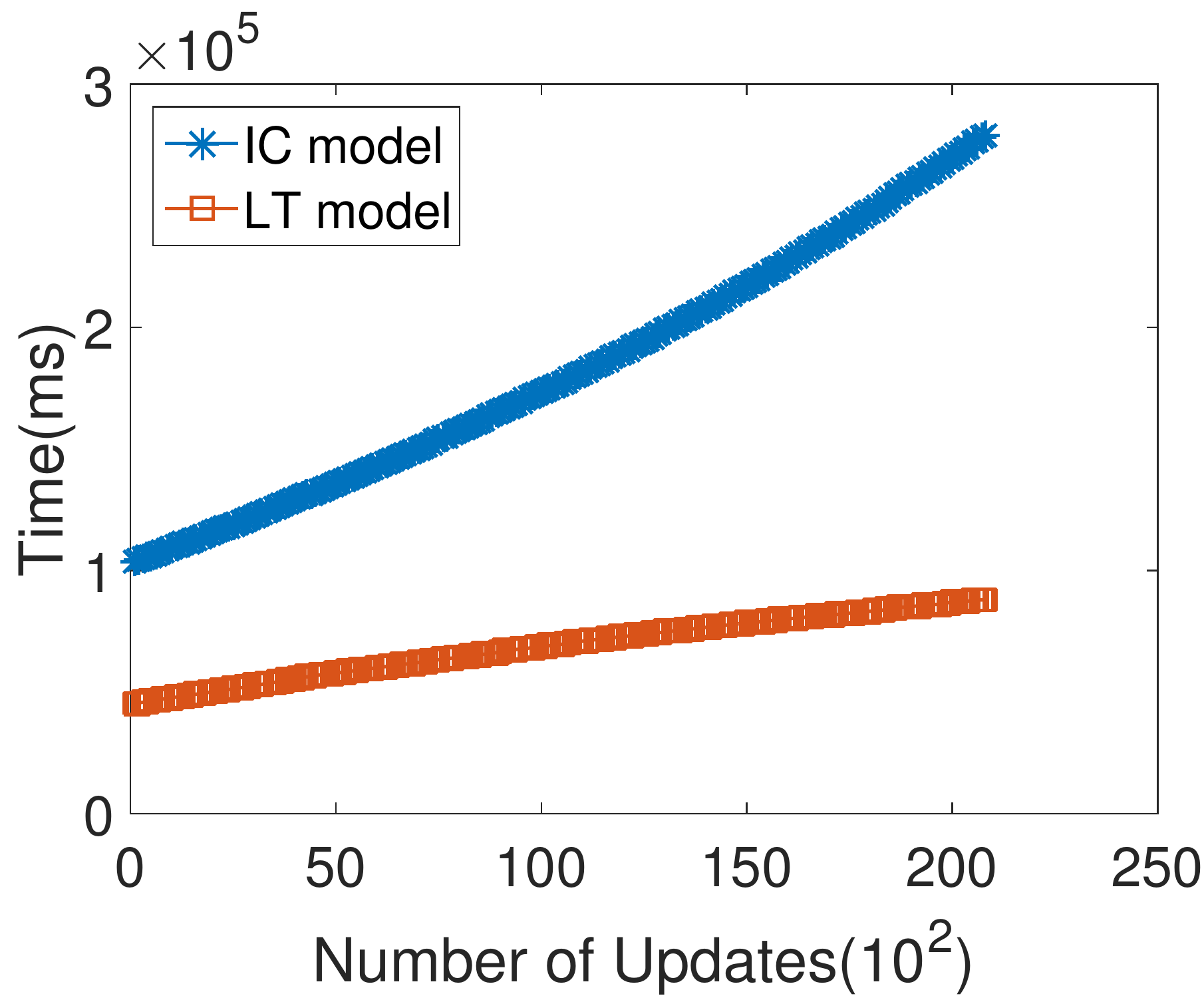}
    }
    \subfigure[Flixster (Threshold)]{
      \includegraphics[width=.22\textwidth]{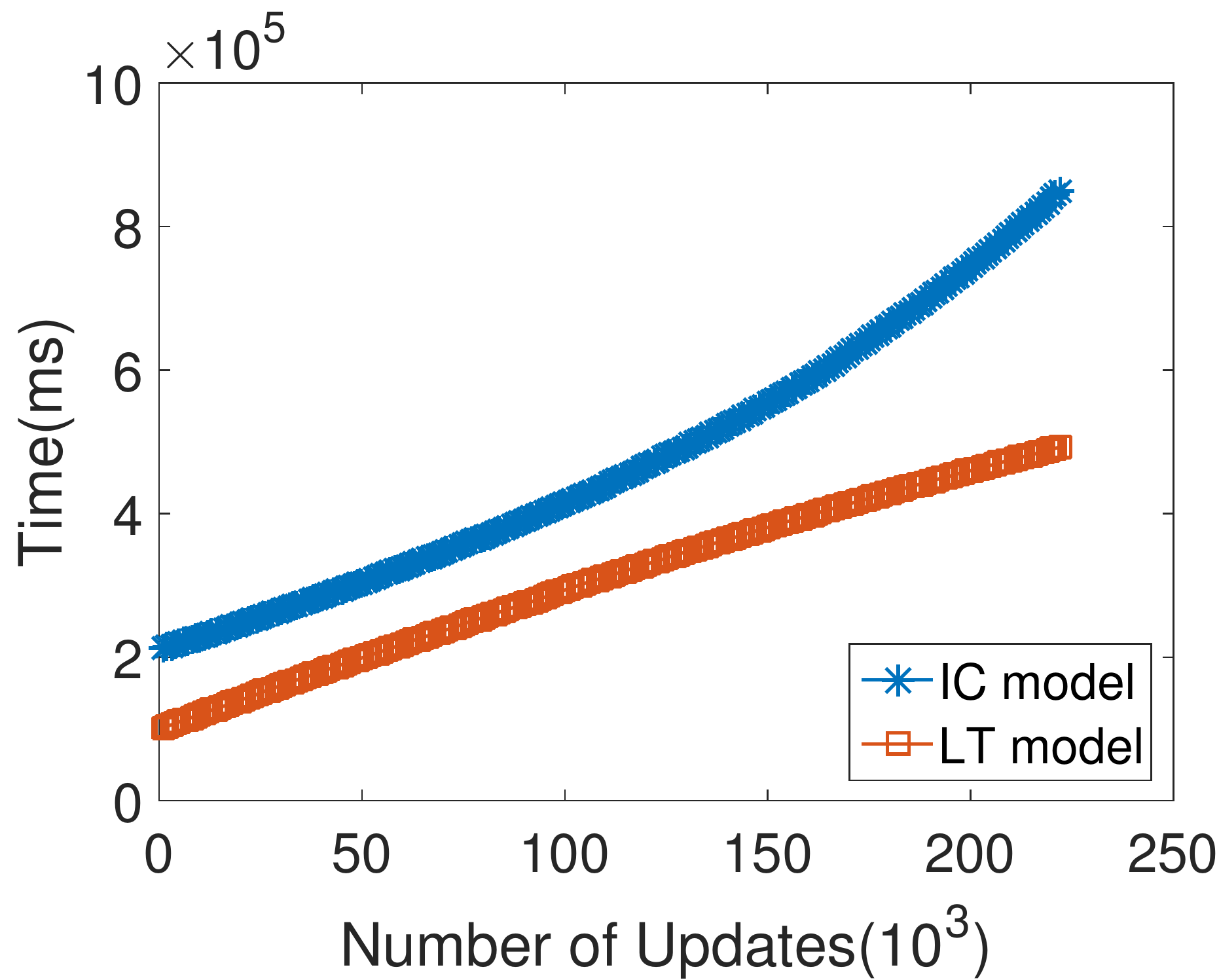}
    }
    \subfigure[Flixster (Top-K)]{
      \includegraphics[width=.22\textwidth]{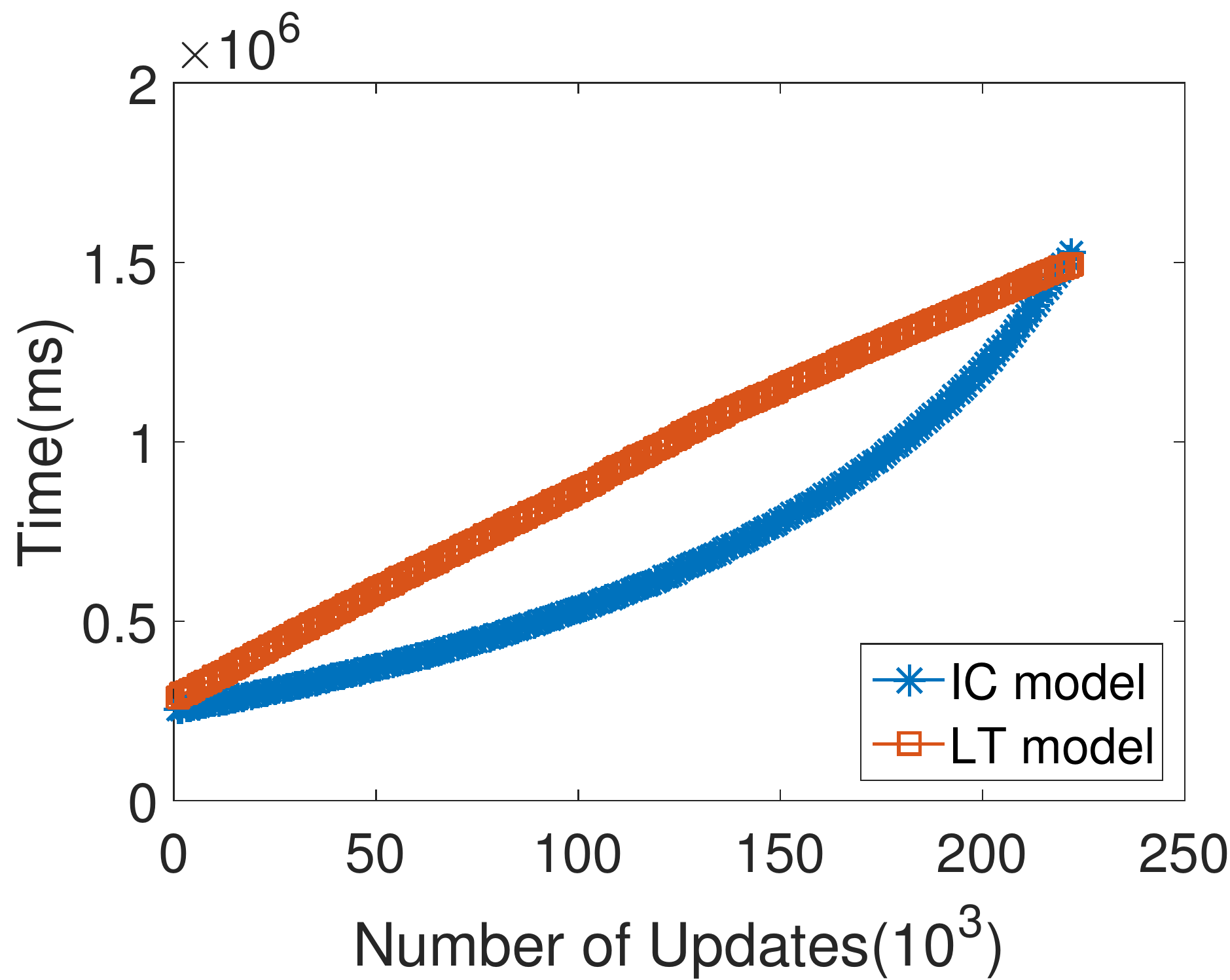}
    }
    \subfigure[soc-Pokec (Threshold)]{
      \includegraphics[width=.22\textwidth]{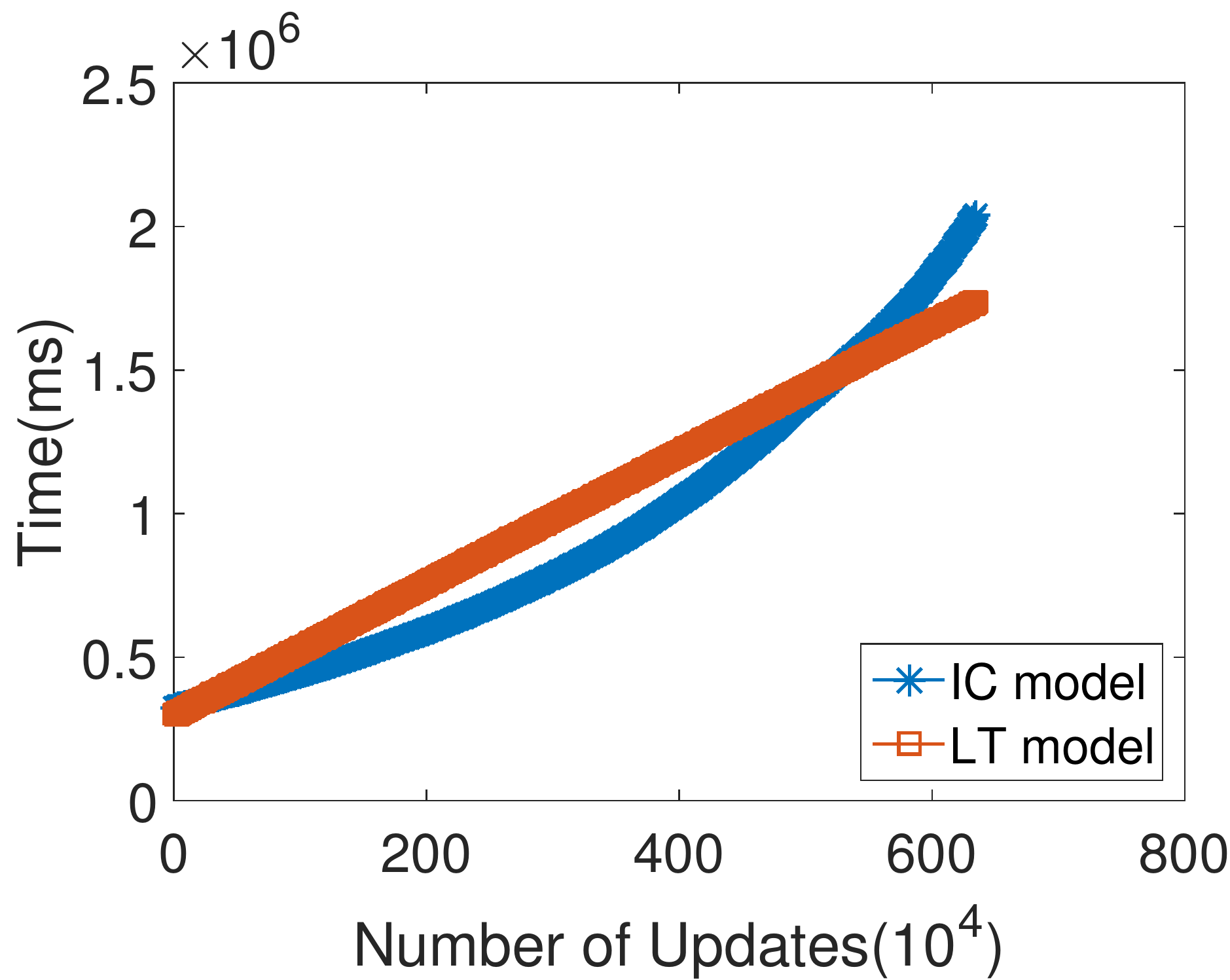}
    }
    \subfigure[soc-Pokec (Top-K)]{
      \includegraphics[width=.22\textwidth]{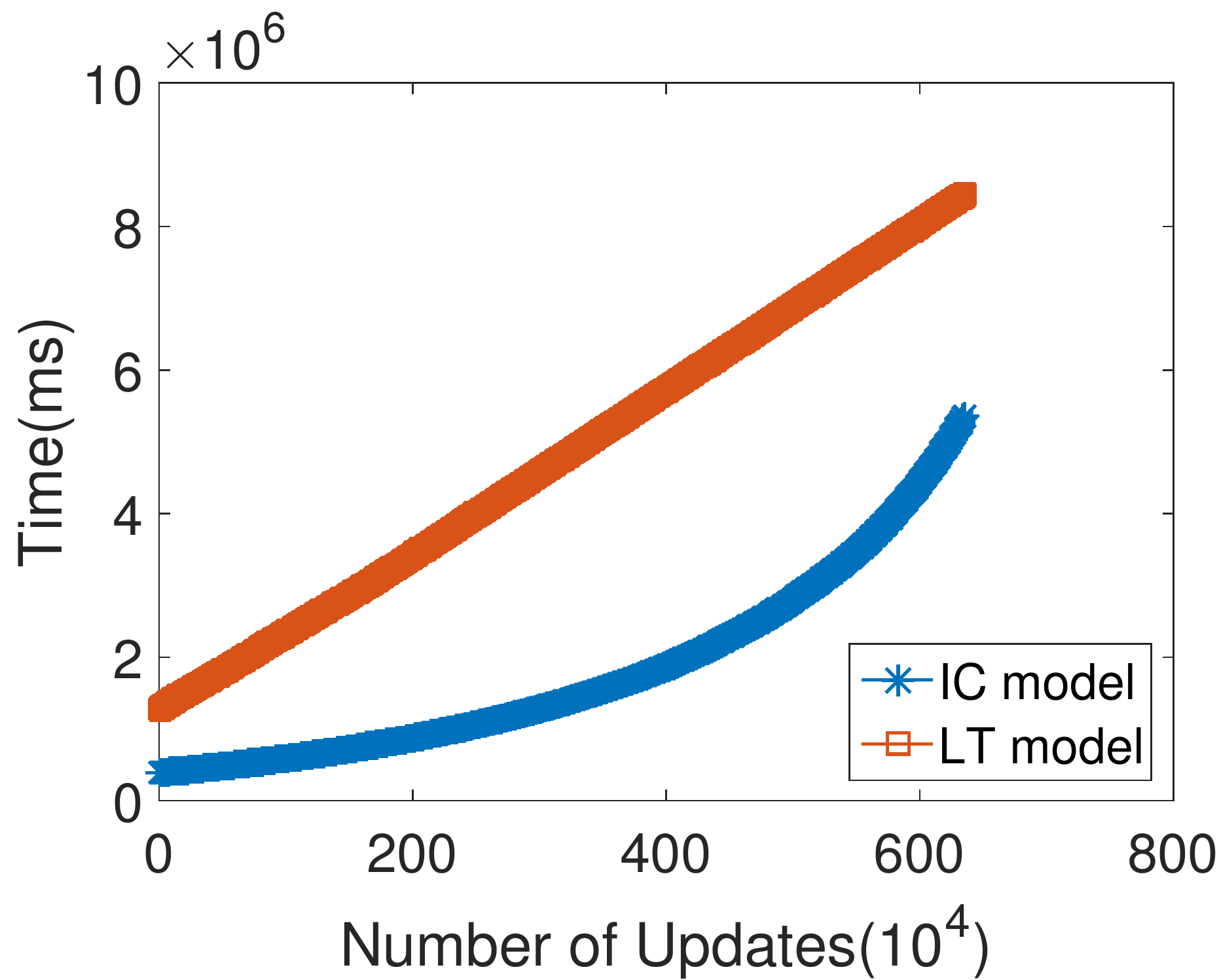}
    }
    \subfigure[flickr-growth (Threshold)]{
      \includegraphics[width=.22\textwidth]{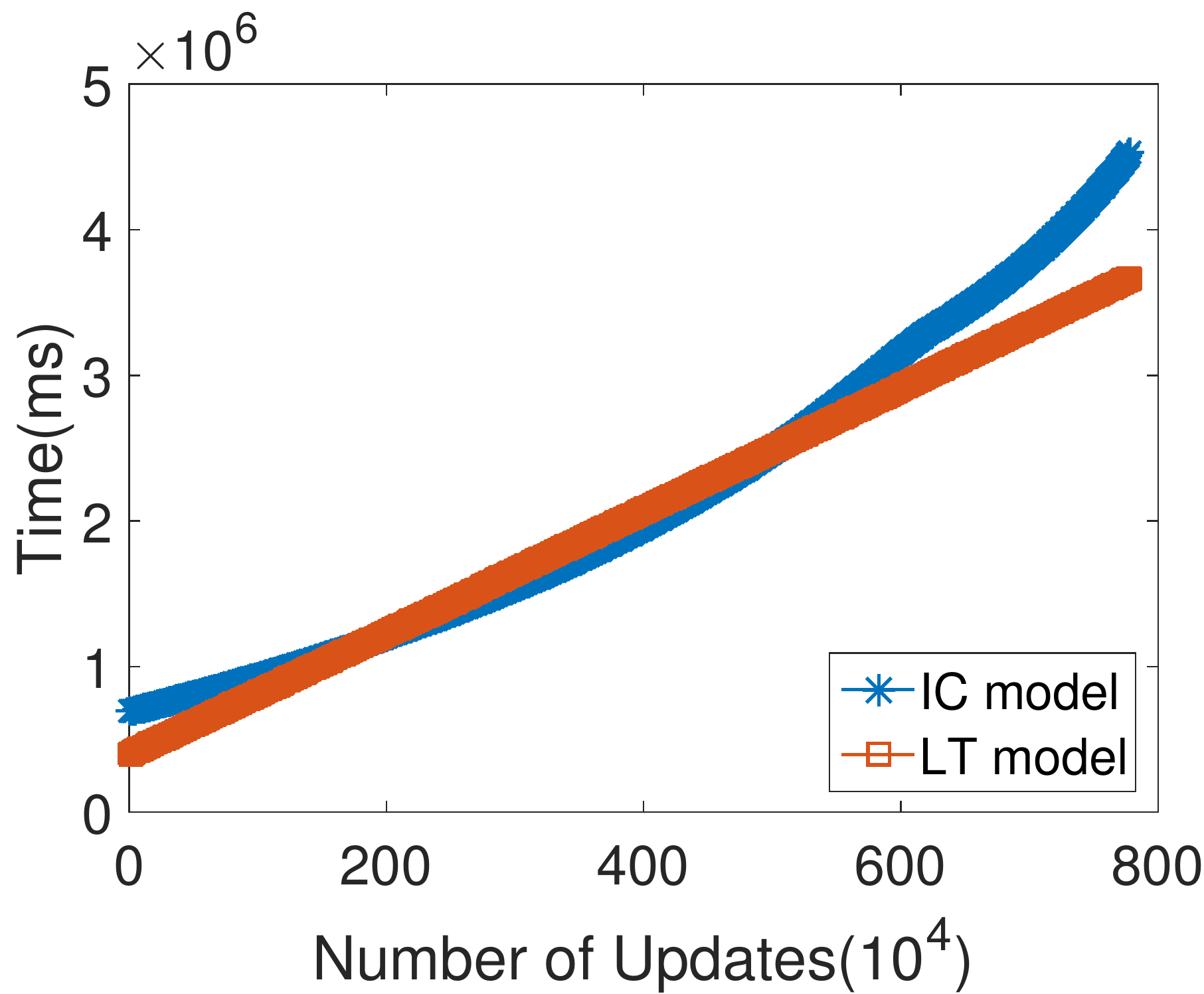}
    }
    \subfigure[flickr-growth (Top-K)]{
      \includegraphics[width=.22\textwidth]{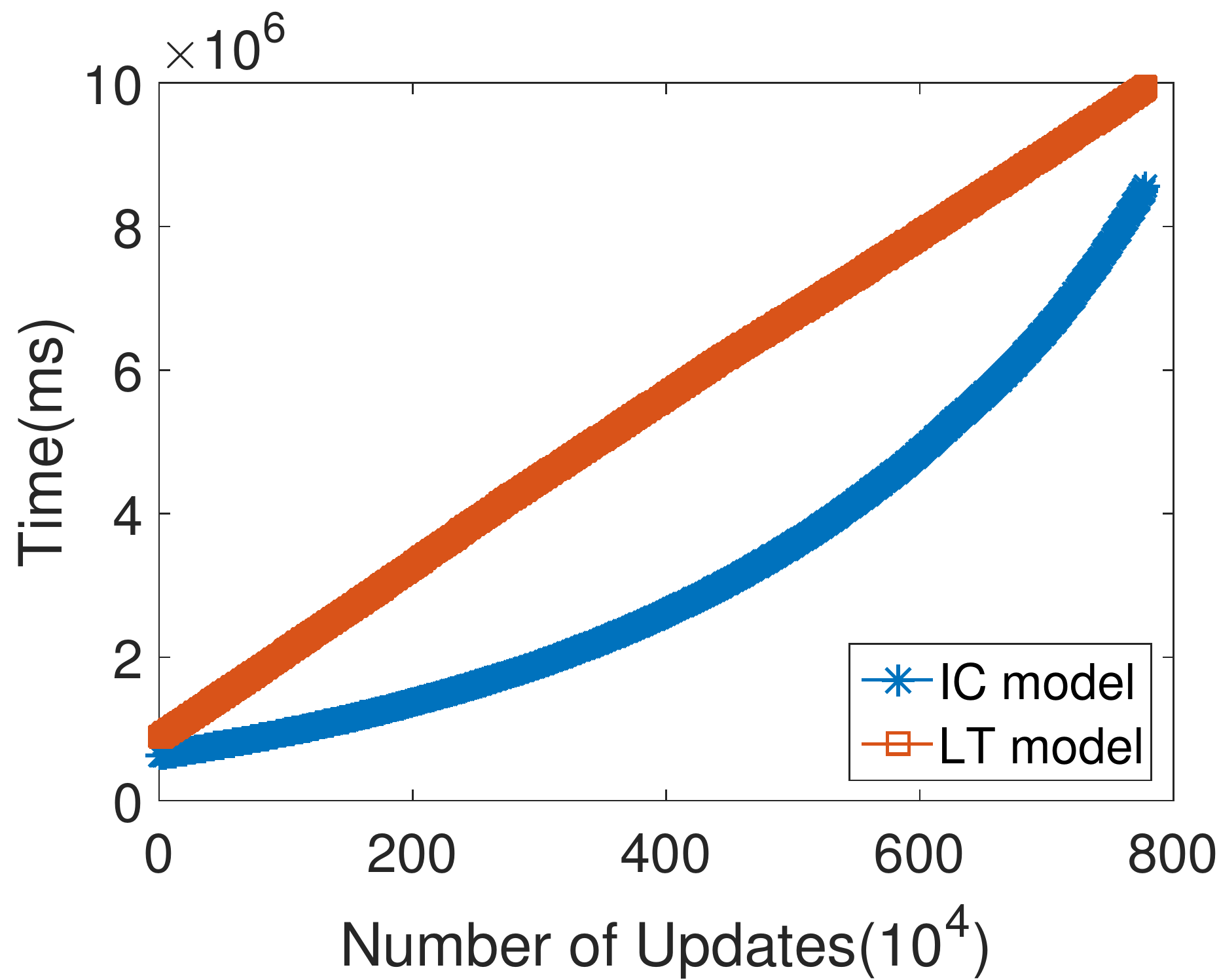}
    }
    \subfigure[Twitter (Threshold)]{
      \includegraphics[width=.22\textwidth]{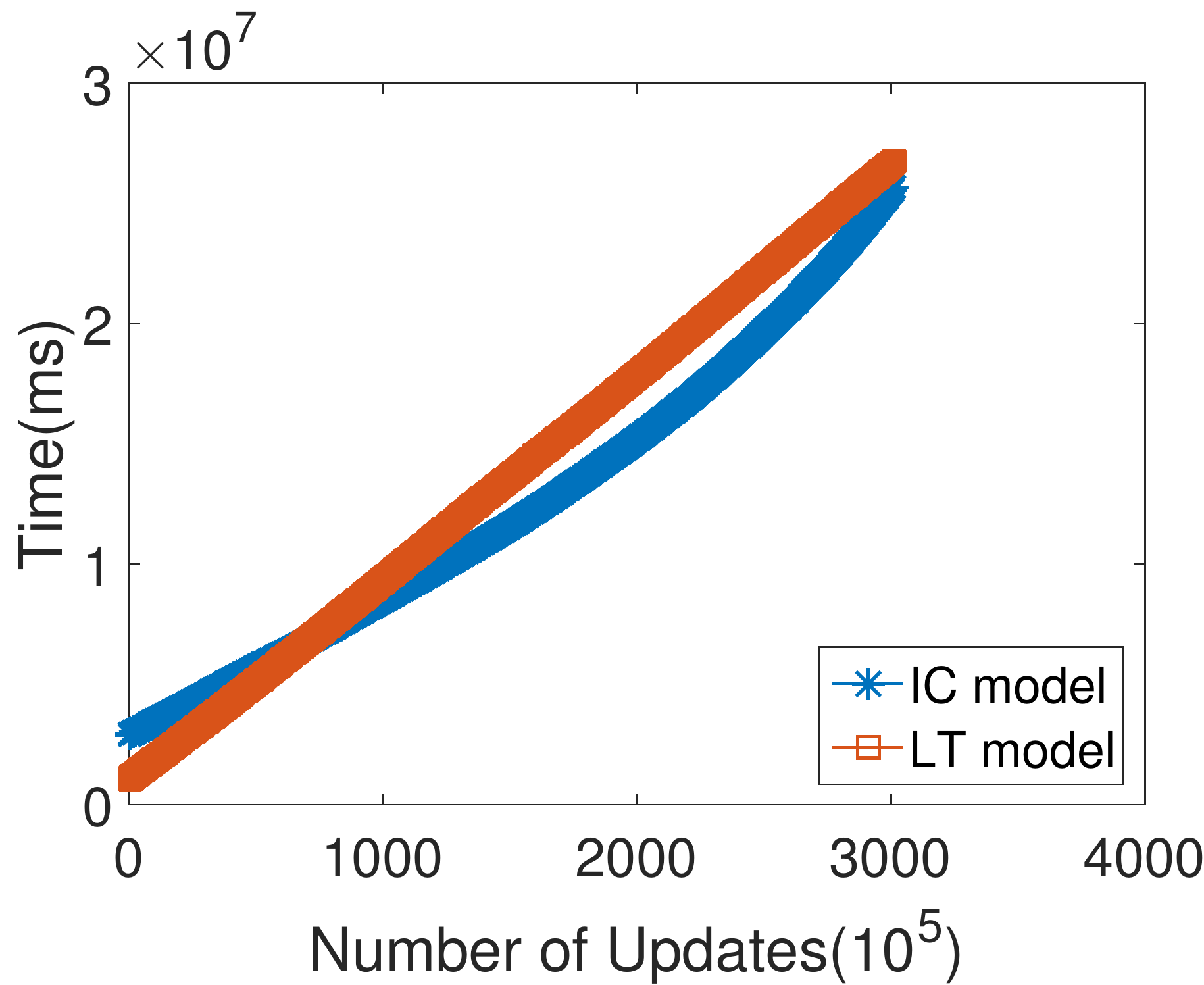}
    }
    \subfigure[Twitter (Top-K)]{
      \includegraphics[width=.22\textwidth]{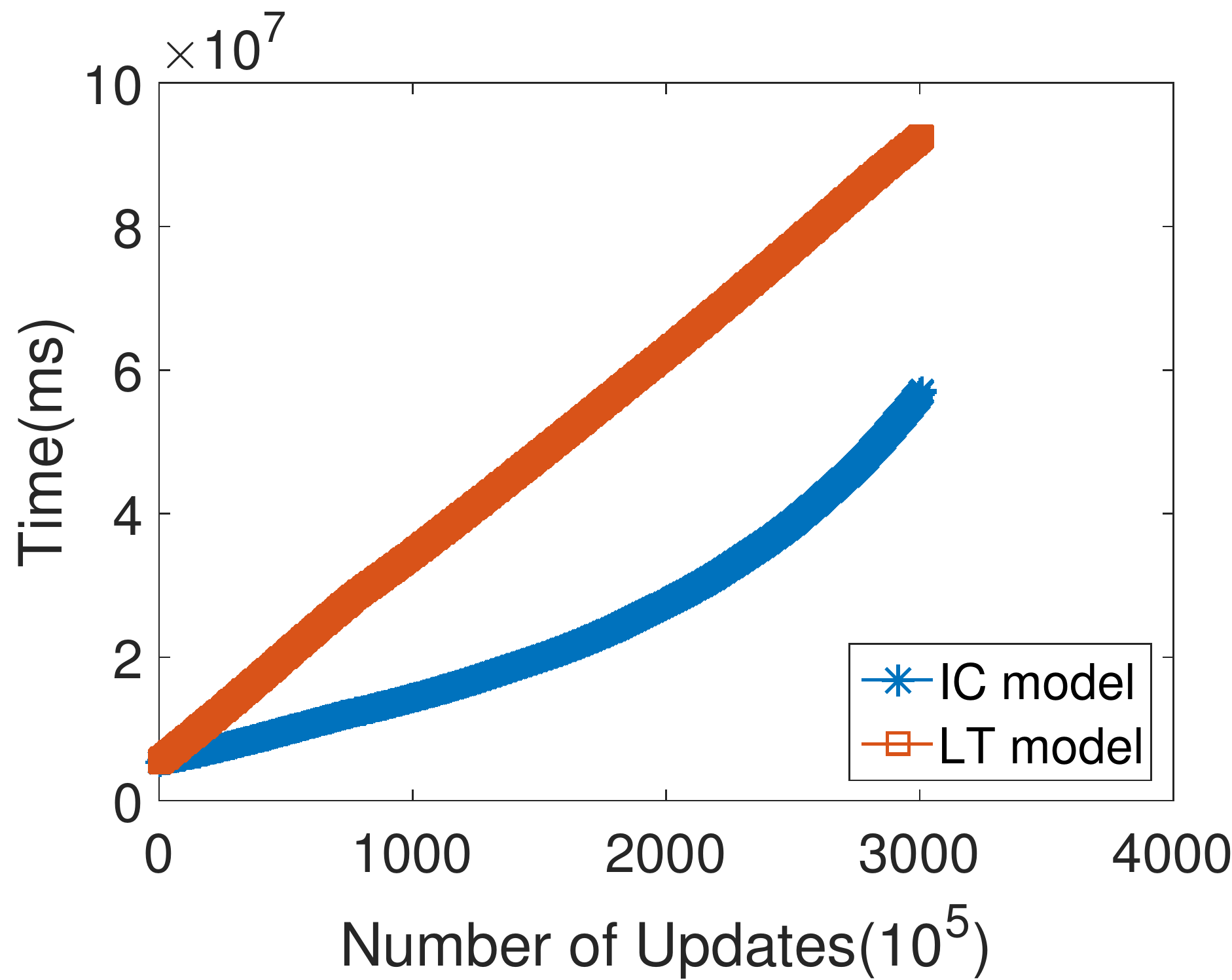}
    }
    \caption{Scalability. }
    \label{fig:efficiency}
\end{figure}

\begin{table}
\centering \caption{Running time (s) on static networks.}
\label{tab:time_static}
\begin{tabular}{*{5}{|c}|}
    \hline
    	\multirow{2}{*}{Dataset} & \multicolumn{2}{c|}{Threshold} & \multicolumn{2}{c|}{Top-K} \\ \cline{2-5}
			&	LT	&	IC	&	LT	&	IC	\\ \hline
	wiki-Vote	&  5.45 & 23.66 & 26.1 & 108.7 \\ \hline
	Flixster	& 79.5 & 380 & 174 & 950 \\ \hline
	soc-Pokec & 246 & 1095 & 2222 & 2685 \\ \hline 
	flickr-growth & 401 & 1977 & 1524 & 5355 \\ \hline
	Twitter & 974 & 8263 & 4828 & 29997 \\ \hline
\end{tabular}
\end{table}

For the LT model, our algorithm scales up roughly linearly. For the IC model, the running time increases more than linear. This is due to our experimental settings. For the LT model, the sum of propagation probabilities from all in-neighbors of a node is always 1, while in the IC model, at the beginning the sum of propagation probabilities from all in-neighbors is roughly 0.9 but becomes 1 finally. Thus, the spreads of nodes change more dramatically in the IC model than in the LT model. According to our analysis in Section~\ref{subsec:update_IC}, the cost of updating the RR sets is proportional to $I_v^{t-1}$, the influence of $v$ at time $t-1$, and $M$, the sample size. In the top-k task, $M$ is decided by $I_{max}$. So the running time curves of the IC model are not linear.

Fig.~\ref{fig:change} shows how $I_{max}$ and the sample size $M$ in the top-k task changes over time in the flickr-growth dataset. We do not report results on other datasets because they are all similar. 

\subsubsection{Memory Usage with respect to Input Size}

We also report the memory usage of our algorithm against the increase of the input graph size. Since the memory needed in Top-K influential nodes mining is usually much higher than the threshold-based mining, we only report results of the Top-K influential nodes mining algorithm. We used the second largest data set, flickr-growth network, to generate some smaller networks. Specifically, we sampled 20\%, 40\%, 60\% and 80\% nodes and extract the induced subgraphs. For each sample rate, we sampled 10 subgraphs and for each subgraph we generated a base network and an update stream as we described in Section~\ref{sec:exp_setting}. We ran the Top-K influential nodes mining algorithm on those generated data. Fig.~\ref{fig:memory} reports the average memory storing the input graph and the average peak memory usage of the RR sets against the sample rate. The results show that the size of sampled graph increases super-linearly while the memory of RR sets increases roughly linearly as the sample rate increases. Fig.~\ref{fig:memory} also shows that the average peak memory used by the RR sets increases sub-linearly as the input graph size increases.

\section{Conclusions and Future Work}\label{sec:con}

In this paper, we proposed novel, effective and efficient polling-based algorithms for tracking influential individual nodes in dynamic networks under the Linear Threshold model and the Independent Cascade model. We modeled dynamics in a network as a stream of edge weight updates. We devised an efficient incremental algorithm for updating random RR sets against network changes. For two interesting settings of influential node tracking, namely, tracking nodes with influence above a given threshold and tracking top-k influential nodes, we derived the number of random RR sets we need to approximate the exact set of influential nodes. We reported a series of experiments on $5$ real networks and demonstrated the effectiveness and efficiency of our algorithms.

There are a few interesting directions for future work. For example, can we apply similar techniques to other influence models such as the Continuous Time Diffusion Model~\cite{du2013scalable}? Since the Continuous Time Diffusion model has an implicit time constraint, how to efficiently update RR sets according to the time constraint is a critical challenge. \nop{Another challenge is the influence maximization problem on dynamic networks with provable quality guarantees.}

\bibliographystyle{abbrv}
\bibliography{DINF}

\begin{thebibliography}{10}

\bibitem{agarwal2008identifying}
N.~Agarwal et~al.
\newblock Identifying the influential bloggers in a community.
\newblock In {\em WSDM}, pages 207--218. ACM, 2008.

\bibitem{aggarwal2012influential}
C.~C. Aggarwal et~al.
\newblock On influential node discovery in dynamic social networks.
\newblock In {\em SDM}, pages 636--647. SIAM, 2012.

\bibitem{bahmani2010fast}
B.~Bahmani et~al.
\newblock Fast incremental and personalized pagerank.
\newblock {\em PVLDB}, 4(3):173--184, 2010.

\bibitem{borgs2014maximizing}
C.~Borgs et~al.
\newblock Maximizing social influence in nearly optimal time.
\newblock In {\em SODA}, pages 946--957. SIAM, 2014.

\bibitem{cha2010measuring}
M.~Cha et~al.
\newblock Measuring user influence in twitter: The million follower fallacy.
\newblock {\em ICWSM}, 10(10-17):30, 2010.

\bibitem{chen2009efficient}
W.~Chen et~al.
\newblock Efficient influence maximization in social networks.
\newblock In {\em SIGKDD}, pages 199--208. ACM, 2009.

\bibitem{chen2010scalable}
W.~Chen et~al.
\newblock Scalable influence maximization for prevalent viral marketing in
  large-scale social networks.
\newblock In {\em SIGKDD}, pages 1029--1038. ACM, 2010.

\bibitem{chen2010scalableLT}
W.~Chen et~al.
\newblock Scalable influence maximization in social networks under the linear
  threshold model.
\newblock In {\em ICDM}, pages 88--97. IEEE, 2010.

\bibitem{chen2013information}
W.~Chen et~al.
\newblock Information and influence propagation in social networks.
\newblock {\em Synthesis Lectures on Data Management}, 5(4):1--177, 2013.

\bibitem{chen2015influential}
X.~Chen et~al.
\newblock On influential nodes tracking in dynamic social networks.
\newblock In {\em SDM}, pages 613--621. SIAM, 2015.

\bibitem{chung2006concentration}
F.~Chung et~al.
\newblock Concentration inequalities and martingale inequalities: a survey.
\newblock {\em Internet Mathematics}, 3(1):79--127, 2006.

\bibitem{cohen2014sketch}
E.~Cohen et~al.
\newblock Sketch-based influence maximization and computation: Scaling up with
  guarantees.
\newblock In {\em CIKM}, pages 629--638. ACM, 2014.

\bibitem{cormode2008finding}
G.~Cormode et~al.
\newblock Finding frequent items in data streams.
\newblock {\em PVLDB}, 1(2):1530--1541, 2008.

\bibitem{domingos2001mining}
P.~Domingos et~al.
\newblock Mining the network value of customers.
\newblock In {\em SIGKDD}, pages 57--66. ACM, 2001.

\bibitem{du2013scalable}
N.~Du et~al.
\newblock Scalable influence estimation in continuous-time diffusion networks.
\newblock In {\em NIPS}, pages 3147--3155, 2013.

\bibitem{goyal2008discovering}
A.~Goyal et~al.
\newblock Discovering leaders from community actions.
\newblock In {\em CIKM}, pages 499--508. ACM, 2008.

\bibitem{goyal2010learning}
A.~Goyal et~al.
\newblock Learning influence probabilities in social networks.
\newblock In {\em WSDM}, pages 241--250. ACM, 2010.

\bibitem{goyal2011simpath}
A.~Goyal et~al.
\newblock Simpath: An efficient algorithm for influence maximization under the
  linear threshold model.
\newblock In {\em ICDM}, pages 211--220. IEEE, 2011.

\bibitem{hayashi2015fully}
T.~Hayashi et~al.
\newblock Fully dynamic betweenness centrality maintenance on massive networks.
\newblock {\em PVLDB}, 9(2):48--59, 2015.

\bibitem{kempe2003maximizing}
D.~Kempe et~al.
\newblock Maximizing the spread of influence through a social network.
\newblock In {\em SIGKDD}, pages 137--146. ACM, 2003.

\bibitem{kimura2006tractable}
M.~Kimura et~al.
\newblock Tractable models for information diffusion in social networks.
\newblock In {\em PKDD}, pages 259--271. Springer, 2006.

\bibitem{lei2015online}
S.~Lei et~al.
\newblock Online influence maximization.
\newblock In {\em SIGKDD}, pages 645--654. ACM, 2015.

\bibitem{leskovec2005graphs}
J.~Leskovec et~al.
\newblock Graphs over time: densification laws, shrinking diameters and
  possible explanations.
\newblock In {\em SIGKDD}, pages 177--187. ACM, 2005.

\bibitem{leskovec2008microscopic}
J.~Leskovec et~al.
\newblock Microscopic evolution of social networks.
\newblock In {\em SIGKDD}, pages 462--470. ACM, 2008.

\bibitem{lucier2015influence}
B.~Lucier et~al.
\newblock Influence at scale: Distributed computation of complex contagion in
  networks.
\newblock In {\em SIGKDD}, pages 735--744. ACM, 2015.

\bibitem{ohsaka2015efficient}
N.~Ohsaka et~al.
\newblock Efficient pagerank tracking in evolving networks.
\newblock In {\em SIGKDD}, pages 875--884. ACM, 2015.

\bibitem{ohsaka2016dynamic}
N.~Ohsaka et~al.
\newblock Dynamic influence analysis in evolving networks.
\newblock {\em PVLDB}, 9(12):1077--1088, 2016.

\bibitem{pietracaprina2010mining}
A.~Pietracaprina et~al.
\newblock Mining top-k frequent itemsets through progressive sampling.
\newblock {\em Data Mining and Knowledge Discovery}, 21(2):310--326, 2010.

\bibitem{riondato2014efficient}
M.~Riondato et~al.
\newblock Efficient discovery of association rules and frequent itemsets
  through sampling with tight performance guarantees.
\newblock {\em ACM Transactions on Knowledge Discovery from Data}, 8(4):20,
  2014.

\bibitem{riondato2015mining}
M.~Riondato et~al.
\newblock Mining frequent itemsets through progressive sampling with rademacher
  averages.
\newblock In {\em SIGKDD}, pages 1005--1014. ACM, 2015.

\bibitem{rossi2015spread}
M.-E.~G. Rossi and otehrs.
\newblock Spread it good, spread it fast: Identification of influential nodes
  in social networks.
\newblock In {\em WWW}, pages 101--102. ACM, 2015.

\bibitem{saito2008prediction}
K.~Saito et~al.
\newblock Prediction of information diffusion probabilities for independent
  cascade model.
\newblock In {\em Knowledge-based intelligent information and engineering
  systems}, pages 67--75. Springer, 2008.

\bibitem{shogenji2003condition}
T.~Shogenji.
\newblock A condition for transitivity in probabilistic support.
\newblock {\em The British Journal for the Philosophy of Science},
  54(4):613--616, 2003.

\bibitem{tang2014influence}
Y.~Tang et~al.
\newblock Influence maximization: Near-optimal time complexity meets practical
  efficiency.
\newblock In {\em SIGMOD}, pages 75--86. ACM, 2014.

\bibitem{tang2015influence}
Y.~Tang et~al.
\newblock Influence maximization in near-linear time: A martingale approach.
\newblock In {\em SIGMOD}. ACM, 2015.

\bibitem{vitter1985random}
J.~S. Vitter.
\newblock Random sampling with a reservoir.
\newblock {\em ACM Transactions on Mathematical Software (TOMS)}, 11(1):37--57,
  1985.

\bibitem{weng2010twitterrank}
J.~Weng et~al.
\newblock Twitterrank: finding topic-sensitive influential twitterers.
\newblock In {\em WSDM}, pages 261--270. ACM, 2010.

\end{thebibliography}

\end{sloppy}
\end{document}